\documentclass[12pt]{article}
\usepackage{times,amssymb,amsmath,amsfonts,eurosym,geometry,ulem,graphicx,caption,color,setspace,sectsty,comment,footmisc,pdflscape,array,threeparttable,tikz,subcaption,natbib,tabularx,float} 
\newsavebox{\measurebox}
\usepackage[scale=2]{ccicons}
\usepackage{pgfplots}
\usepgfplotslibrary{dateplot}
\usetikzlibrary{intersections,shapes.multipart}
\usepackage{pifont}
\usetikzlibrary{tikzmark}
\usetikzlibrary{shapes,arrows,backgrounds,positioning}
\usetikzlibrary{intersections,shapes.multipart}
\usepackage{tikz}
\usetikzlibrary{fit,calc}

\usetikzlibrary{intersections,shapes.multipart}
\usepackage{pifont}
\usetikzlibrary{tikzmark}
\usetikzlibrary{shapes,arrows,backgrounds,positioning}
\usetikzlibrary{intersections,shapes.multipart}
\usepackage{pifont}
\usepackage{tikz}
\usetikzlibrary{fit,calc}

\usepackage{booktabs} 

\usepackage[colorlinks=true,
citecolor=blue,
linkcolor=blue,
anchorcolor=blue,
urlcolor=blue]{hyperref}

\usepackage{multicol}
\def\E{\mathbb{E}\,}

\usepackage{tikz,pgfplots}
\usepackage{amsmath,amsthm}
\newtheorem{assumption}{Assumption}
\newtheorem{lemma}{Lemma}
\newtheorem{corollary}{Corollary}
\newtheorem{proposition}{Proposition}
\newtheorem{theorem}{Theorem}
\theoremstyle{definition}
\newtheorem{definition}{Definition}

\theoremstyle{definition}
\newtheorem{example}{Example}

\newtheorem{remark}{Remark}

\usepackage[titletoc,toc,title,page,header]{appendix}
\usepackage{minitoc}

\noptcrule
\usepackage{caption}

\newcommand{\Var}{\operatorname{Var}}
\newcommand{\Cov}{\operatorname{Cov}}
\newcommand{\MSE}{\operatorname{MSE}}

\newcommand{\R}{\mathbb{R}}
\newcommand{\1}{\mathbf{1}}
\newcommand{\DM}{\mathrm{DM}}

\newcommand{\lin}{\mathrm{Lin}}

\newcommand{\calZ}{\mathcal Z}
\newcommand{\DeltaZ}{\Delta_Z}

\newcommand{\trans}{\mathsf T}

\newcommand{\argmin}{\operatorname*{arg\,min}}

\title{Covariate Adjustment for Randomized Experiments: A Unified Framework for Decision and Practice}
\author{Jiawei Fu\footnote{Assistant Professor, Duke University. \url{jiawei.fu@duke.edu}} \and  Donald P. Green\footnote{Burgess Professor of Political Science, Columbia University. \url{dpg2110@columbia.edu} \\  
We thank Winston Lin, Gaurav Sood, and Anna Wilke for detailed comments and suggestions that helped us reinterpret several results. We also thank participants at the 2026 APSA Annual Meeting for their helpful comments and feedback.}}

\date{\today}

\begin{document}

\maketitle
\singlespacing


\begin{abstract}
Should researchers adjust for covariates when analyzing randomized experiments, and if so, how? The literature offers three distinct prescriptions: avoid adjustment because randomization guarantees unbiasedness; adjust for outcome-prognostic covariates to improve precision; or adjust for covariates that are imbalanced between treatment arms. These competing prescriptions create confusion about when and how covariate adjustment is warranted. We develop a unified decision-making framework centered on \emph{ex-post bias}. The framework reveals that conventional criteria based on outcome-mean prognosticity and raw covariate imbalance are neither necessary nor sufficient for determining whether covariate adjustment improves ATE estimation. This finding motivates the introduction of ATE-relevant prognosticity and imbalance, with the optimal adjustment targeting the unobserved ex-post bias. We investigate two approaches to estimating this bias and demonstrate that many existing covariate-adjustment estimators are special cases of these two approaches, thereby providing a unified theoretical justification of existing methods. Simulation studies illustrate the theoretical findings and evaluate alternative strategies for covariate selection and adjustment. Overall, our framework reconciles competing perspectives on covariate adjustment and provides practical guidance for analyzing randomized experiments.

\end{abstract}

\noindent 
\vspace{.1in}

\vspace{.1in}

\thispagestyle{empty}

\clearpage
\doparttoc 
\faketableofcontents 

\setcounter{page}{1}

\section{Introduction}

Suppose that a randomized experiment has been completed. When analyzing the resulting data, should the researcher adjust for pretreatment covariates? If so, which covariates should be used, and how should the adjustment be carried out? These questions appear routine, but the literature offers conflicting advice.

Three distinct schools of thought predominate. One approach is to make no adjustment and report the unadjusted difference in means, because properly implemented randomization renders this estimator unbiased, whereas covariate adjustment opens the door to bias.  
Covariate adjustment risks bias when $N$ is small \citep{freedman2008regression} or when the covariates are measured post-treatment \citep{montgomery2018conditioning}.  Adjustment also creates opportunities for specification searches and selective reporting of results \citep{humphreys2013fishing,meier1993illusion}. A second approach is to adjust when covariates are imbalanced, on the grounds that the two groups may not be comparable given the realized assignment \citep{johansson2022inference}. A third approach is to adjust for covariates that predict outcomes because prognostic adjustment can reduce residual variation and increase precision \citep{zhang2008improving}. These distinct rationales lead to diverse practices. Researchers sometimes report the unadjusted estimator and treat it as the gold-standard estimate. In other cases, they use a balance table that displays covariate means by treatment arm in order to identify mean-imbalanced covariates, which they adjust for in order to correct what they take to be a bad randomization. In still other cases, methods such as LASSO are used to select covariates that are predictive of outcomes, which are then used to adjust the estimated average treatment effect. 

These rationales appear in published experimental studies across disciplines. Table~\ref{tab:covariate-adjustment} presents one example from each of economics, political science, and public health for each rationale. The quotations illustrate the authors' stated reasons for their covariate-adjustment decisions. These examples illustrate the breadth of the adjustment question across disciplines. But are these rationales well founded? Which adjustment method should researchers use in practice? Existing debates often talk past one another. A clear, unified answer is therefore needed. This paper develops such a framework.

\begin{table}[p]
\centering
\begin{threeparttable}
\caption{Illustrative rationales for covariate adjustment across disciplines}
\label{tab:covariate-adjustment}
\small
\setstretch{1}
\setlength{\tabcolsep}{5pt}
\renewcommand{\arraystretch}{1.12}
\begin{tabularx}{\linewidth}{@{}>{\raggedright\arraybackslash}p{0.13\linewidth}>{\raggedright\arraybackslash}p{0.28\linewidth}>{\raggedright\arraybackslash}X@{}}
\toprule
Discipline & Study and source & Quotation \\
\midrule
\multicolumn{3}{@{}l}{\textit{R1: Randomization makes adjustment unnecessary}} \\
\addlinespace[3pt]
Economics & \citet{ericson2025reducing}\newline \textit{REStat}, p.~10 & ``controls are not needed to address bias given the randomized design'' \\
\addlinespace[5pt]
Political science & \citet{hjortskov2018encouraging}\newline \textit{AJPS}, p.~604 & ``Since the treatment is randomly assigned, it should not be strictly necessary to control for other variables.'' \\
\addlinespace[5pt]
Public health & \citet{susukida2017generalizability}\newline \textit{Addiction}, Statistical Analysis & ``Assuming that randomization was successful in each trial, we did not adjust for baseline variables within the trial samples.'' \\
\midrule
\multicolumn{3}{@{}l}{\textit{R2: Adjust to correct baseline imbalance}} \\
\addlinespace[3pt]
Economics & \citet{angelucci2024economic}\textit{AER}, p.~181, n.~27 & ``to correct for the baseline imbalance in PHQ-9 scores'' \\
\addlinespace[5pt]
Political science & \citet{balcells2022transitional}\newline \textit{JOP}, p.~502 & ``For imbalances (likely arising due to our small sample size), we include control variables in regression estimates.'' \\
\addlinespace[5pt]
Public health & \citet{boyce2025mediating}\newline \textit{AJE}, p.~1223 & ``sociodemographic factors found to be imbalanced across arms at baseline'' \\
\midrule
\multicolumn{3}{@{}l}{\textit{R3: Adjust for outcome predictors to improve precision}} \\
\addlinespace[3pt]
Economics & \citet{dupas2025negligible}\newline \textit{AER}, p.~2670 & ``can reduce the variance of the error term and thus improve the precision of the estimates'' \\
\addlinespace[5pt]
Political science & \citet{zelizer2022talking}\newline \textit{AJPS}, p.~910 & ``In order to improve the precision with which the parameter $b_1$ is estimated, Equation (1) includes pretreatment covariates that predict bill support.'' \\
\addlinespace[5pt]
Public health & \citet{watson2023cluster}\newline \textit{BMJ Global Health}, p.~5 & ``To increase precision, adjusted rate (or count) ratios were computed, adjusting for factors determined a priori to be associated with the outcome'' \\
\bottomrule
\end{tabularx}
\begin{tablenotes}[flushleft]
\footnotesize
\item \textit{Notes:} Selected excerpts illustrate stated rationales, which may overlap within a study.
\item \textit{Journal abbreviations:} AER: \textit{American Economic Review}; AJE: \textit{American Journal of Epidemiology}; AJPS: \textit{American Journal of Political Science}; JOP: \textit{Journal of Politics}; ; REStat: \textit{Review of Economics and Statistics}.
\end{tablenotes}
\end{threeparttable}
\end{table}


Our decision framework begins with a practical concern. Although the difference-in-means estimator $\hat \tau$ is ex-ante unbiased for the average treatment effect (ATE), researchers usually observe only a single experiment, and its realized estimate will generally differ from the true effect $\tau$ because the treatment and control groups are not perfectly balanced on observed or unobserved covariates. 
In other words, a gap generally remains between the true effect and the estimate from a realized experiment. Can available information reveal this \textit{ex-post} randomization error $e=\hat \tau-\tau$ and support a correction that moves the estimate closer to the true effect? We call the error ex-post because it is considered after the randomization has been conducted, when researchers may have learned information $\mathcal{I}$ about the randomization. For example, researchers may learn that a covariate mean is larger in the treatment group than the control group and therefore become concerned that the realized estimate deviates from the unobserved true effect.

We formalize and solve the decision problem. We show that the optimal correction is exactly the \textit{ex-post bias} $\E_Z[e|\mathcal I]$. Estimating and subtracting this term can remove the ex-post bias. This framework unifies two objectives that are often treated separately. From an ex-post perspective, adjustment removes conditional predictable error in the particular experiment that occurred. From an ex-ante perspective, the same adjustment reduces variation over the randomization distribution and increases precision. Therefore, concerns about covariate imbalance and precision are not competing justifications for covariate adjustment; rather, they are two manifestations of the same underlying decision problem. 

Building on this framework, we introduce the concepts of \textit{ATE-relevant prognosticity} and \textit{ATE-relevant imbalance}. Conventional notions of prognosticity and imbalance are frequently invoked in the literature to motivate covariate adjustment. We show, however, that these notions do not necessarily capture the properties of covariates that are relevant for ATE estimation. 

A covariate is traditionally considered prognostic if it predicts the outcome mean or improves the $R^2$ of a linear regression. For ATE estimation, however, the relevant question is more specific: does the covariate predict the ex-post randomization error? We characterize the conditions under which conventional prognosticity translates into ATE-relevant prognosticity and provide several examples to illustrate their relationship. Similarly, imbalance typically refers to raw covariate imbalance, which describes differences in observed baseline characteristics between treatment groups. Such differences are useful for assessing the implementation of randomization and diagnosing potential problems in experimental implementation, including those arising from attrition. However, raw covariate imbalance does not necessarily imply ex-post bias in ATE estimation and, therefore, does not by itself justify covariate adjustment.


Because ex-post bias is unobserved in practice, we introduce two approaches to estimating it. We show that many existing covariate-adjustment estimators can be understood from the vantage point of these two approaches, thereby providing a unified theoretical justification for these methods. The first approach estimates the relevant components separately within the treatment and control arms and then combines them to construct an estimator. This approach encompasses familiar covariate-adjustment methods, including conventional regression adjustment, generalized regression estimators, augmented estimators, and flexible machine-learning procedures \citep{lin2013agnostic, robins1994estimation}. The second approach estimates ex-post bias directly \citep{zhang2019machine}. Although less commonly used, this approach performs competitively in our finite-sample simulations. One reason is that the joint approach requires estimating only a single nuisance function, allowing it to use the available data more efficiently, potentially reducing estimation noise relative to arm-specific approaches that estimate separate nuisance functions for each treatment arm. We also show that, under balanced allocation, the observed outcome provides an unbiased proxy for the unit randomization error; furthermore, under the sharp null hypothesis, ex-post bias is directly observed. 

Following the theoretical analysis, we illustrate the framework through a series of simulations. We compare a range of estimators across four settings: no useful information about the ATE estimation error; a sparse systematic component that depends jointly on four of the sixty covariates; strong within-arm outcome relationships that cancel out for ATE adjustment; and unequal treatment allocation. We also vary the functional relationship between covariates and ex-post bias (linear or nonlinear), as well as the sample size, to examine finite-sample estimation performance. Overall, different adjustment methods offer distinct advantages across scenarios. In general, the joint estimator exhibits relatively strong performance across scenarios. The arm-specific AIPW estimator remains a competitive alternative. 

Our unified framework provides clear guidance on both \emph{when} and \emph{how} researchers should adjust for covariates. Researchers should consider covariate adjustment when the covariates are prognostic for ex-post bias. The relevant criterion is not conventional outcome-mean prognosticity or raw covariate imbalance, but rather the ATE-relevant prognosticity and imbalance characterized in our framework. Under equal treatment allocation, ATE-relevant prognosticity coincides with outcome-mean prognosticity. In practice, researchers should collect a broad set of pretreatment covariates and use either of the two approaches to construct covariate-adjusted estimators with flexible machine-learning methods. These methods can leverage the available covariates to estimate the component of ex-post bias predictable from observed covariates, thereby reducing the need for researchers to select covariates manually in advance. Researchers should not decide whether to adjust based on tests of whether individual raw covariates happen to be balanced in the realized assignment or solely on whether those covariates predict control and/or treatment group outcomes. At the same time, we do not recommend abandoning conventional balance tables. Raw covariate balance remains useful for describing the experimental subjects in each arm at the point of random assignment. 

The remainder of the paper proceeds as follows. We first review competing arguments concerning covariate balance and adjustment. We then develop our decision framework and derive feasible prognostic-adjustment procedures. Next, we clarify the concepts of prognosticity and imbalance and their relevance to ATE estimation. Finally, we illustrate the framework through simulation studies and an empirical application before concluding.

\section{Debates over Covariate Balance and Adjustment}
\label{sec:debates}



This section reviews three rationales for covariate adjustment and then explains how the decision framework developed below places them within a common structure while highlighting what each perspective overlooks.

\paragraph{Ex-Ante Unbiasedness.}
One influential view begins from a simple principle: in a properly randomized experiment, the difference-in-means estimator is unbiased for the average treatment effect in expectation over all possible randomizations. Covariate adjustment is therefore not required for identification. This feature makes the unadjusted estimator transparent and provides a natural benchmark against which adjusted estimators should be evaluated.

Concerns about conventional regression adjustment reinforce this position. \citet{freedman2008regression,freedman2008regression1} show that ordinary least squares can have finite-sample bias when applied to experimental data and that adjustment need not always improve precision. These results caution against assuming that data from a randomized experiment were generated by a correctly specified linear regression model. Subsequent work, however, substantially qualifies these concerns. With treatment--covariate interactions, regression adjustment is asymptotically unbiased and no less efficient than the unadjusted estimator under standard regularity conditions \citep{lin2013agnostic}. Exact or higher-order bias corrections can further address the finite-sample bias emphasized by Freedman \citep{wu2018loop,chang2024exact}.

\paragraph{Covariate Imbalance.}
A second perspective agrees that randomization eliminates systematic bias over the assignment distribution. However, because an experiment is conducted only once, a particular random assignment may place units with more (or less) favorable prognostic traits in one treatment arm. Conditional on the realized signed covariate imbalance, the difference-in-means estimator can therefore deviate substantially from the true average treatment effect, potentially leading researchers and practitioners to draw misleading conclusions.

This concern has a long history. Early discussions of clinical-trial analysis describe covariate adjustment as a way to account for baseline variables on which randomization failed to achieve close balance \citep{beach1989choosing,pocock2002subgroup}. In practice, balance reporting is therefore common in experimental research and is encouraged by reporting standards in political science and other fields \citep{gerber2014reporting}. Randomization-based omnibus tests have also been developed to summarize whether the observed treatment--control differences are unusual under the specified assignment mechanism \citep{hansen2008covariate}. 

A further step is to choose covariates for adjustment on the basis of a balance test. This practice is much more controversial. \citet{senn1994testing} argues that significance tests of baseline equality are uninformative and potentially misleading in randomized trials. Moreover, adjustment after a balance test can have undesirable efficiency and inferential properties \citep{mutz2019perils,zhao2024randomization}.

\paragraph{Prognostic Adjustment and Efficiency.}

A third school of thought justifies adjustment on efficiency grounds. The key principle is prognosticity: baseline information is useful when it predicts outcomes. In the context of linear regression, pre-treatment covariates that increase the adjusted $R^2$ decrease residual variance. The net effect on the variance of the estimated average treatment effect reflects the trade-off between increased variance due to incidental correlation between the covariates and the treatment indicator (which diminishes asymptotically) and decreased variance due to a reduction in residual variance.

This principle underlies classical analysis of covariance and modern semiparametric augmentation methods
\citep{tsiatis2008covariate,zhang2008improving}. Fully interacted regression adjustment is asymptotically no less efficient than the difference in means \citep{lin2013agnostic,negi2021revisiting}. Related work shows that substantial gains are possible when baseline variables are strongly predictive of outcomes \citep{colantuoni2015leveraging}. High-dimensional methods extend this logic by using regularization, cross-estimation, or machine learning to estimate prognostic components while controlling the cost of overfitting \citep{bloniarz2016lasso,wager2016high,wu2018loop}. Current FDA guidance likewise focuses on prespecified prognostic baseline covariates as a means of increasing precision and power.\footnote{https://www.fda.gov/regulatory-information/search-fda-guidance-documents/adjusting-covariates-randomized-clinical-trials-drugs-and-biological-products}

\paragraph{Summary.}

As illustrated in Table~\ref{tab:covariate-adjustment}, the debate remains unresolved, with researchers adopting different rationales for covariate adjustment and employing different adjustment strategies. In our view, each school of thought offers valid but incomplete insights. The ex-ante unbiasedness school correctly observes that randomization eliminates systematic confounding in expectation, but it overlooks efficiency and the fact that each realized experiment has its own error. The imbalance school correctly recognizes that covariate imbalance can matter in a single realized experiment. In practice, however, balance-test-and-adjust procedures may reduce the precision with which the ATE is estimated. The efficiency school emphasizes prognosticity and rests on a strong theoretical foundation, but it gives less attention to bias associated with covariate imbalance, which is a principal concern motivating balance tables in experimental research.

The framework developed in the next section connects these practical concerns and perspectives by formalizing when and how a realized estimate should be corrected. Our central result is that the relevant criterion for covariate adjustment is prognosticity for ex-post bias. Neither raw covariate imbalance nor conventional outcome prognosticity is, by itself, necessary or sufficient to determine whether covariate adjustment is warranted. Instead, we introduce the concepts of ATE-relevant prognosticity and imbalance, which provide a unified perspective for understanding the debate and guiding covariate-adjustment decisions.

\section{A Decision Framework with Ex-Post Bias}

Consider a finite population of $N$ units. Unit $i$ has fixed potential outcomes $Y_i(1)$ and $Y_i(0)$ under treatment and control, respectively, and a $K$-dimensional vector of pretreatment covariates
\(
X_i=(X_{i1},\ldots,X_{iK})^{\mathsf T}.
\) Let $Z_i\in\{0,1\}$ denote treatment assignment. We assume complete randomization with exactly $N_1$ treated units and $N_0=N-N_1$ control units, so $Z$ is drawn uniformly from $\calZ_{N_1}=\{z\in\{0,1\}^N:\sum_i z_i=N_1\}$. The observed outcome is $Y_i=Z_iY_i(1)+(1-Z_i)Y_i(0)$, and the finite-population ATE is $\tau=N^{-1}\sum_{i=1}^N\{Y_i(1)-Y_i(0)\}$.

The unadjusted difference-in-means estimator is $$\widehat\tau_{\DM}=\hat{\tau}=\frac{1}{N_1} \sum_{i=1}^N Z_iY_i-\frac{1}{N_0} \sum_{i=1}^N (1-Z_i)Y_i.$$ Complete randomization makes this estimator unbiased over all possible assignments. Its value under a particular assignment, however, generally differs from the finite-population ATE. As \citet{deaton2018understanding} emphasize, errors may cancel across a hypothetical sequence of randomizations even though a single estimate from a realized experiment may be far from the target. This observation motivates a distinction between the estimator's ex-ante bias and its realized randomization error.

We first formalize the assignment-specific discrepancy between an estimator and the true ATE.

\begin{definition}[Realized randomization error]
For a realized assignment $Z=z$, the realized randomization error of an estimator $\widehat\tau$ is
\begin{equation}
e(z)=\widehat\tau(z)-\tau.
\end{equation}
\end{definition}

Each possible treatment assignment produces an estimate $\hat{\tau}$ and thus an error $e(z)$. For example, as shown in Table~\ref{tab:rand}, the first randomization produces an estimation error of 7, whereas the second produces an error of $-8$, reflecting differences in covariate balance across the two treatment assignments. In practice, researchers cannot observe this oracle randomization table or know the realized error. However, because of randomization, the expected error over all possible assignments is zero ex ante: $\E_Z\{e(Z)\}=0$. This expectation obviously does not imply that $e(z)=0$ for every $z\in\calZ_{N_1}$. 

\begin{table}[h!]
    \centering
     \caption{An artificial randomization table known to the oracle}\label{tab:rand}
    \begin{tabular}{|c|c|c|}
        \hline
        Randomization & Imbalanced covariates$^1$ & Error $e(z$) \\ \hline
        1 & $X^2$, $X^5$, $X^{10}$ & 7 \\ \hline
        2 & $X^1$, $X^4$, $X^5$, $X^7$ & -8 \\ \hline
        3 & $X^2$, $X^{10}$  &  -4 \\ \hline
       \vdots &  &  \\ \hline
       ${N \choose N_1}$  & $X^1$, $X^2$, $X^6$, $X^9$, $X^{10}$  & 5 \\ \hline
    \end{tabular}   
    \begin{minipage}{0.96\textwidth}
\footnotesize
\emph{Note:} 1. Imbalanced covariates refer to covariates for which the difference in means between treatment groups is nonzero.
\end{minipage}
\end{table} 

Because the realized randomization error is unknown, researchers may use the distribution of observed covariates to learn about it. For motivation, suppose researchers learn that $X_2$ is imbalanced in the realized assignment. This observation rules out some randomizations, including the second row of Table~\ref{tab:rand}, but remains compatible with many others, such as rows 1, 3, and the final row. We use $\mathcal I$ to denote the information available for covariate adjustment. This information may include observed differences in covariate means, the Mahalanobis balance statistic, or two-sided balance-test $p$-values. Formally, let $T(Z,X)$ denote any post-assignment statistic constructed from the realized treatment assignment and pretreatment covariates, and define $\mathcal I = \sigma\{T(Z,X)\}$ as the $\sigma$-algebra generated by this statistic. 

However, conditioning on such summary statistics may discard relevant information. A richer information set can instead be constructed from the realized covariate distributions in the treatment and control arms, as follows. For every measurable set $A\subseteq\mathcal X$, define the signed measure
\begin{equation}
\nu_Z(A)
=\frac1{N_1}\sum_iZ_i\1\{X_i\in A\}
-\frac1{N_0}\sum_i(1-Z_i)\1\{X_i\in A\}.
\label{eq:section3-signed-empirical-measure}
\end{equation}

\noindent The induced $\sigma$-algebra is then defined as
\(
\mathcal F_X
=\sigma\{\nu_Z(A):A\text{ is measurable}\}.
\) Since the covariates and treatment-arm sizes are fixed, $\nu_Z$ fully determines the empirical covariate distribution within each treatment arm. Consequently, $\mathcal F_X$ captures all realized differences in the distributions of \textit{observed} covariates between the treatment and control groups, including imbalances in nonlinear functions of $X$. By construction, $\mathcal I \subseteq \mathcal F_X$, provided that $T(Z,X)$ depends on the assignment only through these empirical covariate distributions.

We define ex-post bias as the expected estimation error conditional on the available information. Here, we focus on the information captured by $\mathcal F_X$, although the framework allows researchers to condition on alternative information sets depending on their research objectives.

\begin{definition}[Ex-post bias]\label{def:expost}
The ex-post bias associated with the observed covariate allocation is
\[
B_{\mathcal F_X}:=\E_Z[\hat\tau-\tau\mid \mathcal F_X].
\]
\end{definition}
\noindent Thus, ex-post bias is the average randomization error across assignments that yield the same realized covariate distributions in the treatment and control arms. Importantly, conditioning on $\mathcal F_X$ does not necessarily eliminate all randomization uncertainty. Even when the empirical distributions of observed covariates are identical across assignments, the assignments may still differ in their distributions of unobserved or unmeasured characteristics, leaving residual variation in the randomization error.

\subsection{Decision Problem}
Our decision framework begins with a practical concern. Given the empirical covariate distribution, the researcher must decide whether to correct the unadjusted estimator and, if so, by how much. Any adjustment estimator $\widehat\tau(Y,Z,X)$ can be written
algebraically as
\[
\widehat\tau(Y,Z,X)
=
\widehat\tau_{\DM}(Y,Z)-a(Y,Z,X),
\]
where
\(a(Y,Z,X)
=
\widehat\tau_{\DM}(Y,Z)-\widehat\tau(Y,Z,X).
\) Thus, the additive representation itself imposes no restrictions. When adjustment $a(Y,Z,X)$ is allowed to depend unrestrictedly on the observed outcomes $Y$, the problem is that its randomization expectation generally depends on the subjects' unknown potential outcomes. As we show in Proposition \ref{prop:exact-design-calibration}, this expectation cannot, in general, be recovered exactly from a single realized experiment. 

In this paper, we focus on adjustments that depend on the realized treatment assignment and pretreatment covariates, consistent with common practice. Specifically, we consider adjustments
\begin{equation}
\widehat\tau_a
=
\widehat\tau_{\DM}-a,
\qquad
a\in L^2(\mathcal F_X).
\label{eq:general-balance-action}
\end{equation}

We adopt the conventional mean-squared error (MSE) as our decision criterion and study two versions of it. Conditional MSE evaluates a correction over the assignments consistent with the observed information, whereas unconditional MSE evaluates it over the full randomization distribution.\footnote{For further discussion of conditional inference, see \citet{zhang2019conditional} and \citet{johansson2022inference}.} The following theorem shows that both criteria identify the same optimal correction: ex-post bias.

\begin{proposition}[Decision theorem]
\label{thm:conditional-decision}
For every square-integrable, $\mathcal F_X$-measurable adjustment $a(Z,X)$,
\begin{align}
\MSE_Z(\widehat\tau_{\DM}\mid\mathcal F_X)
-\MSE_Z(\widehat\tau_a\mid\mathcal F_X)
&=B_{\mathcal F_X}^2-(B_{\mathcal F_X}-a)^2,
\label{eq:value-bias1}\\
\MSE_Z(\widehat\tau_{\DM})-\MSE_Z(\widehat\tau_a)
&=\Var_Z(B_{\mathcal F_X})
-\E_Z\{(B_{\mathcal F_X}-a)^2\}.
\label{eq:value-bias}
\end{align}
The unique optimal correction is
\begin{equation}
a^*(Z,X)=B_{\mathcal F_X}
\label{eq:optimal-action}
\end{equation}
under both the unconditional MSE and the conditional MSE at every
realization of $\mathcal F_X$.
\end{proposition}

\noindent Setting $a=a^*$ in Equation~\eqref{eq:optimal-action} gives the oracle MSE gain directly:
$$
\begin{aligned}
  & \MSE_Z(\widehat\tau_{\DM}|\mathcal F_X)-
\MSE_Z(\hat{\tau}_{a^*}|\mathcal F_X)
=
B^2_{\mathcal F_X} \ge 0.\\
& \MSE_Z(\widehat\tau_{\DM})-
\MSE_Z(\hat{\tau}_{a^*})
=
\Var_Z(B_{\mathcal F_X}) \ge 0.
\end{aligned}
$$

Proposition~\ref{thm:conditional-decision} characterizes covariate adjustment as a problem of predicting the ex-post bias associated with the realized covariate allocation. Equation~\eqref{eq:value-bias1} shows that adjustment reduces conditional MSE if and only if the proposed correction approximates the ex-post bias more accurately than the zero correction, that is, when $(B_{\mathcal F_X}-a)^2<B_{\mathcal F_X}^2$. Across randomizations, Equation~\eqref{eq:value-bias} decomposes the effect of adjustment on MSE into its potential benefit and the cost of approximating the optimal correction. Specifically, $\Var_Z(B_{\mathcal F_X})$ represents the maximum attainable MSE reduction among $\mathcal F_X$-measurable corrections, whereas $\E_Z\{(B_{\mathcal F_X}-a)^2\}$ measures the loss in MSE reduction relative to this optimum. A misguided correction may therefore offset the potential benefit of adjustment or even increase MSE. Under both the conditional and unconditional MSE criteria, the optimal correction is the ex-post bias, $B_{\mathcal F_X}$.

\subsection{Approximating Ex-Post Bias}

The decision framework establishes that the central task of covariate adjustment is to approximate ex-post bias. In randomized experiments, both the randomization error and ex-post bias admit explicit theoretical expressions.

We first introduce notation for treatment--control imbalance. For any fixed scalar finite-population variable $A=(A_1,\ldots,A_N)$, define
\(
\Delta_Z(A)
=\frac1{N_1}\sum_{i=1}^NZ_iA_i
-\frac1{N_0}\sum_{i=1}^N(1-Z_i)A_i.
\) We distinguish expectation over treatment assignment, $\E_Z$, from
the empirical average over fixed units,
$\E_N(A_i)=N^{-1}\sum_{i=1}^NA_i=\overline A$. Finally, we use $\E$ to denote the combined expectation over treatment assignments and the finite population.\footnote{For any quantity $H_i(Z)$,
\(
\E[H_i(Z)]
=\E_Z\left[\frac{1}{N}\sum_{i=1}^N H_i(Z)\right].
\)}


\begin{definition}[Design-weighted randomization-error]
\label{def:M}
For treatment fraction $p=N_1/N$, define
\begin{equation}
M_i=(1-p)Y_i(1)+pY_i(0).
\label{eq:M}
\end{equation}
\end{definition}

\noindent The score $M_i$ is an oracle object because the two potential outcomes are not jointly observed. Under equal allocation,
\(
M_i=\frac{Y_i(1)+Y_i(0)}{2},
\)
so \(M_i\) is simply the midpoint of the two potential outcomes. Under unequal allocation, the midpoint is shifted toward the potential outcome observed in the smaller arm, whose sample composition has greater leverage on the difference-in-means estimator. The following proposition establishes that the realized treatment--control imbalance of $M_i$ exactly equals the error of the unadjusted estimator.

\begin{proposition}[Exact randomization-error identity]
\label{thm:M-identity}
Under complete randomization,
\begin{equation}
\widehat\tau_{\DM}-\tau
=
\DeltaZ(M)
=
(1-p)\DeltaZ\{Y(1)\}+p\DeltaZ\{Y(0)\}.
\label{eq:M-identity}
\end{equation}
\end{proposition}

\noindent Equivalently,
$
\Delta_z(M)
=
\left\{\overline Y_1(1)-\overline Y(1)\right\}
+
\left\{\overline Y(0)-\overline Y_0(0)\right\},
$ where $\overline Y_z(z)$ denotes the sample average of the potential outcomes. The first term is the error from using the realized treatment group to represent the population under treatment. The second is the error from using the realized control group to represent the population under control. \(\Delta_z(M)\) asks: Did the realized treatment group contain units with systematically higher or lower ATE-relevant outcome levels than the realized control group? Accordingly, \(\Delta_z(M)>0\) means that the realized difference-in-means estimate overstates the ATE, while \(\Delta_z(M)<0\) means that it understates the ATE.

By Propositions \ref{thm:conditional-decision} and \ref{thm:M-identity}, the oracle correction is $B_{\mathcal F_X}=\E_Z[\Delta_Z(M)|\mathcal F_X]$. It turns out it is equivalently to approximate $M_i$. Let $\mathcal H$ denote the class of all measurable functions
$g:\mathcal X\to\mathbb R$ satisfying
\(
\|g\|_N^2
:=
\E_N\{g(X_i)^2\}
<\infty.
\)
Define the oracle score by
\[
g_M
\in
\argmin_{g\in\mathcal H}
\E_N\!\left[\{M_i-g(X_i)\}^2\right].
\]
The minimizer is uniquely determined on the observed covariate support
and equals the conditional mean under finite-population averaging:
\(
g_M(X_i)=\E_N(M_i\mid X_i).
\) 


\begin{proposition}[Covariate-score representation of ex-post bias]\label{lem:FX-score-orthogonality}
Under complete randomization,
\begin{equation}
B_{\mathcal F_X}
=\Delta_Z\{g_M(X)\}.
\label{eq:FX-bias-score-representation}
\end{equation}
\end{proposition}


\noindent Therefore, the optimal adjustment estimator is
\begin{equation}
\widehat\tau(g_M)
=
\widehat\tau_{\DM}
-
\Delta_Z\{g_M(X)\}.
\label{eq:oracle-score-adjustment}
\end{equation}
Its optimality under both conditional and unconditional MSE follows directly from Proposition~\ref{thm:conditional-decision}. When the covariate values $X_i$ are distinct across units, we have $g_M(X_i)=M_i$, as expected. However, because $M_i$ is unobserved, $g_M$ cannot be directly evaluated. We discuss how to estimate $g_M$ in practice in Section~\ref{subsec:feasible-prognostic-adjustment}.


\begin{remark}
    Although the optimal correction is unique, its score representation need not be: adding a constant to $g_M$ leaves $\Delta_Z\{g_M(X)\}$ unchanged.
\end{remark}

\section{Prognosticity and Imbalance}\label{sec:prog_imbalance}

In the methodological literature and empirical practice, two common rationales for covariate adjustment emphasize different properties of pretreatment covariates. The prognostic rationale focuses on whether a covariate predicts the outcome, whereas the balance-based rationale focuses on whether its distribution differs between the treatment and control groups under the realized assignment. However, we show that these conventional interpretations do not fully capture the essential roles of prognosticity and imbalance in improving ATE estimation. Building on our framework, this section provides precise characterizations of prognosticity and imbalance and establishes the necessary and sufficient conditions under which they contribute to ATE estimation.

\subsection{ATE-Relevant Prognosticity}

Prognosticity concerns what a covariate predicts. 
Often, prognosticity is expressed in a somewhat vague manner.  In his review of the covariate adjustment literature, \citet[p.~241]{pocock2026covariate} writes 
``If the covariates are associated with the outcome, then the treatment effect estimate will be more precise, that is, with a smaller
standard error.'' 
\citet[p.~84]{angrist2009mostly} indicate that covariates with ``predictive power'' will improve precision, a point stressed as well by \citet[p.~331]{glennerster2013running}, who write ``Adding covariates increases the precision of our estimate as long as they explain quite a lot of the unexplained variance in the outcome measures.''  \citet[p.~26]{wu2011experiments} offer a variety of diagnostic statistics that researchers might use to evaluate prognosticity empirically, such as the adjusted $R^2$ or the Bayesian information criterion.

This practical guidance is helpful but stops short of defining prognosticity. For ATE estimation, it is not sufficient to ask whether $X$ predicts $Y(1)$ or $Y(0)$, or whether including $X$ increases a regression's adjusted $R^2$. The relevant question is whether $X$ predicts the particular combination
\(
M_i=(1-p)Y_i(1)+pY_i(0)
\)
that determines the randomization error of the difference-in-means estimator. We say that $X$ has \emph{ATE-relevant
mean prognosticity} when it improves prediction of $M$ beyond a constant, or equivalently when $\Var_N\{g_M(X)\}>0$. If this does not hold, then according to Proposition \ref{lem:FX-score-orthogonality}, the optimal adjustment is exactly zero , $B_{\mathcal F_X}=0$, and thus we cannot improve the raw DM estimator.

\subsubsection{Outcome mean prognosticity}

The most popular usage of prognosticity concerns the ability of pretreatment covariates to predict outcomes. For example, \citet[Section~3]{schuler2022increasing} note that the prognostic score is often understood informally as ``the expected value of the control
outcome given the baseline covariates,'' namely
$\mathbb{E}\{Y(0)\mid X\}$.\footnote{This mean-prediction usage should
be distinguished from the formal prognostic-score definition of
\citet{hansen2008prognostic}, which requires a function $\psi(X)$
satisfying $Y(0)\perp X\mid\psi(X)$. In general, the conditional mean alone need not satisfy this conditional-independence requirement.} The arm-specific outcome means for $z\in\{0,1\}$ are
\[
m_z(x)=\E_N\{Y(z)\mid X=x\}.
\] and the pooled-outcome mean prognosticity is
\begin{equation}
m_Y(x):=\E\{Y\mid X=x\}
=p\,m_1(x)+(1-p)\,m_0(x)
\label{pr:pooled-mean}
\end{equation}

Under squared-error loss, the best predictor of an outcome is the conditional mean. Thus, saying that $X$ predicts an outcome mean beyond a constant means that the conditional mean varies across covariate profiles. However, outcome-mean prognosticity and ATE-relevant mean prognosticity concern different prediction targets. For ATE estimation, the relevant prediction target is instead
\begin{equation}
g_M(x)=\E_N\{M\mid X=x\}
=(1-p)\,m_1(x)+p\,m_0(x).
\label{pr:score-mean}
\end{equation}
There are two main differences. First, the control outcome mean depends only on $m_0$, whereas the ATE-relevant function $g_M$ also requires information about $m_1$. Second, although the pooled outcome mean incorporates both $m_0$ and $m_1$, the weights assigned to these two components are reversed relative to those used in $g_M$. Consequently, predicting an outcome mean does not generally amount to predicting the mean relevant to ATE estimation. The following proposition gives two cases in which these targets agree up to
additive constants.

\begin{proposition}[Outcome-mean and ATE-relevant mean prognosticity]\label{pr:mean-targets}
Let $\tau(x)=m_1(x)-m_0(x)$ denote the conditional mean
treatment effect. Then
\begin{align}
g_M(x)&=m_0(x)+(1-p)\tau(x),\\
m_Y(x)&=g_M(x)+(2p-1)\tau(x).
\end{align}
Consequently, 

(1) The control-outcome mean $m_0(X)$ and $g_M(X)$ differ only by additive constants, and hence have the same mean prognosticity, if and only if $\tau(X)$ is constant.

(2) The pooled-outcome mean $m_Y(X)$ and the ATE-relevant
mean $g_M(X)$ differ only by an additive constant, and hence have the same mean prognosticity, if and only if treatment allocation is balanced ($p=1/2$) or $\tau(X)$ is constant.
\end{proposition}

The decomposition $g_M(X)=m_0(X)+(1-p)\tau(X)$ separates the control-outcome mean from the conditional mean treatment effect. The control-outcome mean yields the oracle correction for every assignment when $\tau(X)$ is constant. When $\tau(X)$ varies, its contribution can reinforce or offset variation in $m_0(X)$. Thus, a covariate can predict the control-outcome mean while having no ATE-relevant mean prognosticity if variation in $m_0(X)$ is canceled by variation in $(1-p)\tau(X)$. Conversely, when $m_0(X)$ is constant
but $\tau(X)$ varies, $X$ has ATE-relevant mean prognosticity entirely through conditional mean treatment effect heterogeneity. 

The pooled-outcome mean weights $m_1$ and $m_0$ by $p$ and $(1-p)$, whereas the ATE-relevant mean weights them by $(1-p)$ and $p$ \citep{lin2013agnostic}. Equal allocation makes these weights identical, so pooled-outcome mean prediction targets $g_M$, without any restriction on $\tau(X)$. Under unequal allocation, the targets agree up to a constant when $\tau(X)$ is constant. 

We now give a concrete example to illustrate the limitations of outcome-mean prognosticity.

\begin{example}[Outcome prognosis without ATE-relevant prognosis]
\label{ex:outcome-prognosis-positive}
Consider a four-unit finite population containing each pair in
\(
(X_i,U_i)\in\{-1,0\}\times\{0,2\}
\)
just once. The covariate $X_i$ is observed, whereas $U_i$ is
unobserved. Exactly one unit is
assigned to treatment, so $p=1/4$ and $q=3/4$. Let
\begin{equation}\label{eq:positive-outcome-prognosis-dgp}
Y_i(1)=6+(1+X_i)^2+U_i,
\qquad
Y_i(0)=1+3X_i^2+U_i.
\end{equation}
Individual treatment effects equal $2$ when $X_i=-1$
and $6$ when $X_i=0$.

The potential-outcome conditional means are
\[
m_1(x):=\mathbb E_N\{Y(1)\mid X=x\}=7+(1+x)^2,
\qquad
m_0(x):=\mathbb E_N\{Y(0)\mid X=x\}=2+3x^2.
\]
Both vary with $x$. Averaging over both unit and
randomized assignment gives
\begin{align*}
m_Y(x)
&:=\mathbb E(Y\mid X=x)=\frac72+2x^2
\end{align*}
The covariate therefore predicts the pooled observed-outcome mean, not merely the control-outcome mean. 

In contrast, the ATE-relevant score is
\begin{align*}
M_i
&=\frac34Y_i(1)+\frac14Y_i(0)\\
&=\frac{19}{4} +\frac34\bigl\{(1+X_i)^2+X_i^2\bigr\}+U_i\\
&=\frac{11}{2}+U_i
\end{align*}
The last equality follows from $(1+X_i)^2+X_i^2=1$. Consequently, \(g_M(x):=\mathbb E_N(M\mid X=x)=\frac{13}{2}\) for both $x\in\{-1,0\}$. Thus, $X$ provides no information about $M$. In particular,
\(\operatorname{Var}_N\{g_M(X)\}=0\) and \(B_{\mathcal{F}_X}=\Delta_Z\{g_M(X)\}=0\). In other words, although covariate $X$ is prognostic for the outcome, it is not prognostic for the $ATE$.


\end{example}


\subsubsection{$R^2$ and ATE-relevant prognosticity}
\label{subsubsec:population-r2-ate}

The preceding subsection distinguishes outcome-mean and ATE-relevant mean prognosticity according to what the covariates predict. Regression $R^2$ provides another measure of predictive accuracy within a specified linear class, but it does not obviate this distinction. For example, \citet[Section~2.2]{bicalho2026power} measure joint covariate prognosis using the $R^2$ from a regression of control-group outcomes on covariates. Such a measure evaluates prediction of $Y(0)$. Alternatively, researchers could evaluate the $R^2$ in the
pooled outcome $Y$ regression. However, neither approach is generally equivalent to approximating $M_i$.

Let $X_1$ denote the existing regressors, possibly a vector, and let
$X_2$ be an additional scalar regressor. Let $\mathcal G_1$ and $\mathcal G_2$ be the spaces of linear functions generated by $(1,X_1)$ and $(1,X_1,X_2)$, respectively. Write 
$$R
_{0,j}^2=1-\frac{\min_{g \in \mathcal{G}_j}\E_N[(Y(0)-g(X))^2]}{\Var(Y_N(0))}
$$ 
for the finite-population $R^2$ from projecting $Y(0)$ onto $\mathcal G_j$, and write
$$
R_j^2=1-\frac{\min_{\gamma \in \R,g \in \mathcal{G}_j}\E[(Y-\gamma Z-g(X))^2]}{\Var(Y)}
$$ for the population $R^2$ when regressing $Y$ on $Z$ and the regressors in $\mathcal G_j$ (Excluding the treatment in the population outcome regression leaves the covariate slopes unchanged and decreases $R^2$ by $\frac{p(1-p)\tau^2}{\Var(Y)}$; Therefore, results in Proposition \ref{pr:r2-gains} also hold.) For comparison, define the best fixed predictor of $M$ in each class and its associated ATE estimator:
\begin{equation}
\begin{aligned}
g_j
&\in\argmin_{g\in\mathcal G_j}
       \E_N\bigl[\{M_i-g(X_i)\}^2\bigr],\\
\widehat\tau(g_j)
&=\widehat\tau_{\DM}-\Delta_Z\{g_j(X)\},
\qquad j\in\{1,2\}.
\end{aligned}
\label{pr:r2-oracle-estimators}
\end{equation}
Let $\widetilde X_{i2}$ be the residual from projecting $X_{i2}$ on
$(1,X_{i1})$ over all $N$ units. Assume
$\E_N(\widetilde X_2^2)>0$, and define
\[
b_z
=\frac{\E_N\{\widetilde X_2Y(z)\}}
       {\E_N(\widetilde X_2^2)}
.
\]
The coefficient $b_z$ measures the partial linear association of the
added regressor with $Y(z)$ after accounting for the existing regressors.

\begin{proposition}[Regression fit and oracle ATE mean-squared error]
\label{pr:r2-gains}
Under complete randomization and the preceding definitions,
suppose $0<p<1$,
$\Var_N(\widetilde X_2)>0$,
$\Var_N\{Y(0)\}>0$, and $\Var(Y)>0$. Then
\begin{align}
R_{0,2}^2-R_{0,1}^2
&=
\frac{\Var_N(\widetilde X_2)}
     {\Var_N\{Y(0)\}}b_0^2,
\label{pr:r2-control-gain}\\
R_2^2-R_1^2
&=
\frac{\Var_N(\widetilde X_2)}
     {\Var(Y)}(pb_1+(1-p)b_0)^2,
\label{pr:r2-outcome-gain}\\
\MSE_Z\{\widehat\tau(g_1)\}
-\MSE_Z\{\widehat\tau(g_2)\}
&=
\frac{\Var_N(\widetilde X_2)}
     {(N-1)pq}((1-p)b_1+pb_0)^2.
\label{pr:r2-mse-gain}
\end{align}
Consequently:
\begin{enumerate}
\item The equivalence
\(
R_2^2>R_1^2
\Longleftrightarrow
\MSE_Z\{\widehat\tau(g_2)\}
<
\MSE_Z\{\widehat\tau(g_1)\}
\)
holds for every admissible pair of arm-specific partial
slopes $(b_1,b_0)$ if and only if allocation is balanced,
$p=1/2$.

\item If $b_1=b_0$, an increase in either control-outcome or pooled-outcome $R^2$ is equivalent to a decrease in oracle ATE MSE, for any allocation fraction. 
\end{enumerate}

\end{proposition}

Proposition~\ref{pr:r2-gains} identifies the exact partial association relevant to each criterion. Control-outcome fit depends on $b_0$, pooled-outcome fit depends on $pb_1+(1-p)b_0$, and oracle ATE precision depends on $(1-p)b_1+pb_0$. 

Balanced allocation makes the pooled-outcome and ATE-relevant partial slopes identical for every pair of arm-specific slopes. Balanced allocation is therefore necessary and sufficient for improved pooled-outcome fit to identify an oracle MSE reduction. The condition $b_1=b_0$ provides another sufficient condition for all three partial slopes to coincide. However, this condition also implies that the additional covariate has no partial linear association with individual treatment effects.

Under other scenarios, such as unequal allocation, outcome fit improvement and oracle MSE reduction need not imply one another without further restrictions. We illustrate this in the following example.

\begin{example}[Raising population $R^2$ without incremental ATE-relevant prognosticity]
\label{pr:r2-example}
Let the finite population contain every combination of
$(X_{i1},X_{i2},U_i)\in\{-1,1\}^3$ with equal multiplicity, and let
$p=1/4$ and $q=3/4$. Suppose
\[
Y_i(1)=4+X_{i1}+X_{i2}+U_i,
\qquad
Y_i(0)=X_{i1}-3X_{i2}+U_i.
\]
The three components are mutually independent, mean zero, and unit
variance under finite-population averaging. Consequently,
$\widetilde X_2=X_2$, $b_1=1$, and $b_0=-3$.

The ATE-relevant score is
\[
M_i=(1-p)Y_i(1)+pY_i(0)=3+X_{i1}+U_i.
\]
Because
\(
\E_N(M\mid X_1,X_2)
=\E_N(M\mid X_1)
=3+X_1,
\)
the additional covariate $X_2$ has no incremental predictive value
for $M$, even without restricting the prediction to be linear.
In particular, the best scores in the two linear classes coincide:
\(
g_1(X)=g_2(X)=3+X_1.
\)
Both leave residual $M_i-g_j(X)=U$, whose finite-population variance
is one. Their associated ATE estimators therefore coincide for every
assignment and have the same randomization MSE:
\begin{equation}
\MSE_Z\{\widehat\tau(g_1)\}
=
\MSE_Z\{\widehat\tau(g_2)\}
=
\frac{1}{(N-1)pq}.
\label{pr:r2-example-oracle-mse}
\end{equation}

Nevertheless, adding $X_2$ improves the pooled-outcome regression fit.
Since
\[
Y=4Z+X_1+(4Z-3)X_2+U,
\qquad
\Var(Y)=12,
\]
the two population regressions have fitted values $4Z+X_1$ and
$4Z+X_1-2X_2$, respectively. Their mean-squared prediction losses
are $8$ and $4$, giving $R_1^2=\frac13$ and $R_2^2=\frac23$. The difference between improved outcome fit and unchanged oracle
ATE MSE follows from the two partial slopes: $pb_1+(1-p)b_0=-2$ and $(1-p)b_1+pb_0=0$. 
Thus, the added covariate predicts the pooled outcome but not the score that determines ATE estimation error.

\end{example}

\subsection{ATE-Relevant Imbalance}
\label{subsec:conventional-ate-prognosticity}

Beyond outcome prognosticity, raw covariate imbalance also leads some researchers to adjust. Given an assignment $z$, raw covariate mean differences is $\Delta_z(X)$, while the oracle adjustment is $B_{\mathcal F_X}(z)=\Delta_z\{g_M(X)\}$, which we call ATE-relevant imbalance. Raw covariate imbalance need not coincide with ATE-relevant imbalance.


\begin{proposition}[General characterization of relevant imbalance]
\label{prop:general-allocation-imbalance}
Let \(
 \pi_z(x)=\E_N(z_i\mid X_i=x).
\) For every feasible assignment $z$,
\begin{align}
 \Delta_z(X)
 &=\frac{\Cov_N\{X,\pi_z(X)\}}{pq},
 \label{eq:imbalance-raw-allocation-covariance}\\
 B_{\mathcal F_X}(z)
 &=\frac{\Cov_N\{g_M(X),\pi_z(X)\}}{pq}.
 \label{eq:imbalance-score-allocation-covariance}
\end{align}
Consequently, raw mean balance holds if and only if
$\Cov_N\{X,\pi_z(X)\}=0$, whereas ATE-relevant balance holds if and only if
$\Cov_N\{g_M(X),\pi_z(X)\}=0$.
\end{proposition}

Proposition~\ref{prop:general-allocation-imbalance}
shows that raw mean imbalance and ATE-relevant imbalance
measure the association of the same realized covariate
allocation with different quantities. Raw mean differences
describe how the allocation is associated with the recorded
covariates, whereas ATE-relevant imbalance describes its
association with the score $g_M(X)$. Therefore, neither raw mean balance nor raw mean imbalance generally determines whether the ATE-relevant correction is zero. The relationship depends on how $g_M(X)$ varies
across the covariate profiles that are overrepresented or underrepresented in treatment.

From equation \eqref{eq:imbalance-score-allocation-covariance}, using
$\operatorname{Corr}_N$ as correlation under finite-population
averaging, we obtain:
\begin{equation}
\begin{split}
 B_{\mathcal F_X}(z)
 ={}&
 \frac{
 \sqrt{\Var_N\{g_M(X)\}}
 \sqrt{\Var_N\{\pi_z(X)\}}
 }{p(1-p)}\\
 &{}\times
 \operatorname{Corr}_N\{g_M(X),\pi_z(X)\}.
\end{split}
 \label{eq:imbalance-score-allocation-factorization}
\end{equation}
Equation~\eqref{eq:imbalance-score-allocation-factorization} highlights three components that jointly determine ex-post bias: ATE-relevant prognosticity, measured by variation in $g_M(X)$; distributional imbalance, measured by variation in the realized treated fraction across covariate profiles; and the association between the two. Thus, both ATE-relevant prognosticity and covariate distributional imbalance are necessary for a nonzero oracle correction under the realized assignment, but they are generally insufficient.

In some special cases, raw covariate mean differences can have a direct relationship with ex-post bias. 

\begin{proposition}[Linear representation of prognostic imbalance]
\label{prop:affine-prognostic-imbalance}
Suppose $0<N_1<N$. For a fixed vector $\gamma\in\mathbb R^K$,
\begin{equation}
 B_{\mathcal F_X}(z)=\gamma^{\mathsf T}\Delta_z(X)
 \quad\text{for every }z\in\mathcal Z_{N_1}
 \label{eq:affine-prognostic-imbalance}
\end{equation}
if and only if there exists $a\in\mathbb R$ such that
$g_M(X_i)=a+\gamma^{\mathsf T}X_i$ for every unit $i$.
\end{proposition}

Under this affine representation, exact balance in every raw
covariate mean guarantees $B_{\mathcal F_X}=0$, and thus there is no gain in adjustment. Nonzero raw differences can, however, cancel after multiplication by their prognostic
coefficients $\gamma$ in equation \eqref{eq:affine-prognostic-imbalance}. 
For example, coefficients $(2,4)$ and raw differences $(2,-1)$ give a zero correction, although both covariates have nonzero coefficients and both are imbalanced. 

We now provide an example to illustrate that raw covariate imbalance of prognostic variable does not necessarily imply ATE-relevant imbalance:
$$
\left[
\Delta_z(X)\neq0
\quad\text{and}\quad
X\text{ predicts }Y(1)\text{ and }Y(0)
\right]
\quad\not\Rightarrow\quad
 B_{\mathcal F_X} \neq0.
$$

\begin{example}[Offsetting Prognostic Imbalances under Unequal Allocation]
\label{ex:cancellation}
Consider four units, one of which is assigned to treatment, so that
$p=1/4$. Let $X_i=(X_{i1},X_{i2})^\trans$ contain two pretreatment
covariates, and suppose
\begin{equation}
 Y_i(1)=24+X_{i1}+5X_{i2},
 \qquad
 Y_i(0)=16+5X_{i1}+X_{i2}.
 \label{eq:same-sign-example-potential-outcomes}
\end{equation}
Both covariates positively predict both potential outcomes. Take the
following finite population and assign unit $1$ to treatment:
\[
\begin{array}{c|rrrr}
 i&1&2&3&4\\
 \hline
 X_{i1}&2&-2&0&2\\
 X_{i2}&0&0&1&2\\
 Y_i(1)&26&22&29&36\\
 Y_i(0)&26&6&17&28\\
 Z_i&1&0&0&0\\
 Y_i&26&6&17&28
\end{array}
\]

Both the control-outcome mean and the pooled observed-outcome mean vary
with the covariates. In the control arm, outcomes follow
$Y_i=16+5X_{i1}+X_{i2}$. Averaging the observed outcome over
randomization gives
\[
 \E_Z(Y_i)
 =\frac14Y_i(1)+\frac34Y_i(0)
 =18+4X_{i1}+2X_{i2},
\]
so the pooled observed-outcome mean also varies with the covariates.

Both covariates are also imbalanced in this assignment:
\[
 \Delta_Z(X^1)
 =2-\frac{-2+0+2}{3}=2,
 \qquad
 \Delta_Z(X^2)
 =0-\frac{0+1+2}{3}=-1.
\]
However, their ATE-relevant prognostic contributions have offsetting
imbalances. The randomization-error score is
$
M_i
=\frac34Y_i(1)+\frac14Y_i(0)\nonumber=22+2X_{i1}+4X_{i2}
=g_M(X_i).
$
Consequently, $
 \Delta_Z(M)
 =2\Delta_Z(X^1)+4\Delta_Z(X^2)
 =2(2)+4(-1)=0.
$
Indeed, the finite-population ATE is
\[
 \tau
 =\frac{(26-26)+(22-6)+(29-17)+(36-28)}4
 =9,
\]
and the realized difference in means is
\(
 \widehat\tau_{\DM}
 =26-\frac{6+17+28}{3}=9.
\)
Thus, despite the observed raw covariate imbalance, no covariate adjustment is needed for ATE estimation in this example.
\end{example}

Conversely, we can construct an example in which exact raw covariate mean balance and an absence of pooled-outcome mean prognosticity can coexist with nonzero ATE-relevant imbalance:
$$
\left[
\Delta_z(X) = 0
\text{ and }
X\text{ does not predict mean of } Y
\right]
\not\Rightarrow B_{\mathcal F_X}=0
$$

\begin{example}[Raw mean balance with prognostic imbalance]
\label{ex:reverse-cancellation}
Consider four units, exactly one of which is assigned to treatment. Let $X_i\in\{-2,-1,0\}$ and specify
\begin{equation}
Y_i(1)=1+3(1+X_i)^2,
\qquad
Y_i(0)=2+\{X_i(2+X_i)\}^2.
\label{eq:reverse-example-potential-outcomes}
\end{equation}

Take the following finite population and realized assignment:
\[
\begin{array}{c|rrrr}
i&1&2&3&4\\
\hline
X_i    &-1&-2&-1&0\\
Y_i(1) &1&4&1&4\\
Y_i(0) &3&2&3&2\\
Z_i    &1&0&0&0\\
Y_i    &1&2&3&2
\end{array}.
\]
The raw covariate means are exactly balanced:
\[
\overline X_1=-1,
\qquad
\overline X_0=\frac{-2-1+0}{3}=-1,
\qquad
\Delta_Z(X)=0.
\]

Moreover, its association with the two potential-outcome means
cancels in the pooled observed-outcome mean:
\begin{align*}
m_Y(x)
&:=\mathbb E(Y\mid X=x)\\
&=\frac14\{1+3(1+x)^2\}
  +\frac34\{2+\{x(2+x)\}^2\}\\
&=\frac74+\frac34
  \bigl[(1+x)^2+\{x(2+x)\}^2\bigr]
=\frac52.
\end{align*} where the last equality follows because $(1+x)^2+\{x(2+x)\}^2=1$ on the specified support $x\in\{-2,-1,0\}$. Thus $X$ supplies no predictive value for the pooled outcome mean under squared-error loss. The same absence of pooled mean prognosis holds empirically in the displayed assignment. The observed outcomes are $(1,2,3,2)$, so
\[
\overline Y_{X=-2}=2,
\qquad
\overline Y_{X=-1}=\frac{1+3}{2}=2,
\qquad
\overline Y_{X=0}=2.
\]
No function of $X$ improves the empirical pooled squared-error
loss over the best constant prediction.

However, the ATE-relevant score is
\begin{align}
M_i
&=\frac34Y_i(1)+\frac14Y_i(0)\notag\\
&=\frac54+\frac94(1+X_i)^2
  +\frac14\{X_i(2+X_i)\}^2\notag\\
&=\frac32+2(1+X_i)^2
\label{eq:reverse-example-score}
\end{align}
Thus $X$ perfectly predicts $M$ through the nonlinear score, and the optimal adjustment is
\[
B_{\mathcal F_X}
=\mathbb E_Z(e_0\mid\mathcal F_X)
=\Delta_Z\{g_M(X)\}
=\Delta_Z(M)
=-\frac43.
\]

\end{example}

In general, there is no direct relationship between raw covariate imbalance and ATE-relevant imbalance. Nonzero raw mean imbalance is neither necessary nor sufficient for a nonzero oracle correction. In the affine case, nonzero raw mean imbalance is necessary but generally remains insufficient. The relevant quantity for ATE estimation is $\Delta_z\{g_M(X)\}$, whose sign and magnitude characterize the oracle correction in outcome units. 

In summary, conventional measures of prognosticity and covariate imbalance are neither necessary nor sufficient for determining whether covariate adjustment provides more precise estimates than the raw DM estimator. Instead, the potential precision gains from adjustment come down to whether 
$X \text{ or } g_M(X)\text{ predicts }M$.  The oracle correction for the realized covariate allocation comes down to whether 
$g_M(X)\text{ is imbalanced in the realized assignment}$.\footnote{This section has focused on population quantities. In Appendix \ref{sec:sample-prognosticity-fit-imbalance}, we extend the discussion to estimated measures of prognosticity and imbalance.}


\section{Feasible Estimation}
\label{subsec:feasible-prognostic-adjustment}

The preceding section discussed the oracle adjustment problem: we need a pretreatment score \(g(X)\) that approximates $M$ to produce the adjusted estimator
\(
\widehat\tau(g)
=
\widehat\tau_{\DM}-\Delta_Z\{g(X)\}.
\) The difficulty is that \(M_i\) is not observed because \(Y_i(1)\) and \(Y_i(0)\) are never jointly
observed. There are two workarounds. Researchers may learn the single score \(g\) directly from observed outcomes, or they may estimate the two arm-specific outcome functions and combine them using the complementary design weights. We show that many existing estimators can be classified within one of these two approaches. Rather than proposing new estimators, we provide a unified theoretical justification for existing methods.



\subsection{Joint Learning of the Single Score}\label{sec:joint}

To learn \(M\) without observing both potential outcomes, consider the design-weighted finite-population loss
\begin{equation}
\mathcal L_p(g)
=
(1-p)\frac1N\sum_{i=1}^N
\{Y_i(1)-g(X_i)\}^2
+
p\frac1N\sum_{i=1}^N
\{Y_i(0)-g(X_i)\}^2.
\label{eq:design-weighted-loss}
\end{equation}
Its observable analogue is
\begin{equation}
\widehat{\mathcal L}_p(g)
=
(1-p)\frac1{N_1}\sum_{i=1}^N
Z_i\{Y_i-g(X_i)\}^2
+
p\frac1{N_0}\sum_{i=1}^N
(1-Z_i)\{Y_i-g(X_i)\}^2.
\label{eq:sample-design-loss}
\end{equation}
Under unequal allocation, the
criterion includes complementary weights:
\[
\widehat{\mathcal L}_p(g)
=
\frac1N\sum_{i=1}^N
\omega_i\{Y_i-g(X_i)\}^2, \quad \text{where} \quad \omega_i
=
\begin{cases}
(1-p)/p, & Z_i=1,\\
p/(1-p), & Z_i=0.
\end{cases}
\] 
Thus, the potential-outcome surface observed in the smaller treatment arm receives greater weight because its composition contributes more to the randomization error of the difference-in-means estimator.

Under equal allocation, the observed criterion reduces to pooled squared-error loss:
\(
\widehat{\mathcal L}_{1/2}(g)
=
\frac1N\sum_{i=1}^N
\{Y_i-g(X_i)\}^2.
\) In essence, under a balanced design the adjustment goal can be simplified to finding a function of $X$ that best approximates the outcome $Y$.

The following proposition shows that the observable criterion targets the prediction problem identified by the oracle MSE analysis.

\begin{proposition}\label{thm:observable-loss}
Let \(\mathcal G\) be a prespecified class of functions, and let \(\mathcal L_p(g)\) be defined as in \eqref{eq:design-weighted-loss}. For every
\(g\in\mathcal G\),
\(
\argmin_{g\in\mathcal G}\mathcal L_p(g)
=
\argmin_{g\in\mathcal G}
\frac1N\sum_{i=1}^N
\{M_i-g(X_i)\}^2 \), and  \(
\E_Z\{\widehat{\mathcal L}_p(g)\}
=
\mathcal L_p(g)
\).
\end{proposition}

Proposition~\ref{thm:observable-loss} shows that our proposed loss function targets the ATE-relevant quantity $M_i$. We define the estimator as
\(
\widehat g
\in
\argmin_{g\in\mathcal G}
\widehat{\mathcal L}_p(g).
\)
We recommend using flexible machine-learning methods to estimate $g$, allowing the procedure to accommodate high-dimensional covariates and complex functional relationships \citep{bloniarz2016lasso,wager2016high,wu2018loop}. Depending on the method employed, researchers may combine regularization with cross-fitting or leave-one-out prediction.
For example, partition the observations into folds \(I_1,\ldots,I_V\), estimate \(\widehat g^{(-v)}\) without using units in fold \(I_v\), and set
\(
\widehat g_i
=
\widehat g^{(-v)}(X_i)\), \(
 i\in I_v.
\)
The resulting joint score-adjustment estimator is
\begin{equation}
\begin{aligned}
\widehat\tau_{\mathrm{dir}}
&=
\frac1{N_1}\sum_{i=1}^N
Z_i\{Y_i-\widehat g_i\}
-
\frac1{N_0}\sum_{i=1}^N
(1-Z_i)\{Y_i-\widehat g_i\}\\
&=
\widehat\tau_{\DM}
-
\Delta_Z(\widehat g).
\end{aligned}
\label{eq:direct-score-estimator}
\end{equation}
We provide results on consistency and variance estimation in SI \ref{si:direct}. It also provides conditions under which the fitted estimator has the same first-order asymptotic distribution as an estimator based on an adjustment rule fixed across treatment assignments. Proposition~\ref{thm:conditional-decision} then applies exactly to this reference adjustment. We note that \citet{zhang2019machine} develop a related approach that directly learns an augmentation function and derives the corresponding variance estimator from a semiparametric perspective. Their empirical loss uses an estimated, arm-centered influence function response, whereas our learner uses the raw-outcome complementary-weighted loss. These criteria have the same population target after translation and scale, but their finite-sample minimizers need not coincide.

\begin{remark}[Special case under sharp null]
\label{rem:sharp-null-prediction}
Under Fisher's sharp null,
$H_0^{\mathrm F}:Y_i(1)=Y_i(0)$ for every $i$, the oracle score is
fully observed:
\begin{equation}
M_i=Y_i(1)=Y_i(0)=Y_i
\label{eq:sharp-null-observed-score}
\end{equation}
Thus, the oracle prediction loss is observable, and covariates and tuning
parameters can be selected using pooled prediction of $Y$ without complementary weighting. In practice, this special case applies to a broad class of experiments that seek to debunk dubious causal claims (e.g., the effects of astrological sign).
\end{remark}

\subsection{Arm-Specific Outcome Regression}
\label{sec:arm}

Since \(g_M(x)=(1-p)\,m_1(x)+p\,m_0(x)\), instead of learning a single score jointly, researchers may estimate the arm-specific outcome functions and combine their predictions. The arm-specific estimator can be readily expressed as
\begin{equation}
\begin{aligned}
\widehat\tau_{\mathrm{arm}}
={}&\frac1N\sum_{i=1}^N(\widehat m_{1i}-\widehat m_{0i})\\
&+\frac1{N_1}\sum_{i:Z_i=1}(Y_i-\widehat m_{1i})
-\frac1{N_0}\sum_{i:Z_i=0}(Y_i-\widehat m_{0i}).
\end{aligned}
\label{eq:greg}
\end{equation}



In a low-dimensional setting, researchers may estimate \(m_1\) and \(m_0\) using separate regressions within the treatment and control groups. With linear arm-specific regressions on centered covariates, the resulting estimator is equivalent to fully interacted regression adjustment \citep{lin2013agnostic,negi2021revisiting}. More generally, Equation~\eqref{eq:greg} is the standard randomized-trial AIPW or generalized-regression estimator studied in the semiparametric literature
\citep{robins1994estimation,robins1995analysis,tsiatis2008covariate, zhang2008improving}. Therefore, we do not repeat the variance estimation and distributional results here. 

\textbf{Summary.} In practice, researchers need not classify each pretreatment covariate individually as prognostic or imbalanced. Instead, they can prespecify a substantively motivated set of baseline covariates and use them in joint or arm-specific estimation. When the set of candidate covariates is large or more flexible prediction methods are employed, regularization and appropriate procedures such as cross-fitting or leave-one-out prediction can help control overfitting and estimation bias, subject to the conditions established in the relevant literature.


\section{Simulation Design and Results}
\label{sec:simulation}

In this section, we conduct simulation studies to examine our theoretical framework and evaluate the performance of alternative covariate-adjustment strategies. The main simulation study considers four scenarios with $N=400$ and $K=60$. For adjustment strategies involving machine learning, we use gradient-boosted trees. Appendix~\ref{app:lasso-simulation} reports the corresponding results using LASSO, Appendix~\ref{app:small-n-simulation} examines performance in smaller samples, and Appendix~\ref{app:selection-simulation} compares alternative covariate-selection strategies.

\subsection{Estimators}
\label{subsec:simulation-comparisons}

We begin with the difference-in-means estimator as the unadjusted benchmark, followed by two regression-adjustment estimators. All-\(X\) OLS represents the familiar additive regression adjustment using the full set of covariates, while All-\(X\) Lin allows the covariate slopes to differ across treatment arms. Both use the 60 raw covariates.

The remaining estimators focus on learning the adjustment score. Following Section~\ref{subsec:feasible-prognostic-adjustment}, we consider joint and arm-specific estimators. Joint learning estimates a single score under the complementary design-weighted loss. Arm-specific learning instead fits separate treatment- and control-arm outcome models and combines their predictions; the resulting estimator has the familiar cross-fitted AIPW form with the known assignment probability. We implement both approaches with gradient-boosted trees \citep{chen2016xgboost}. 

Under unequal allocation, we also include pooled unweighted learning. It fits the observed outcome using the same learner and covariates as Joint learning but gives every observation equal weight. In other words, it targets $\E[Y|X=x]=pg_1(x)+(1-p)g_0(x)$, while our joint learning targets $\E[M|X=x]=(1-p)g_1(x)+pg_0(x)$. This comparison isolates the choice of prediction loss. (Under balanced allocation, their criteria and fitted predictions coincide, so we do not report duplicate entries.) Finally, the DGP oracle uses the known systematic score and therefore benchmarks performance in the absence of score-estimation error. Table~\ref{tab:simulation-methods} summarizes these comparisons. Implementation details are included in SI~\ref{si:implement}.

\begin{table}[!htbp]
\centering
\small
\renewcommand{\arraystretch}{1.16}
\caption{Estimators in the main simulation}
\label{tab:simulation-methods}
\begin{tabularx}{\textwidth}{>{\raggedright\arraybackslash}p{0.22\textwidth} X >{\raggedright\arraybackslash}p{0.28\textwidth}}
\hline
Estimator & Construction & Reason for inclusion \\
\hline
Difference in means (DM)
& No covariate adjustment.
& No-adjustment benchmark. \\
All-$X$ OLS
& Additive OLS using all 60 raw covariates.
& Additive regression adjustment. \\
All-$X$ Lin
& Fully interacted Lin regression using all 60 raw covariates.
& Adjustment with arm-specific slopes. \\
Joint learning
& Cross-fitted estimator introduced in section \ref{sec:joint}.
& Expect better performance in some cases.\\
Arm-specific learning (AIPW)
& Cross-fitted estimator introduced in section \ref{sec:arm}.
& Cross-fitted AIPW implementation. \\
Pooled unweighted learning
& Cross-fitted estimator under ordinary pooled squared-error loss.
& Isolates the role of design weighting under unequal allocation. \\
DGP oracle
& Adjustment using the known systematic score.
& Reference using the known generating-model score. \\
\hline
\end{tabularx}
\begin{minipage}{0.96\textwidth}
\footnotesize
\emph{Notes:} Joint, Arm-specific, and Pooled unweighted learning each use gradient-boosted trees. Pooled unweighted learning is reported separately only under unequal allocation.
\end{minipage}
\end{table}

\subsection{Data-generating framework}
\label{subsec:simulation-design}

To evaluate these estimators in a design-based setting, we generate one finite-population science table for each case and functional-form specification, with $N=400$ units and $K=60$ pretreatment covariates. We then hold all covariates and both potential outcomes fixed across the 1,000 treatment-assignment replications. At construction, covariates are independent across units and follow
\begin{equation}
X_i\sim\mathcal N(0,\Sigma),
\text{ where } 
\Sigma_{jk}=0.3^{|j-k|},
i=1,\ldots,N.
\label{eq:simulation-x}
\end{equation}
Thus, every covariate has variance one under the generating distribution, and correlations decline with the distance between their indices. 

We generate the potential outcomes through the unit-level treatment effect $\tau_i$ and the score $M_i$ from Proposition~\ref{thm:M-identity}. For scenario $s$, let $g_s^\circ(X)$ be a normalized covariate index that carries useful signal for $M$, and let $h_s^\circ(X)$ be a normalized index that generates treatment-effect heterogeneity. We use the prespecified covariance matrix $\Sigma$ for all normalizations. Define
\begin{align}
M_i
=
\sqrt{R_M^2}\,&g_s^\circ(X_i)
+\sqrt{1-R_M^2}\,\varepsilon_i,
 \quad
\tau_i
=0.20+\sigma_{\tau,s}h_s^\circ(X_i),
\nonumber\\
&Y_i(1)=M_i+p\tau_i, \quad
\text{ and } \quad
Y_i(0)=M_i-(1-p)\tau_i.
\label{eq:simulation-dgp}
\end{align}
This construction guarantees both $Y_i(1)-Y_i(0)=\tau_i$ and $qY_i(1)+pY_i(0)=M_i$.

The parameter $R_M^2$ controls how much of the relevant score can be predicted from baseline covariates. Because every nonzero $g_s^\circ(X)$ is normalized to have variance one under that distribution and is independent of $\varepsilon_i\sim\mathcal N(0,1)$, the corresponding generating-distribution moments satisfy
\begin{equation}
\E(M_i\mid X_i)=\sqrt{R_M^2}\,g_s^\circ(X_i),
\text{ and }
R_M^2=
\frac{\Var\{\E(M_i\mid X_i)\}}{\Var(M_i)}.
\label{eq:simulation-r2}
\end{equation}
Thus, $R_M^2=0$ specifies no systematic covariate signal in the generating model, whereas $R_M^2=0.50$ specifies that covariates explain one half of its variation in $M$. 
Table~\ref{tab:simulation-dgps} summarizes how the four scenarios vary the predictive signal, treatment-effect heterogeneity, and allocation ratio. Its final column explains why each scenario is included. 

\begin{table}[!htbp]
\centering
\small
\renewcommand{\arraystretch}{1.16}
\caption{Structural scenarios in the main simulation}
\label{tab:simulation-dgps}
\begin{tabularx}{\textwidth}{>{\raggedright\arraybackslash}p{0.22\textwidth} >{\centering\arraybackslash}p{0.065\textwidth} >{\centering\arraybackslash}p{0.075\textwidth} >{\centering\arraybackslash}p{0.075\textwidth} X}
\hline
Scenario & $p$ & $R_M^2$ & $\sigma_{\tau,s}$ & Purpose \\
\hline
Null
& $1/2$ & $0$ & $0$
& Compares the methods when no baseline covariate can improve ATE precision. \\
Sparse $M$ signal
& $1/2$ & $0.50$ & $0$
& Tests whether the methods obtain precision gains from four of 60 covariates. \\
Arm-specific cancellation
& $1/2$ & $0$ & $2$
& Makes outcomes predictable within arms even though the relevant score has no covariate signal. \\
Unequal allocation
& $1/4$ & $0.50$ & $1$
& Makes the arm weights consequential while different covariates predict the score and treatment effects. \\
\hline
\end{tabularx}
\end{table}

We study both linear and nonlinear $M_i$ signals. To implement the linear signal and heterogeneity components, we use the following
sparse normalized indices:
\begin{align}
g_s^\circ(X)
&=\frac{X_1+X_2+X_3+X_4}{\sqrt{6.214}}
&&\text{in the sparse-signal scenario},
\nonumber\\
h_s^\circ(X)
&=\frac{X_1+X_2-X_3}{\sqrt{2.82}}
&&\text{in the cancellation scenario},
\nonumber\\
g_s^\circ(X)
&=\frac{X_1+X_2+X_3}{\sqrt{4.38}},
\qquad
h_s^\circ(X)=\frac{X_4-X_5+X_6}{\sqrt{1.98}}
&&\text{under unequal allocation}.
\label{eq:simulation-directions}
\end{align}
The denominators are the theoretical standard deviations of the displayed indices under $\Sigma$.

In the nonlinear versions, we replace each $X_j$ in the indices
$g_s^\circ(X)$ and $h_s^\circ(X)$ with $\phi_j(X)=X_j^2-1$, retaining
the same covariates and coefficient signs. We then rescale each nonzero index to have variance one, using
$\Cov\{\phi_j(X),\phi_k(X)\}=2\Sigma_{jk}^2$.
For example, the quadratic sparse-signal index is
\begin{equation}
g_s^\circ(X)
=
\frac{(X_1^2-1)+(X_2^2-1)+(X_3^2-1)+(X_4^2-1)}
{\sqrt{9.15}}.
\label{eq:simulation-quadratic}
\end{equation}
The signal strength, scale of treatment-effect heterogeneity, allocation
fraction, and number of active covariates remain those in Table~\ref{tab:simulation-dgps}. Note that under this DGP, $\Cov\{X_j,g_s^\circ(X)\}=0$ for every $j$, even when $g_s^\circ(X)$ carries useful signal for $M$. Therefore, the nonlinear learners must discover these relationships from the raw covariates.

\subsection{Results}
\label{subsec:simulation-results}

After constructing each science table, we repeatedly assign exactly $N_1=pN$ units to treatment by complete randomization. We use 1,000 replications in each scenario. For each estimator $\ell$, errors are evaluated against the realized finite-population ATE $\tau_N=N^{-1}\sum_i\{Y_i(1)-Y_i(0)\}$. We report bias, root MSE, and MSE relative to the difference in means:
\begin{equation}
\operatorname{RelativeMSE}_{\ell}
=
\frac{\sum_{r=1}^R(\widehat\tau_{\ell r}-\tau_N)^2}
{\sum_{r=1}^R(\widehat\tau_{\DM,r}-\tau_N)^2}.
\label{eq:simulation-relative-mse}
\end{equation}
A value below one indicates an MSE improvement over no adjustment. We report paired Monte Carlo standard errors (MCSE) and 95\% Monte Carlo intervals for relative MSE. 

\subsubsection{Linear specifications}
\label{subsubsec:simulation-baseline-results}

Table~\ref{tab:simulation-ate-results} and Figure~\ref{fig:simulation-performance} first report performance in the linear specifications. In the null scenario, Joint and Arm-specific boosting both have relative MSE $1.011$. Thus, the learned-score estimators remain close to no adjustment when the covariates contain no useful signal. By contrast, all-$X$ OLS and Lin increase MSE by $19.6\%$ and $24.3\%$, respectively.

When the useful signal is linear and sparse, Joint and Arm-specific boosting attain relative MSEs of $0.617$ and $0.611$, compared with $0.516$ for the oracle. The two learning approaches perform similarly: their paired MSE difference, expressed in units of difference-in-means MSE, is $0.006$, with a 95\% Monte Carlo interval of $[-0.014,0.027]$. All-$X$ OLS and Lin attain $0.624$ and $0.633$, respectively. The unpenalized regressions remain competitive in this linear generating model (albeit with many irrelevant covariates), although their estimated MSEs are slightly higher than those of the boosted-tree procedures.

Under arm-specific cancellation, Joint boosting has relative MSE $1.005$, whereas Arm-specific boosting has slightly worse MSE $1.030$. All-$X$ OLS and Lin have relative MSEs of $1.209$ and $1.078$. 

Under unequal allocation, Joint and Arm-specific boosting again perform similarly, with relative MSEs of $0.663$ and $0.684$. Pooled unweighted boosting attains $0.717$, compared with $0.509$ for the oracle. The paired difference between Joint and Pooled learning is $-0.053$, with a 95\% Monte Carlo interval of $[-0.091,-0.016]$. This comparison illustrates the value of complementary design weighting in the linear specification. All-$X$ OLS and Lin have relative MSEs of $0.887$ and $1.065$, showing that fitting many coefficients can limit the gain from a linear systematic score.

\begin{table}[!htbp]
\centering
\small
\renewcommand{\arraystretch}{1.12}
\caption{Relative MSE in the linear specifications}
\label{tab:simulation-ate-results}
\begin{tabular}{lrrrr}
\hline
\noalign{\vskip 3pt}
Estimator & Null & \shortstack{Sparse\\signal} & Cancellation & \shortstack{Unequal allocation\\with HTE} \\[3pt]
\hline
Difference in means & 1.000 (0.000) & 1.000 (0.000) & 1.000 (0.000) & 1.000 (0.000) \\
All-$X$ OLS & 1.196 (0.027) & 0.624 (0.030) & 1.209 (0.042) & 0.887 (0.044) \\
All-$X$ Lin & 1.243 (0.032) & 0.633 (0.031) & 1.078 (0.032) & 1.065 (0.056) \\
Joint boosted trees & 1.011 (0.002) & 0.617 (0.022) & 1.005 (0.006) & 0.663 (0.023) \\
AIPW boosted trees & 1.011 (0.003) & 0.611 (0.018) & 1.030 (0.017) & 0.684 (0.023) \\
Pooled boosted trees & -- & -- & -- & 0.717 (0.028) \\
DGP oracle & 1.000 (0.000) & 0.516 (0.023) & 1.000 (0.000) & 0.509 (0.022) \\
\hline
\end{tabular}

\begin{minipage}{0.96\textwidth}
\footnotesize
\emph{Notes:} Each entry reports MSE relative to the difference in means, with a paired Monte Carlo standard error in parentheses. All cases use $N=400$, $K=60$, and 1,000 replications. Pooled unweighted learning coincides with Joint learning under balanced allocation and is reported separately only under unequal allocation.
\end{minipage}
\end{table}

\begin{figure}[!htbp]
\centering
\includegraphics[width=0.96\textwidth]{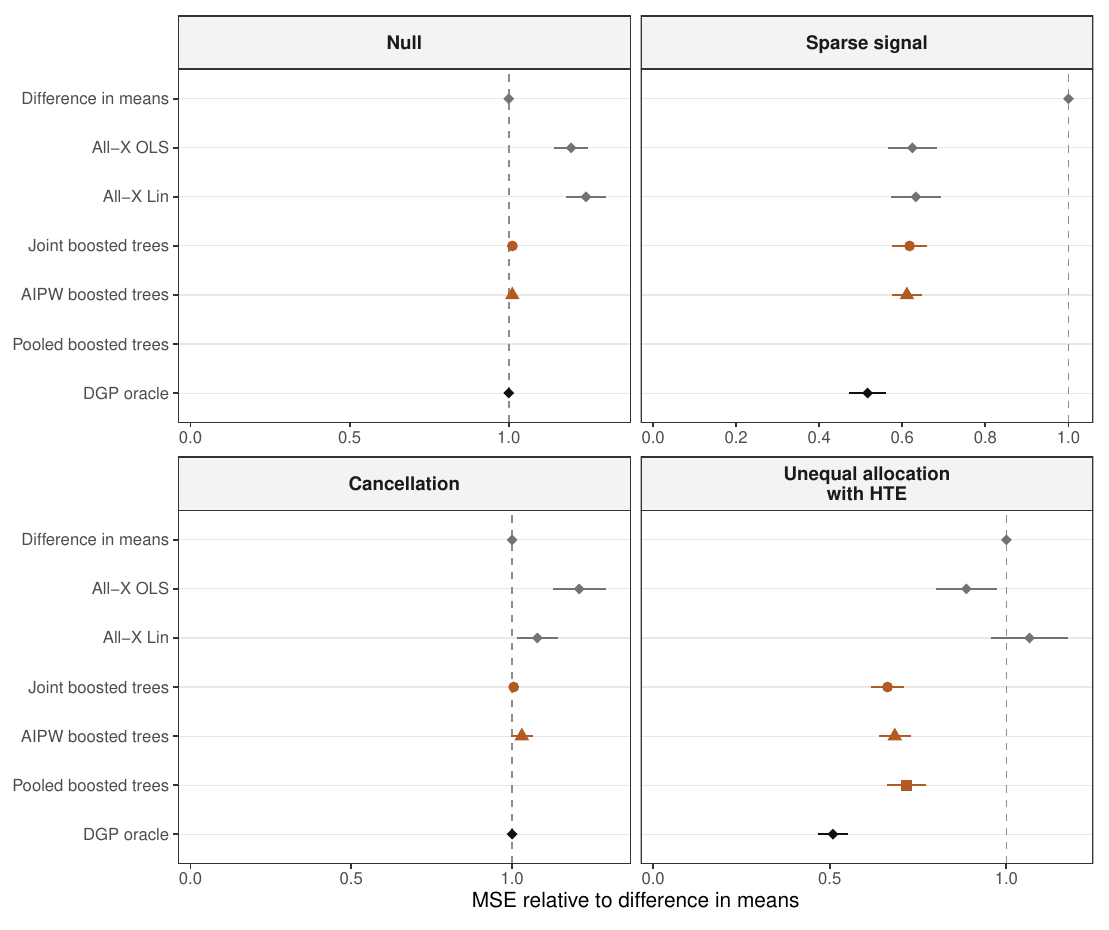}
\caption{MSE relative to the difference in means: linear specifications}
\label{fig:simulation-performance}
\begin{minipage}{0.94\textwidth}
\footnotesize
\emph{Notes:} Points show relative MSE, and horizontal intervals equal the estimate plus or minus $1.96$ paired Monte Carlo standard errors. The dashed vertical line marks the MSE of the unadjusted difference in means.
\end{minipage}
\end{figure}

\subsubsection{Nonlinear specifications}
\label{subsubsec:simulation-nonlinear-results}

Table~\ref{tab:simulation-nonlinear-results} and Figure~\ref{fig:simulation-nonlinear-performance} report the quadratic specifications. The sparse-signal scenario illustrates the benefit of learning a nonlinear score. All-$X$ OLS and Lin have relative MSEs of $1.061$ and $1.114$: linear adjustment using the raw covariates does not recover the systematic quadratic signal. Joint and Arm-specific boosting reduce relative MSE to $0.695$ and $0.806$, respectively, compared with $0.441$ for the oracle. The Joint--Arm-specific difference is $-0.111$, with a 95\% Monte Carlo interval of $[-0.140,-0.082]$. The flexible ML learner thus successfully captures nonlinear predictive signals that are not reflected in population-level linear associations with the raw covariates.

Quadratic cancellation preserves the distinction between outcome prediction and prediction of the relevant score. Joint boosting remains close to no adjustment, with relative MSE $1.040$, whereas Arm-specific boosting attains $1.113$. Their paired difference is $-0.073$, with a 95\% Monte Carlo interval of $[-0.113,-0.034]$. All-$X$ OLS and Lin perform worse, at $1.403$ and $1.478$. Here, the costs of learning adjustment scores outweigh their precision benefits, with a smaller net cost for Joint learning.

The nonlinear unequal-allocation scenario is more demanding. Joint and Arm-specific boosting attain relative MSEs of $0.850$ and $0.895$, while All-$X$ OLS and Lin attain $1.212$ and $2.032$. However, Pooled unweighted boosting has relative MSE $0.808$, slightly below Joint boosting. One explanation is that complementary weighting targets the relevant score, but it also concentrates fitting and tuning on the smaller treatment arm. The resulting estimation cost may offset the benefit of targeting the correct score when nonlinear relationships must be learned from a limited sample. Their paired difference is $0.042$, with a 95\% Monte Carlo interval of $[-0.002,0.08]$. Thus, this result does not establish a finite-sample advantage of weighting and should not be interpreted as a general finding, although it may suggest a direction for future research. The oracle value is $0.507$, leaving a substantial gap for all three feasible learners. 

\begin{table}[!htbp]
\centering
\small
\renewcommand{\arraystretch}{1.12}
\caption{Relative MSE in the quadratic specifications}
\label{tab:simulation-nonlinear-results}
\begin{tabular}{lrrr}
\hline
\noalign{\vskip 3pt}
Estimator & \shortstack{Sparse\\signal} & Cancellation & \shortstack{Unequal allocation\\with HTE} \\[3pt]
\hline
Difference in means & 1.000 (0.000) & 1.000 (0.000) & 1.000 (0.000) \\
All-$X$ OLS & 1.061 (0.029) & 1.403 (0.042) & 1.212 (0.039) \\
All-$X$ Lin & 1.114 (0.033) & 1.478 (0.048) & 2.032 (0.091) \\
Joint boosted trees & 0.695 (0.019) & 1.040 (0.007) & 0.850 (0.022) \\
AIPW boosted trees & 0.806 (0.014) & 1.113 (0.021) & 0.895 (0.015) \\
Pooled boosted trees & -- & -- & 0.808 (0.026) \\
DGP oracle & 0.441 (0.022) & 1.000 (0.000) & 0.507 (0.021) \\
\hline
\end{tabular}

\begin{minipage}{0.96\textwidth}
\footnotesize
\emph{Notes:} Each entry reports MSE relative to the difference in means, with a paired Monte Carlo standard error in parentheses. The signal strength, scale of treatment-effect heterogeneity, and allocation fraction match the corresponding linear specification. The null scenario is unchanged and is not duplicated.
\end{minipage}
\end{table}

\begin{figure}[!htbp]
\centering
\includegraphics[width=0.96\textwidth]{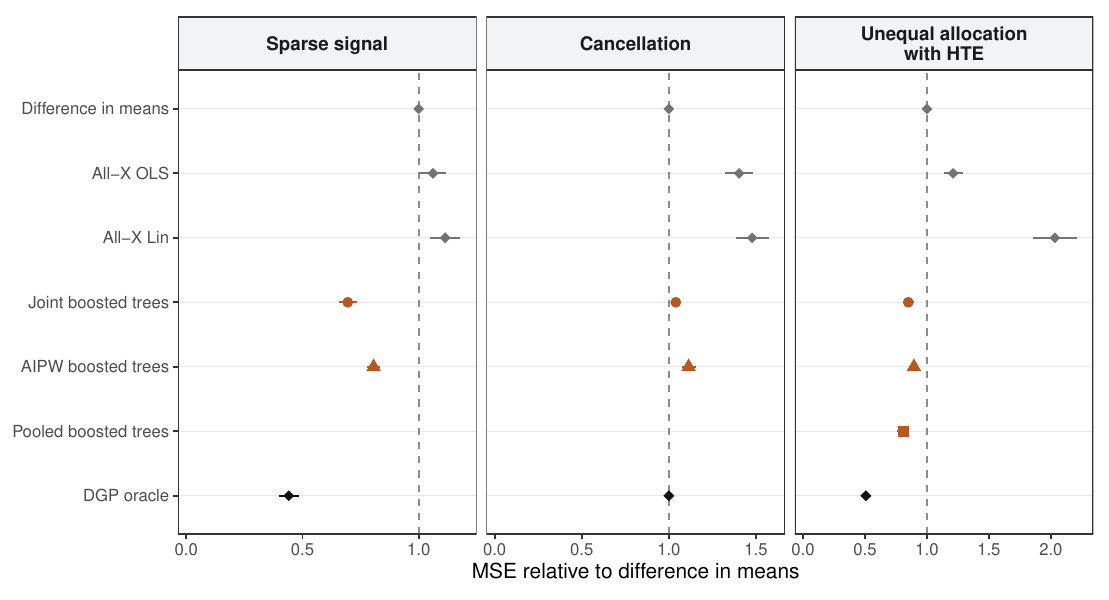}
\caption{MSE relative to the difference in means: quadratic specifications}
\label{fig:simulation-nonlinear-performance}
\begin{minipage}{0.94\textwidth}
\footnotesize
\emph{Notes:} Points show relative MSE, and horizontal intervals equal the estimate plus or minus $1.96$ paired Monte Carlo standard errors. All feasible methods use only the raw covariates.
\end{minipage}
\end{figure}




\subsection{Practical guidance for estimator choice}
\label{subsec:simulation-takeaways}

These simulations support choosing covariates and adjustment methods according to their ability to learn the ATE-relevant score. When the DGP is nonlinear, flexible learning can improve precision even though linear adjustment using the raw covariates fails to recover the signal. When strong outcome relationships cancel in $M$, the benefits of predicting outcomes within arms need not translate into benefits for ATE adjustment. In general, Joint boosting performs similarly to Arm-specific boosting when the useful signal is linear, under both allocation ratios, and better in the cancellation and nonlinear scenarios considered here. All-$X$ OLS performs well when the useful relationship is linear, while its performance deteriorates under the nonlinear specifications. 

\section{Discussion}
\label{sec:prognostic-balance-report}

\subsection{Balance Tables}

Conventional balance tables report treatment and control group means, standardized mean differences, and sometimes balance-test \(p\)-values for each baseline covariate. These quantities are useful for describing the realized assignment and checking implementation (e.g., differential attrition). Our theory implies, however, that they do not provide clear guidance for selecting an ATE adjustment model. A balance statistic measures how differently a variable was distributed across treatment arms; it does not indicate whether the variable predicts the ex-post bias. 

By the same token, the \(p\)-values reported in a balance table may be a poor guide to whether and how to conduct covariate adjustment. A \(p\)-value indicates how surprising the raw imbalance is under the design, not how large the implied ATE correction is in outcome units. A high-variance, weakly prognostic covariate can have a small \(p\)-value, while a highly prognostic score can have a large \(p\)-value because it happened to be close to balanced. Neither \(p\)-value reveals the covariate's ability to predict \(M\).

\paragraph{A Recommended Prognostic Balance Report.}

We recommend retaining the conventional raw-covariate balance table and adding an outcome- and estimand-specific prognostic-score panel. The conventional panel documents the assignment and can reveal coding mistakes, attrition, deviation from the randomization protocol, or differences that readers may regard as substantively important.

Let \(\widehat g_i\) denote a prespecified, externally trained, or cross-fitted prediction of \(M_i\). For each outcome and estimand, report
\begin{align}
\overline{\widehat g}_1
&=
\frac{1}{N_1}\sum_{i=1}^N Z_i\widehat g_i,
&
\overline{\widehat g}_0
&=
\frac{1}{N_0}\sum_{i=1}^N(1-Z_i)\widehat g_i,
\label{eq:score-means-report}\\
\Delta_Z(\widehat g)
&=
\overline{\widehat g}_1-\overline{\widehat g}_0,
&
D(\widehat g)
&=
\frac{\Delta_Z(\widehat g)}
{\sqrt{N/(N_1N_0)}\,S_{\widehat g}}.
\label{eq:score-balance-report}
\end{align}
The unstandardized difference \(\Delta_Z(\widehat g)\) is of primary interest because it indicates the point correction and is measured in outcome units. The standardized difference \(D(\widehat g)\), together with an optional \(p\)-value, answers the separate descriptive question of how unusual the apparent imbalance is under the design. One can also report these statistics separately for each covariate, as illustrated in the application in Appendix Section~\ref{sec:twitter-application}.

\subsection{Practical Guidance}
\label{subsec:practical-guidance}

Researchers should collect pretreatment covariates whenever feasible, prioritizing information that can help predict the design-relevant score $M$. The decomposition
\(M_i=Y_i(0)+(1-p)\tau_i\) provides a useful guide: relevant covariates may predict baseline potential outcomes, treatment-effect heterogeneity, or both. The objective, however, is to predict their weighted combination, rather than either component in isolation, because their associations with a covariate may reinforce or offset one another. Under a balanced design, we can simply collect covariates that are predictive of the outcome. Otherwise, one can collect a broad set of candidate predictors and use machine learning methods to adaptively predict $M_i$.

Because $M$ is unobserved (except under the sharp null), researchers estimate its predictor $g(X)$ using either approach developed above: joint or arm-specific estimation. With many covariates or flexible learners, regularization and cross-fitting can help control overfitting. Under either approach, adjustment subtracts the treatment--control difference in the fitted score from the unadjusted estimate: $\widehat\tau_{\mathrm{adj}}
=\widehat\tau_{\DM}-\Delta_Z(\widehat g)$.

No preliminary balance test is needed to decide whether to adjust or which covariates to include. These decisions should be guided by prediction of $M$, rather than the statistical significance of realized covariate differences. The adjustment already incorporates realized balance through $\Delta_Z(\widehat g)$: when the fitted score is nearly balanced, the point correction is correspondingly small, without requiring a separate testing step. 

\section{Conclusion}

This paper develops a finite-population framework for deciding when and how to adjust for pretreatment covariates after a randomized experiment has been realized. The framework reconciles the three debates in the literature and practice. Ex-post error correction and ex-ante precision improvement are not competing rationales for adjustment; they are two interpretations of the same projection. Using this framework, we derive a joint estimator as well as an arm-specific estimator, with the latter recovering the familiar generalized-regression and AIPW estimators. The central argument is that when constructing or selecting an adjustment score, prognosticity for $M$ is the sole criterion. Raw covariate imbalance, balance-test $p$-values, generic prediction of the observed outcome, and prediction of one potential outcome are not independent adjustment criteria; they are informative only when they help predict $M$. 


The present analysis focuses on the finite-population ATE under complete randomization. Extending the decision framework to stratified, clustered, rerandomized, multi-arm, or repeated-experiment settings, and to estimands other than the ATE, requires deriving the corresponding design- and estimand-specific error score. Further work can also sharpen finite-sample guarantees for data-adaptive score learning. The practical lesson for completely randomized experiments may be stated simply: covariate adjustment should be guided not by which baseline variables happen to be imbalanced, but by which pretreatment information predicts the randomization error relevant to the estimand. Randomization validates the unadjusted estimator across assignments; prognostic adjustment uses baseline information to improve the estimate from the assignment that was actually observed.

\section*{Statement of AI Usage} In preparing this paper, the authors used ChatGPT Sol models for the following tasks: constructing illustrative examples; collecting articles and preparing table 1; assisting with code for generating figures and tables from simulation results; drafting and refining descriptions of the simulation design and results; and checking mathematical proofs.  The final draft was proofread by Claude Opus 5.5 and ChatGPT 6 Pro. All AI-assisted output was carefully reviewed, and where necessary, revised by the authors. The authors take full responsibility for the content of the paper.
\clearpage

\begin{spacing}{0.0}
	\setlength{\bibsep}{10pt}
	\setstretch{0.1}
	\bibliographystyle{apalike}
	\bibliography{literature}
\end{spacing}

\newpage


\clearpage

\appendix
\addcontentsline{toc}{section}{Appendix} 
\part{Supplementary Information} 
\parttoc 

\setcounter{figure}{0}
\setcounter{table}{0}
\setcounter{proposition}{0} 
\renewcommand\thefigure{A.\arabic{figure}}
\renewcommand\thetable{A.\arabic{table}}
\renewcommand\theproposition{\thesection.\arabic{proposition}}

\clearpage
\onehalfspacing
\setcounter{page}{1}

\section{Design-based calibration of the correction}
Write
\(
\widehat\tau_a
=
\widehat\tau_{\mathrm{DM}}-a(Y,Z,X).
\)
Let $\mathcal Y=\{Y_i(0),Y_i(1)\}_{i=1}^N$ denote the
potential-outcome schedule, and write $Y_{\mathcal Y}(Z)$ for
the observed outcome vector, whose $i$th entry is
$Z_iY_i(1)+(1-Z_i)Y_i(0)$. For a specified correction rule $a$,
define its randomization mean by
\(
\mu_a(\mathcal Y;X)
=
\mathbb E_Z\!\left[
    a\{Y_{\mathcal Y}(Z),Z,X\}
\right].
\)
Since the difference in means is design-unbiased,
\(
\mathbb E_Z(\widehat\tau_a)-\tau
=
-\mu_a(\mathcal Y;X).
\)
Thus, subtracting a correction preserves design unbiasedness
when its randomization mean is zero. More generally, the
correction can be centered by subtracting
$\mu_a(\mathcal Y;X)$, provided that this quantity is available.
The following result characterizes when this centering
constant can be recovered exactly from the observed experiment
without restrictions on the potential outcomes.

\begin{proposition}[Exact design calibration]
\label{prop:exact-design-calibration}
Fix $X$ and consider complete randomization with
$1\leq N_1\leq N-1$. Allow the potential-outcome schedule
$\mathcal Y$ to range over $\mathbb R^{2N}$.
Let $a(Y,Z,X)$ be a specified, finite-valued correction rule.
The following statements are equivalent:
\begin{enumerate}
    \item[(i)]
    There exists a function $C_a(Y,Z,X)$, not depending on
    unobserved potential outcomes, such that, for every
    $\mathcal Y$,
    \(
    C_a\{Y_{\mathcal Y}(Z),Z,X\}
    =
    \mu_a(\mathcal Y;X)\) almost surely under randomization.

    \item[(ii)]
    The randomization mean $\mu_a(\mathcal Y;X)$ is constant
    across potential-outcome schedules.
\end{enumerate}
When these conditions hold, the centering constant is
\[
c_a(X)
=
\mathbb E_Z\{a(\mathbf 0,Z,X)\}
=
\binom{N}{N_1}^{-1}
\sum_{\substack{
    z\in\{0,1\}^N\\
    \sum_i z_i=N_1
}}
a(\mathbf 0,z,X),
\]
where $\mathbf 0$ is the $N$-vector of zeros. Consequently,
\(
\widehat\tau_a^{\mathrm{cal}}
=
\widehat\tau_{\mathrm{DM}}
-
\{a(Y,Z,X)-c_a(X)\}
\)
is design-unbiased for $\tau$ for every $\mathcal Y$. 
\end{proposition}

In particular, every specified baseline-information rule
$b(Z,X)$ admits this calibration, with
\(
c_b(X)=\mathbb E_Z\{b(Z,X)\}
\)
determined by $X$, the rule, and the assignment design.

\begin{proof}
Suppose (i) holds. Fix a unit $i$ and a treatment state
$t\in\{0,1\}$. Consider two schedules $\mathcal Y$ and
$\mathcal Y'$ that differ only in the value of $Y_i(t)$.
Since both treatment arms are nonempty, there is an
admissible assignment $z$ with $z_i=1-t$. This assignment
has strictly positive probability under complete
randomization. At this assignment, the changed potential
outcome is unobserved, so
\[
Y_{\mathcal Y}(z)=Y_{\mathcal Y'}(z).
\]

Because the randomization support is finite and every
admissible assignment has positive probability, the
almost-sure equality in (i) must hold at this particular
$z$ under both schedules. Therefore,
\[
\begin{aligned}
\mu_a(\mathcal Y;X)
&=
C_a\{Y_{\mathcal Y}(z),z,X\}\\
&=
C_a\{Y_{\mathcal Y'}(z),z,X\}\\
&=
\mu_a(\mathcal Y';X).
\end{aligned}
\]
Thus, changing any one potential outcome leaves $\mu_a$
unchanged. Any two unrestricted schedules can be connected
by changing their $2N$ potential outcomes one at a time.
It follows that $\mu_a$ is constant across all schedules,
establishing (ii).

Conversely, suppose (ii) holds. Evaluate $\mu_a$ at the
schedule in which all potential outcomes are zero.
Its observed outcome vector is $\mathbf 0$ under every
assignment, so the common value is
\[
\mu_a(\mathcal Y;X)
=
\mathbb E_Z\{a(\mathbf 0,Z,X)\}
=
c_a(X).
\]
Setting $C_a(Y,Z,X)=c_a(X)$ proves (i). Moreover,
\(
\mathbb E_Z\{a(Y,Z,X)-c_a(X)\}=0,
\)
which, together with the design unbiasedness of
$\widehat\tau_{\mathrm{DM}}$, proves the asserted
unbiasedness of $\widehat\tau_a^{\mathrm{cal}}$.

Finally, for a rule $b(Z,X)$, its randomization mean does
not depend on $\mathcal Y$, so (ii) holds automatically.
\end{proof}

\section{Sample ATE-relevant Prognosticity}
\label{sec:sample-prognosticity-fit-imbalance}
\label{subsubsec:sample-r2-ate}

Section~\ref{sec:prog_imbalance} examines population-level ATE-relevant prognosticity (and imbalance). We now turn to their sample counterparts. For any fitted prediction vector $\widehat g$, define
$\widehat\tau(\widehat g)=\widehat\tau_{\DM}-\Delta_Z(\widehat g)$.

\paragraph{Sample outcome mean prognosticity.}
First consider the arm-specific estimator in Section~5.2. Given the
same fitted arm means $\widehat m_0$ and $\widehat m_1$, define
\begin{equation}
\begin{aligned}
\widehat\tau(X)&=\widehat m_1(X)-\widehat m_0(X),\\
\widehat g(X)&=q\widehat m_1(X)+p\widehat m_0(X),\\
\widehat m_Y^{\mathrm{mix}}(X)
&=p\widehat m_1(X)+q\widehat m_0(X).
\end{aligned}
\label{eq:sample-constructed-means}
\end{equation}
Here $\widehat\tau(X)$ is the fitted treatment-effect function, and
$\widehat m_Y^{\mathrm{mix}}$ estimates the pooled population mean
by combining the same arm-specific predictions. 

\begin{proposition}
\label{prop:sample-mean-imbalance}
Suppose $0<p<1$ and $q=1-p$. For the fitted functions in
\eqref{eq:sample-constructed-means},
\begin{align}
\widehat g(X)
&=\widehat m_0(X)+q\widehat\tau(X),
\label{eq:sample-control-score-relation}\\
\widehat m_Y^{\mathrm{mix}}(X)
&=\widehat g(X)+(2p-1)\widehat\tau(X).
\label{eq:sample-pooled-score-relation}
\end{align}
Consequently:
\begin{enumerate}
\item $\widehat m_0(X)$ and $\widehat g(X)$ differ only by an
additive constant across the observed units if and only if
$\widehat\tau(X)$ is constant across those units.
\item $\widehat m_Y^{\mathrm{mix}}(X)$ and $\widehat g(X)$ differ
only by an additive constant if and only if $p=1/2$ or
$\widehat\tau(X)$ is constant. When $p=1/2$, the two fitted
functions are exactly equal.
\end{enumerate}
These statements also apply to cross-fitted prediction vectors,
with constancy interpreted across their unit-level predictions.
\end{proposition}

This is the sample counterpart of Proposition~\ref{pr:mean-targets}
for estimates constructed from the same arm-specific fits. The direct-learning estimator has a separate exact relationship to ordinary
pooled prediction. Write
\[
\widehat L_z(f)
=\frac1{N_z}\sum_{i:Z_i=z}\{Y_i-f(X_i)\}^2.
\]
On the same function class, an ordinary pooled fit and the direct
score fit minimize, respectively,
\begin{equation}
\begin{aligned}
\widehat m_Y^{\mathrm{pool}}
&\in\argmin_f\{p\widehat L_1(f)+q\widehat L_0(f)\},\\
\widehat g_{\mathrm{dir}}
&\in\argmin_f\{q\widehat L_1(f)+p\widehat L_0(f)\}.
\end{aligned}
\label{eq:sample-pooled-direct-criteria}
\end{equation}
Under balanced allocation these objectives coincide, so
$\widehat m_Y^{\mathrm{pool}}=\widehat g_{\mathrm{dir}}$ when the
same fitting rule selects a minimizer. The equality also holds with
an identical penalty added to both criteria.

\begin{proof}[Proof of Proposition \ref{prop:sample-mean-imbalance}]
Substituting
$\widehat m_1(X)=\widehat m_0(X)+\widehat\tau(X)$ gives
\[
\widehat g(X)=\widehat m_0(X)+q\widehat\tau(X),
\qquad
\widehat m_Y^{\mathrm{mix}}(X)
=\widehat m_0(X)+p\widehat\tau(X).
\]
Subtracting yields the second identity. Because $q>0$, the first
difference is constant exactly when $\widehat\tau(X)$ is constant.
The second difference is identically zero when $p=1/2$; otherwise
it is constant exactly when $\widehat\tau(X)$ is constant. The
argument applies to numerical prediction vectors regardless of how
they were fitted.
\end{proof}

\paragraph{Sample regression fit.}
For the same observed outcomes and
$\operatorname{TSS}=\sum_i(Y_i-\overline Y)^2>0$, ordinary sample
$R^2=1-\operatorname{RSS}/\operatorname{TSS}$ is a normalization of
squared prediction loss. Its increase therefore measures improved
training fit to $Y$. To see how this differs from improved ATE
accuracy, still consider OLS regressions on $(1,Z,X_1)$ and
$(1,Z,X_1,X_2)$, with treatment coefficients $\widehat\tau_1$ and
$\widehat\tau_2$ and fit measures $\widehat R_1^2$ and
$\widehat R_2^2$. Here $X_1$ may be a vector and $X_2$ is one added
regressor. Let $\widehat b_2$ be its coefficient in the longer outcome
regression, and define the auxiliary regression
\begin{equation}
X_{i2}=\widehat\pi_0+\widehat\pi_ZZ_i
       +X_{i1}^{\mathsf T}\widehat\pi_X+r_i.
\label{eq:sample-r2-auxiliary}
\end{equation}

\begin{proposition}
\label{prop:sample-r2-fwl}
Suppose the longer design matrix has full column rank and
$\operatorname{TSS}>0$. Then
\begin{align}
\widehat R_2^2-\widehat R_1^2
&=\frac{\widehat b_2^{\,2}\sum_i r_i^2}{\operatorname{TSS}},
\label{eq:sample-r2-increment}\\
\widehat\tau_2-\widehat\tau_1
&=-\widehat\pi_Z\widehat b_2,
\label{eq:sample-r2-treatment-change}\\
(\widehat\tau_1-\tau)^2-(\widehat\tau_2-\tau)^2
&=\widehat\pi_Z\widehat b_2
 \{2(\widehat\tau_1-\tau)-\widehat\pi_Z\widehat b_2\}.
\label{eq:sample-r2-error-change}
\end{align}
Consequently, sample $R^2$ increases if and only if
$\widehat b_2\ne0$, whereas the treatment coefficient changes if
and only if $\widehat\pi_Z\widehat b_2\ne0$.
The realized squared ATE error decreases if and only if
\begin{equation}
\widehat\pi_Z\widehat b_2(\widehat\tau_1-\tau)>0
\quad\text{and}\quad
|\widehat\pi_Z\widehat b_2|<2|\widehat\tau_1-\tau|.
\label{eq:sample-r2-error-improvement-iff}
\end{equation}
\end{proposition}

The auxiliary normal equations give
\begin{equation}
\widehat\pi_Z
=\Delta_Z(X_2)-\widehat\pi_X^{\mathsf T}\Delta_Z(X_1),
\label{eq:sample-r2-adjusted-imbalance}
\end{equation}
so the coefficient change combines fitted outcome association with
imbalance after accounting for the existing covariates. Accuracy
requires this correction also to have the appropriate direction and
magnitude relative to the unobserved error $\widehat\tau_1-\tau$.

\begin{proof}
Let $\widehat u_{ij}$ be the outcome residuals from specification $j$.
The auxiliary residual $r$ is orthogonal to $(1,Z,X_1)$, and full rank
implies $\sum_i r_i^2>0$. Substituting the auxiliary regression into
the longer outcome regression and using the OLS normal equations gives
\[
\widehat\tau_1
=\widehat\tau_2+\widehat b_2\widehat\pi_Z,
\qquad
\widehat u_{i1}=\widehat b_2r_i+\widehat u_{i2},
\qquad
\sum_i r_i\widehat u_{i2}=0.
\]
In particular,
\begin{equation}
\widehat b_2=\frac{\sum_i r_iY_i}{\sum_i r_i^2},
\qquad
\operatorname{RSS}_1-\operatorname{RSS}_2
=\widehat b_2^{\,2}\sum_i r_i^2.
\label{eq:sample-r2-partial-slope}
\end{equation}
Dividing the residual-sum-of-squares difference by
$\operatorname{TSS}$ proves \eqref{eq:sample-r2-increment}.
Expanding the square of
$\widehat\tau_2-\tau=(\widehat\tau_1-\tau)
-\widehat\pi_Z\widehat b_2$
proves \eqref{eq:sample-r2-error-change}.
Its right-hand side is positive exactly under
\eqref{eq:sample-r2-error-improvement-iff}.
The remaining equivalences follow from
$\sum_i r_i^2>0$ and $\operatorname{TSS}>0$.
Taking treatment--control mean differences in the auxiliary regression
gives \eqref{eq:sample-r2-adjusted-imbalance}, because the intercept
and treatment normal equations imply $\Delta_Z(r)=0$.

For a counterexample to the converse, take four units with
$Z=(0,0,1,1)$, $X_2=(0,1,1,2)$,
$Y(0)=(0,1,0,1)$, and $Y(1)=Y(0)+2$, with no $X_1$ regressors.
Then $p=1/2$, $\tau=2$, and $Y=(0,1,2,3)$.
The regression on $(1,Z)$ gives $\widehat\tau_1=2$ and
$\widehat R_1^2=4/5$; the regression on $(1,Z,X_2)$ fits
$Y=Z+X_2$ exactly, giving $\widehat\tau_2=1$ and
$\widehat R_2^2=1$. Hence better fit increases squared ATE error
from $0$ to $1$, despite balanced allocation and constant treatment
effects.
\end{proof}

\section{Proofs}
We define $q=1-p$ throughout the proof.

\subsection{Proof of Proposition \ref{thm:conditional-decision}}

\begin{proof}
For each feasible assignment $z\in\calZ_{N_1}$, define
\(
e_0(z):=\widehat\tau_{\DM}(z)-\tau.
\)
Thus, $e_0(Z)$ is the randomization error of the unadjusted
estimator, whereas $e_0(z)$ is its value under the fixed assignment
$z$. Write $a=a(Z,X)$ and recall that
\[
\widehat\tau_a-\tau=e_0(Z)-a,
\qquad
B_{\mathcal F_X}
=\E_Z\{e_0(Z)\mid\mathcal F_X\}.
\]
By the definition of conditional expectation,
\(
\E_Z\{e_0(Z)-B_{\mathcal F_X}\mid\mathcal F_X\}=0.
\)
Both $a$ and $B_{\mathcal F_X}$ are $\mathcal F_X$-measurable.
Consequently,
\begin{align*}
\MSE_Z(\widehat\tau_a\mid\mathcal F_X)
&=\E_Z\bigl[\{e_0(Z)-a\}^2\mid\mathcal F_X\bigr]\\
&=\E_Z\bigl[\{e_0(Z)-B_{\mathcal F_X}\}^2
\mid\mathcal F_X\bigr]
+(B_{\mathcal F_X}-a)^2\\
&\quad+2(B_{\mathcal F_X}-a)
\E_Z\{e_0(Z)-B_{\mathcal F_X}\mid\mathcal F_X\}\\
&=\Var_Z\{e_0(Z)\mid\mathcal F_X\}
+(B_{\mathcal F_X}-a)^2.
\end{align*}
Setting $a=0$ and subtracting gives
Equation~\eqref{eq:value-bias1}.

Complete randomization makes the difference-in-means estimator
unbiased, so $\E_Z\{e_0(Z)\}=0$. The law of iterated
expectations therefore gives
\[
\E_Z(B_{\mathcal F_X})
=\E_Z\bigl[\E_Z\{e_0(Z)\mid\mathcal F_X\}\bigr]
=0.
\]
Taking unconditional expectations in
Equation~\eqref{eq:value-bias1} yields
\begin{align*}
\MSE_Z(\widehat\tau_{\DM})
-\MSE_Z(\widehat\tau_a)
&=\E_Z(B_{\mathcal F_X}^2)
-\E_Z\{(B_{\mathcal F_X}-a)^2\}\\
&=\Var_Z(B_{\mathcal F_X})
-\E_Z\{(B_{\mathcal F_X}-a)^2\},
\end{align*}
which proves Equation~\eqref{eq:value-bias}.

At each attainable covariate allocation, the conditional risk is
uniquely minimized by setting $a=B_{\mathcal F_X}$ at that
allocation. The same rule uniquely minimizes unconditional risk.
Indeed, complete randomization assigns strictly positive
probability to every feasible assignment, so
\[
\E_Z\{(B_{\mathcal F_X}-a)^2\}=0
\quad\Longleftrightarrow\quad
a(z,X)=B_{\mathcal F_X}(z)
\quad\text{for every }z\in\calZ_{N_1}.
\]
\end{proof}

\subsection{Proof of Proposition \ref{thm:M-identity}}


\begin{proof}
Let $\overline Y_z(t)$ denote the mean of $Y_i(t)$ among units with $Z_i=z$, and let $\overline Y(t)$ denote its population mean. Partitioning the population into treatment and control gives
\[
\overline Y(t)=p\overline Y_1(t)+(1-p)\overline Y_0(t),
\qquad t=0,1.
\]
Consequently,
\begin{align*}
\overline Y_1(1)-\overline Y(1)
&=(1-p)\{\overline Y_1(1)-\overline Y_0(1)\},\\
\overline Y(0)-\overline Y_0(0)
&=p\{\overline Y_1(0)-\overline Y_0(0)\}.
\end{align*}
Subtracting the population ATE from the difference in means and using these two identities yields
\begin{align*}
\widehat\tau_{\DM}-\tau
&=\{\overline Y_1(1)-\overline Y(1)\}
+\{\overline Y(0)-\overline Y_0(0)\}\\
&=(1-p)\Delta_Z\{Y(1)\}+p\Delta_Z\{Y(0)\}\\
&=\Delta_Z\{(1-p)Y(1)+pY(0)\}
=\Delta_Z(M).
\end{align*}
The last equality uses the linearity of $\Delta_Z$. 
\end{proof}

\subsection{Proof of Proposition \ref{lem:FX-score-orthogonality}}
\begin{proof}
Let $\mathcal H$ be the class of measurable scores with $\E_N\{g(X_i)^2\}<\infty$, identifying functions that agree under the empirical distribution of $X_i$. By definition of $g_M$, we have $\E_N\{M-g_M(X)\mid X\}=0$, which implies
\begin{equation}
\E_N[\{M-g_M(X)\}h(X)]=0,
\qquad h\in\mathcal H.
\label{eq:section3-empirical-orthogonality}
\end{equation}

Write $R_i=M_i-g_M(X_i)$. Equation~\eqref{eq:section3-empirical-orthogonality} gives $\E_N\{R_ih(X_i)\}=0$ for every $h\in\mathcal H$.

Complete randomization is invariant to permutations of unit labels.
If $X_i=X_k$, interchanging the assignments of units $i$ and $k$
also leaves $\nu_Z$ unchanged. Therefore
$\E_Z(Z_i\mid\mathcal F_X)=\E_Z(Z_k\mid\mathcal F_X)$. It follows that, for each realization of $\mathcal F_X$, there is a measurable function $\pi_Z:\mathcal X\to[0,1]$ such that
\[
\E_Z(Z_i\mid\mathcal F_X)=\pi_Z(X_i),\qquad i=1,\ldots,N.
\]
Using $N_1=Np$ and $N_0=N(1-p)$, we can write
\[
\Delta_Z(R)=\frac1{Np(1-p)}\sum_{i=1}^N(Z_i-p)R_i.
\]
Hence
\begin{align*}
\E_Z\{\Delta_Z(R)\mid\mathcal F_X\}
&=\frac1{Np(1-p)}\sum_{i=1}^N\{\pi_Z(X_i)-p\}R_i\\
&=\frac1{p(1-p)}\E_N[\{\pi_Z(X_i)-p\}R_i]\\
&=0.
\end{align*}
The last equality applies $\E_N\{R_ih(X_i)\}=0$.
Because $\pi_Z$ is bounded, this function belongs to $\mathcal H$.
Also,
$\Delta_Z\{g_M(X)\}=\int g_M(x)\,\nu_Z(dx)$ is
$\mathcal F_X$-measurable. Thus
\[
\E_Z(e_0\mid\mathcal F_X)
=\Delta_Z\{g_M(X)\}
+\E_Z\{\Delta_Z(R)\mid\mathcal F_X\}
=\Delta_Z\{g_M(X)\}.
\]
\end{proof}

\subsection{Proof of Proposition~\ref{pr:mean-targets}}
\begin{proof}
Recall that $q=1-p$ and $0<p<1$. By the definitions of
$g_M$, $m_Y$, and $\tau(x)$,
\[
\begin{aligned}
g_M(x)
&=q\,m_1(x)+p\,m_0(x)\\
&=q\{m_0(x)+\tau(x)\}+p\,m_0(x)\\
&=m_0(x)+q\,\tau(x).
\end{aligned}
\]
Similarly,
\[
\begin{aligned}
m_Y(x)-g_M(x)
&=\{p\,m_1(x)+q\,m_0(x)\}
  -\{q\,m_1(x)+p\,m_0(x)\}\\
&=(p-q)\{m_1(x)-m_0(x)\}\\
&=(2p-1)\tau(x).
\end{aligned}
\]
These establish the two identities in the proposition.

\medskip
\noindent\emph{Part (1).} The sufficiency is straightforward. Conversely, suppose that $m_0(X)$, $m_Y(X)$, and $g_M(X)$
differ only by additive constants. In particular,
$g_M(X)-m_0(X)$ is constant. Since
\[
g_M(X)-m_0(X)=q\,\tau(X)
\]
and $q>0$, it follows that $\tau(X)$ is constant. This
proves both directions of part~(1).

\medskip
\noindent\emph{Part (2).} The sufficiency is straightforward. Conversely, suppose that $m_Y(X)-g_M(X)$ is constant.
If $p=1/2$, allocation is balanced and the stated condition
holds. If $p\ne1/2$, then
\[
\tau(X)
=
\frac{m_Y(X)-g_M(X)}{2p-1}
\]
is constant because its numerator is constant and its
denominator is a nonzero constant. Hence agreement up to
an additive constant implies that allocation is balanced
or $\tau(X)$ is constant. This proves both directions of
part~(2).
\end{proof}

\subsection{Proof of Proposition~\ref{pr:r2-gains}}

\begin{proof}
By construction, $\widetilde X_2$ is the residual from
projecting $X_2$ on $(1,X_1)$ under finite-population
averaging. Thus
\[
\E_N(\widetilde X_2)=0,
\qquad
\E_N\{\widetilde X_2h(X)\}=0
\quad\text{for every }h\in\mathcal G_1.
\]
Adding $X_2$ to the regressors in $\mathcal G_1$ is
therefore equivalent to adding the orthogonal regressor
$\widetilde X_2$.

\medskip
\noindent\emph{Control-outcome regression fit.}
The partial coefficient of $X_2$ in the finite-population
projection of $Y(0)$ on $(1,X_1,X_2)$ is
\[
b_0
=
\frac{\E_N\{\widetilde X_2Y(0)\}}
     {\E_N(\widetilde X_2^2)}.
\]
Consequently, adding $X_2$ changes the fitted value by
$b_0\widetilde X_2$. Orthogonality of the nested
projections implies that the reduction in mean-squared
prediction loss is
\(
\E_N\{(b_0\widetilde X_2)^2\}
=
b_0^2\Var_N(\widetilde X_2).
\)
Dividing by the total variance $\Var_N\{Y(0)\}$ gives
\[
R_{0,2}^2-R_{0,1}^2
=
\frac{\Var_N(\widetilde X_2)}
     {\Var_N\{Y(0)\}}b_0^2,
\]
which proves \eqref{pr:r2-control-gain}.

\medskip
\noindent\emph{Pooled-outcome regression fit.}
Under the joint law of a uniformly selected unit and
randomized treatment assignment,
\[
\E(Z\widetilde X_2)
=
p\,\E_N(\widetilde X_2)
=
0.
\]
Hence $\widetilde X_2$ is orthogonal to $Z$ as well as
to $(1,X_1)$. It is therefore also the residual from
the population projection of $X_2$ on $(1,Z,X_1)$.
The partial coefficient of $X_2$ in the pooled-outcome
regression is
\[
\begin{aligned}
\frac{\E(\widetilde X_2Y)}
     {\E(\widetilde X_2^2)}
&=
\frac{
p\,\E_N\{\widetilde X_2Y(1)\}
+
q\,\E_N\{\widetilde X_2Y(0)\}
}{
\E_N(\widetilde X_2^2)
}
=pb_1+qb_0.
\end{aligned}
\]
Adding $X_2$ therefore changes the fitted value by
$(pb_1+qb_0)\widetilde X_2$. By projection orthogonality,
the resulting reduction in mean-squared prediction loss is
\(
\Var_N(\widetilde X_2)(pb_1+qb_0)^2.
\)
Dividing by $\Var(Y)$ gives
\[
R_2^2-R_1^2
=
\frac{\Var_N(\widetilde X_2)}
     {\Var(Y)}(pb_1+qb_0)^2,
\]
which proves \eqref{pr:r2-outcome-gain}.

\medskip
\noindent\emph{Oracle ATE mean-squared error.}
Because $M=qY(1)+pY(0)$, the partial coefficient of
$X_2$ in the finite-population projection of $M$ is
\[
\frac{\E_N(\widetilde X_2M)}
     {\E_N(\widetilde X_2^2)}
=
qb_1+pb_0.
\]
Thus the two oracle scores satisfy
\[
g_2(X)-g_1(X)
=
(qb_1+pb_0)\widetilde X_2.
\]
The projection residual $M-g_2(X)$ is orthogonal to
$\widetilde X_2$. Moreover, both score classes contain
an intercept, so both projection residuals have mean
zero. It follows that
\[
\begin{aligned}
&\Var_N\{M-g_1(X)\}
-\Var_N\{M-g_2(X)\}\\
&\qquad=
\Var_N\{g_2(X)-g_1(X)\}\\
&\qquad=
\Var_N(\widetilde X_2)(qb_1+pb_0)^2.
\end{aligned}
\]

For any fixed finite-population variable $A$, complete
randomization gives
\[
\E_Z\{\Delta_Z(A)\}=0,
\qquad
\Var_Z\{\Delta_Z(A)\}
=
\frac{\Var_N(A)}{(N-1)pq}.
\]
These identities follow from
$\E_Z(Z_i)=p$, $\Var_Z(Z_i)=pq$, and
$\Cov_Z(Z_i,Z_j)=-pq/(N-1)$ for $i\ne j$,
with $\Var_N$ defined under empirical averaging.

Since each oracle score $g_j$ is fixed across assignments,
the exact randomization-error identity implies
\[
\widehat\tau(g_j)-\tau
=
\Delta_Z\{M-g_j(X)\}.
\]
Consequently,
\[
\MSE_Z\{\widehat\tau(g_j)\}
=
\frac{\Var_N\{M-g_j(X)\}}{(N-1)pq}.
\]
Combining this expression with the preceding
projection-variance identity yields
\[
\MSE_Z\{\widehat\tau(g_1)\}
-\MSE_Z\{\widehat\tau(g_2)\}
=
\frac{\Var_N(\widetilde X_2)}
     {(N-1)pq}(qb_1+pb_0)^2,
\]
which proves \eqref{pr:r2-mse-gain}.

\medskip
\noindent\emph{Part (1).}
All factors multiplying the squared partial slopes in
\eqref{pr:r2-outcome-gain} and
\eqref{pr:r2-mse-gain} are strictly positive. Therefore,
\[
R_2^2>R_1^2
\quad\Longleftrightarrow\quad
pb_1+qb_0\ne0,
\]
and
\[
\MSE_Z\{\widehat\tau(g_2)\}
<
\MSE_Z\{\widehat\tau(g_1)\}
\quad\Longleftrightarrow\quad
qb_1+pb_0\ne0.
\]

If $p=1/2$, then $q=1/2$ and
\(
pb_1+qb_0
=
qb_1+pb_0
=
\frac{b_1+b_0}{2}.
\)
Thus the two partial slopes are simultaneously zero
or nonzero for every admissible pair $(b_1,b_0)$.
This proves that balanced allocation is sufficient for the stated equivalence to hold for every pair.

To prove necessity, suppose $p\ne1/2$. Holding the
covariates fixed, consider the potential outcomes
\(Y_i(1)=p\,\widetilde X_{i2},\) and 
\(Y_i(0)=-q\,\widetilde X_{i2}.\)
Their partial slopes are
\(b_1=p\), and
\(b_0=-q.\)
This construction satisfies the variance assumptions:
\[
\Var_N\{Y(0)\}
=
q^2\Var_N(\widetilde X_2)>0,
\]
and, because both potential outcomes have empirical
mean zero,
\[
\Var(Y)
=
(p^3+q^3)\Var_N(\widetilde X_2)>0.
\]
However,
\[
pb_1+qb_0
=
p^2-q^2
=
2p-1
\ne0,
\qquad
qb_1+pb_0
=
qp-pq
=
0.
\]
Hence pooled-outcome $R^2$ strictly increases while
oracle ATE MSE is unchanged. The stated equivalence
therefore fails for an admissible pair whenever
$p\ne1/2$. Together with sufficiency, this establishes
part~(1).

\medskip
\noindent\emph{Part (2) is straightforward.} 

\end{proof}

\subsection{Proof of Proposition \ref{prop:general-allocation-imbalance}}

\begin{proof}
Fix a feasible assignment $z$.
For any fixed function $h$ on the recorded covariate support,
\begin{align}
 \Delta_z\{h(X)\}
 &=\frac{\E_N\{z\,h(X)\}}p
   -\frac{\E_N\{(1-z)h(X)\}}q\notag\\
 &=\frac{\E_N\{(z-p)h(X)\}}{pq}\notag\\
 &=\frac{\E_N[\{\pi_z(X)-p\}h(X)]}{pq}\notag\\
 &=\frac{\Cov_N\{h(X),\pi_z(X)\}}{pq}.
 \label{eq:imbalance-proof-general-covariance}
\end{align}
The third equality follows from the defining property of conditional
expectation:
\[
 \E_N\{z\,h(X)\}
 =\E_N[\E_N(z\mid X)h(X)]
 =\E_N\{\pi_z(X)h(X)\}.
\]
The last equality in
\eqref{eq:imbalance-proof-general-covariance} uses
$\E_N\{\pi_z(X)\}=\E_N(z)=p$.
Taking $h(X)=X_k$ for each covariate coordinate proves
\eqref{eq:imbalance-raw-allocation-covariance}.
Taking $h=g_M$ and using
$B_{\mathcal F_X}(z)=\Delta_z\{g_M(X)\}$
proves \eqref{eq:imbalance-score-allocation-covariance}.
Because $pq>0$, the stated zero-imbalance equivalences follow immediately.
\end{proof}

\subsection{Proof of Proposition~\ref{prop:affine-prognostic-imbalance}}

\begin{proof}
If $g_M(X_i)=a+\gamma^{\mathsf T}X_i$ for every $i$, constants
cancel in treatment--control differences, so
$B_{\mathcal{F_X}}(z)=\gamma^{\mathsf T}\Delta_z(X)$ for every feasible assignment.

Conversely, assume this fixed linear representation holds for every
$z\in\mathcal Z_{N_1}$, and set
$h_i=g_M(X_i)-\gamma^{\mathsf T}X_i$. Then $\Delta_z(h)=0$ for
every such $z$. Take any distinct units $i,k$. Because $0<N_1<N$,
there is a feasible assignment with $i$ treated and $k$ untreated.
Swap their treatment statuses while leaving every other status fixed.
The change in $\Delta_z(h)$ is
\[
 \left(\frac1{N_1}+\frac1{N_0}\right)(h_k-h_i).
\]
It must be zero, so $h_i=h_k$. Since the pair was arbitrary, all
$h_i$ equal a common constant $a$, establishing the claimed affine
representation on the finite support.
\end{proof}

\subsection{Proof of Proposition~\ref{thm:observable-loss}}

\begin{proof}

\noindent

For each unit \(i\), recall that
\(
M_i=qY_i(1)+pY_i(0).
\)
Expanding the design-weighted squared-error loss gives
\begin{align*}
&q\{Y_i(1)-g_i\}^2
+
p\{Y_i(0)-g_i\}^2\\
&\quad=
qY_i(1)^2+pY_i(0)^2
-
2g_i\{qY_i(1)+pY_i(0)\}
+
g_i^2\\
&\quad=
\{M_i-g_i\}^2
+
\left[
qY_i(1)^2+pY_i(0)^2-M_i^2
\right].
\end{align*}
The remaining term satisfies the weighted-variance identity
\begin{align*}
&qY_i(1)^2+pY_i(0)^2
-
\{qY_i(1)+pY_i(0)\}^2\\
&\quad=
pq\{Y_i(1)-Y_i(0)\}^2.
\end{align*}
Consequently, unit by unit,
\begin{equation}
q\{Y_i(1)-g_i\}^2
+
p\{Y_i(0)-g_i\}^2
=
\{M_i-g_i\}^2
+
pq\{Y_i(1)-Y_i(0)\}^2.
\label{eq:unit-observable-loss-proof}
\end{equation}
Averaging \eqref{eq:unit-observable-loss-proof} over the finite population yields
\begin{align*}
\mathcal L_p(g)
&=
q\frac1N\sum_{i=1}^N
\{Y_i(1)-g_i\}^2
+
p\frac1N\sum_{i=1}^N
\{Y_i(0)-g_i\}^2\\
&=
\frac1N\sum_{i=1}^N
\{M_i-g_i\}^2
+
pq\frac1N\sum_{i=1}^N
\{Y_i(1)-Y_i(0)\}^2.
\end{align*}

Now, because the second term does not depend on \(g\),
\[
\argmin_{g\in\mathcal G}\mathcal L_p(g)
=
\argmin_{g\in\mathcal G}
\frac1N\sum_{i=1}^N
\{M_i-g(X_i)\}^2.
\]

\medskip
\noindent

Fix a candidate function \(g\) with respect to the treatment assignment. Since
\(
Z_iY_i=Z_iY_i(1)\) and
\(
(1-Z_i)Y_i=(1-Z_i)Y_i(0),
\)
the observable loss can be written as
\begin{align*}
\widehat{\mathcal L}_p(g)
&=
q\frac1{N_1}\sum_{i=1}^N
Z_i\{Y_i(1)-g_i\}^2\\
&\quad+
p\frac1{N_0}\sum_{i=1}^N
(1-Z_i)\{Y_i(0)-g_i\}^2.
\end{align*}
Under complete randomization,
\(
\E_Z(Z_i)=p=\frac{N_1}{N}\), and
\(
\E_Z(1-Z_i)=q=\frac{N_0}{N}.
\)
Therefore,
\begin{align*}
\E_Z\{\widehat{\mathcal L}_p(g)\}
&=
q\frac1{N_1}\sum_{i=1}^N
\E_Z(Z_i)\{Y_i(1)-g_i\}^2\\
&\quad+
p\frac1{N_0}\sum_{i=1}^N
\E_Z(1-Z_i)\{Y_i(0)-g_i\}^2\\
&=
q\frac{p}{N_1}\sum_{i=1}^N
\{Y_i(1)-g_i\}^2
+
p\frac{q}{N_0}\sum_{i=1}^N
\{Y_i(0)-g_i\}^2\\
&=
q\frac1N\sum_{i=1}^N
\{Y_i(1)-g_i\}^2
+
p\frac1N\sum_{i=1}^N
\{Y_i(0)-g_i\}^2\\
&=
\mathcal L_p(g),
\end{align*}
where the penultimate equality uses \(p/N_1=1/N\) and \(q/N_0=1/N\). 
\end{proof}


\section{Consistency and variance estimation for the joint estimator}\label{si:direct}

From semiparametric theory, \citet[Sections~2.2--2.3 and Appendix]{zhang2019machine} develop joint learning of an augmentation function, an asymptotic linearization for an estimated augmentation, and variance estimation under an i.i.d. framework. For the mean effect, their optimal function is the centered complementary-weighted score divided by $p(1-p)$. Their empirical loss uses an estimated, arm-centered influence-function response, whereas our learner uses the raw-outcome complementary-weighted loss. These criteria have the same population target after translation and scaling, but their finite-sample minimizers need not coincide. We complement these results with a finite-population theorem for the specific estimator used here. 

Consider a sequence of fixed finite populations indexed by $N$.  Partition the units into a fixed number $V\geq2$ of folds $I_1,\ldots,I_V$, independently of treatment assignment, with prespecified sizes $n_v=|I_v|$. For $i\in I_v$, let $\widehat g_i=\widehat g^{(-v)}(X_i)$, where the entire fitting procedure excludes the assignments and outcomes in $I_v$. 
The estimator is:
\begin{equation}
 \widehat\tau_{\mathrm{dir}}
 =\frac{1}{N_1}\sum_{i:Z_i=1}(Y_i-\widehat g_i)
  -\frac{1}{N_0}\sum_{i:Z_i=0}(Y_i-\widehat g_i).
 \label{eq:dir:estimator}
\end{equation}

Let $g_N^*$ be a reference score deterministic given the finite population, and abbreviate $g_{Ni}^*=g_N^*(X_i)$. The natural choice is the finite-population minimizer of the complementary weighted loss in the prespecified prediction class, although validity does not require that this class correctly specifies the potential outcomes. Define
\begin{equation}
 u_i(z)=Y_i(z)-g_{Ni}^*,\qquad
 r_i=q_Nu_i(1)+p_Nu_i(0)=M_i-g_{Ni}^*,\qquad
 M_i=q_NY_i(1)+p_NY_i(0).
 \label{eq:dir:residuals}
\end{equation}
All probability statements below condition on the finite population. When the split or the learning algorithm is randomized, probability also averages over its random seed, which is independent of treatment assignment.

\begin{assumption}[Conditions for joint-score inference]
\label{ass:dir:main}
As $N\to\infty$, the following conditions hold.
\begin{enumerate}
\item Treatment is completely randomized with $N_1$ treated units, $p_N=\frac{N_1}{N}\to p_\infty\in(0,1)$, and $n_v/N\to\rho_v\in(0,1)$ for every $v$.
The fold-specific fitting procedure is restricted to training information as described above.
\item The fitted scores converge in held-out empirical squared error to a common fixed reference score:
\begin{equation}
 Q_N:=\max_{1\leq v\leq V}
 \frac{1}{n_v}\sum_{i\in I_v}
 \{\widehat g^{(-v)}(X_i)-g_{Ni}^*\}^{2}
 =o_{\mathbb P}(1).
 \label{eq:dir:riskconsistency}
\end{equation}
\item For $z\in\{0,1\}$,
\begin{equation}
 \begin{gathered}
 \sup_N\frac1N\sum_{i=1}^N
 \{u_i(z)-\overline u(z)\}^{4}<\infty,\\
 \sigma_N^2:=\frac{S_r^2}{p_Nq_N}\longrightarrow\sigma^2\in(0,\infty),
 \qquad S_\tau^2\longrightarrow s_\tau^2<\infty.
 \end{gathered}
 \label{eq:dir:moments}
\end{equation}
\end{enumerate}
\end{assumption}

Condition~\eqref{eq:dir:riskconsistency} is a prediction-stability condition. It requires convergence to $g_N^*$, not convergence to $M$. It permits misspecification and imposes no $N^{-1/4}$ convergence rate. The fourth-moment condition is a convenient sufficient condition for the finite-population central limit theorem and consistency of residual sample variances; it can be weakened.

For estimation of uncertainty, let
\begin{equation}
 \widehat u_i=Y_i-\widehat g_i,\qquad
 \overline{\widehat u}_z
 =\frac1{N_z}\sum_{i:Z_i=z}\widehat u_i,\qquad
 \widehat s_z^2
 =\frac1{N_z-1}\sum_{i:Z_i=z}
       (\widehat u_i-\overline{\widehat u}_z)^2,
 \label{eq:dir:samplevariances}
\end{equation}
and define \(
 \widehat V_{\mathrm{dir}}
 =\frac{\widehat s_1^2}{N_1}+\frac{\widehat s_0^2}{N_0}.\)

\begin{theorem}[Consistency, asymptotic normality, and conservative inference]
\label{thm:dir:main}
Under Assumption~\ref{ass:dir:main},
\begin{equation}
 \widehat\tau_{\mathrm{dir}}-\tau_N
 =\Delta_Z(r)+o_{\mathbb P}(N^{-1/2}),\qquad
 \sqrt N(\widehat\tau_{\mathrm{dir}}-\tau_N)
 \ \xrightarrow{d}\ \mathcal N(0,\sigma^2).
 \label{eq:dir:linearization}
\end{equation}
The exact variance of the fixed-score leading term is
\begin{equation}
 V_N^*
 :=\operatorname{Var}_Z\{\Delta_Z(r)\}
 =\frac{S_r^2}{Np_Nq_N}
 =\frac{S_{u(1)}^2}{N_1}+\frac{S_{u(0)}^2}{N_0}
       -\frac{S_\tau^2}{N}.
 \label{eq:dir:oraclevariance}
\end{equation}
and,
\begin{equation}
 \left[
 \widehat\tau_{\mathrm{dir}}
       -z_{1-\alpha/2}\sqrt{\widehat V_{\mathrm{dir}}},\quad
 \widehat\tau_{\mathrm{dir}}
       +z_{1-\alpha/2}\sqrt{\widehat V_{\mathrm{dir}}}
 \right]
 \label{eq:dir:CI}
\end{equation}
has limiting randomization coverage at least $1-\alpha$, where $z_a$ is the $a$th standard normal quantile. Its limiting coverage is exactly $1-\alpha$ when $s_\tau^2=0$.
\end{theorem}

The distinction between $V_N^*$ and $\widehat V_{\mathrm{dir}}$ is important. Because $u_i(1)-u_i(0)=\tau_i$, the unidentifiable variance of individual treatment effects remains in the exact finite-population variance. Without additional restrictions, the proposed estimator consistently estimates the Neyman upper bound at the $N^{-1}$ scale. 

Write $p=p_N$ and $q=q_N$ when there is no ambiguity. The proof needs some lemmas. 

\begin{lemma}[Exact identities]
\label{lem:dir:identities}
For any fixed vector $A$,
\begin{equation}
 \Delta_Z(A)=\frac1{Npq}\sum_{i=1}^N(Z_i-p)A_i,
 \qquad
 \operatorname{Var}_Z\{\Delta_Z(A)\}=\frac{S_A^2}{Npq}.
 \label{eq:dir:imbalanceidentity}
\end{equation}
For the fixed reference score in \eqref{eq:dir:residuals},
\begin{equation}
 \widehat\tau(g_N^*)-\tau_N=\Delta_Z(r),
 \qquad
 \frac{S_r^2}{Npq}
 =\frac{S_{u(1)}^2}{N_1}+\frac{S_{u(0)}^2}{N_0}
       -\frac{S_\tau^2}{N}.
 \label{eq:dir:residualvarianceidentity}
\end{equation}
For any fitted numerical score vector, whether or not it is cross-fitted, put $e_i=\widehat g_i-g_{Ni}^*$. Then the following algebraic identity holds for every realized assignment:
\begin{equation}
 \widehat\tau_{\mathrm{dir}}-\tau_N
 =\Delta_Z(r)-\Delta_Z(e).
 \label{eq:dir:exactremainder}
\end{equation}
\end{lemma}

\begin{proof}
The first identity follows from $N_1=Np$ and $N_0=Nq$. Under complete randomization,
\[
 \mathbb E_Z Z_i=p,\qquad
 \operatorname{Var}_Z(Z_i)=pq,\qquad
 \operatorname{Cov}_Z(Z_i,Z_j)=-\frac{pq}{N-1}\quad(i\ne j).
\]
Consequently,
\begin{align*}
 \operatorname{Var}_Z\!\left\{\sum_i(Z_i-p)A_i\right\}
 &=pq\sum_i A_i^2-\frac{pq}{N-1}\sum_{i\ne j}A_iA_j\\
 &=Npq\,S_A^2.
\end{align*}
Dividing by $(Npq)^2$ proves the variance identity. Our main text shows
$\widehat\tau_{\mathrm{DM}}-\tau_N=\Delta_Z(M)$ gives $\widehat\tau(g_N^*)-\tau_N=\Delta_Z(M-g_N^*)=\Delta_Z(r)$. Let $S_{u(1),u(0)}$ denote the finite-population covariance. Because $r=qu(1)+pu(0)$ and $\tau=u(1)-u(0)$,
\begin{align*}
 \frac{S_r^2}{pq}
 &=\frac q p S_{u(1)}^2+\frac p q S_{u(0)}^2
                         +2S_{u(1),u(0)}\\
 &=\frac{S_{u(1)}^2}{p}+\frac{S_{u(0)}^2}{q}
       -\{S_{u(1)}^2+S_{u(0)}^2-2S_{u(1),u(0)}\}\\
 &=\frac{S_{u(1)}^2}{p}+\frac{S_{u(0)}^2}{q}-S_\tau^2.
\end{align*}
Finally, subtracting $\Delta_Z(\widehat g)$ instead of $\Delta_Z(g_N^*)$ yields \eqref{eq:dir:exactremainder}. This last step is algebra and therefore does not require independence of $\widehat g$ and $Z$.
\end{proof}

\begin{proposition}
\label{prop:dir:consistencyonly}
Suppose $p_N\to p_\infty\in(0,1)$, $S_r^2=O(1)$, and $N^{-1}\sum_i e_i^2=o_{\mathbb P}(1)$. Then $\widehat\tau_{\mathrm{dir}}-\tau_N=o_{\mathbb P}(1)$.
\end{proposition}

\begin{proof}
Lemma~\ref{lem:dir:identities} and Chebyshev's inequality imply $\Delta_Z(r)=O_{\mathbb P}(N^{-1/2})$. Since $\sum_i(Z_i-p)^2=Npq$ under complete randomization, Cauchy--Schwarz gives
\begin{equation}
 |\Delta_Z(e)|^2
 \leq \frac1{pq}\frac1N\sum_i e_i^2=o_{\mathbb P}(1).
 \label{eq:dir:consistencybound}
\end{equation}
Apply \eqref{eq:dir:exactremainder}.
\end{proof}

\begin{lemma}[Negligible score-estimation imbalance]
\label{lem:dir:crossfit}
Under conditions 1 and 2 of Assumption~\ref{ass:dir:main},
\begin{equation}
 \Delta_Z(e)=o_{\mathbb P}(N^{-1/2}).
 \label{eq:dir:remainderbound}
\end{equation}
\end{lemma}

\begin{proof}
Fix a fold $v$. Let $\mathcal H_v$ contain the fold partition, all pretreatment covariates, independent algorithmic random seeds, and all assignments and observed outcomes outside $I_v$. Conditional on $\mathcal H_v$, each error $e_{vi}=\widehat g^{(-v)}(X_i)-g_{Ni}^*$, $i\in I_v$, is fixed. Furthermore, the number of treated units in the held-out fold is known:
\[
 m_v=N_1-\sum_{i\notin I_v}Z_i,\qquad p_v=m_v/n_v.
\]
The conditional assignment in $I_v$ is uniform over subsets of size $m_v$. Writing $\overline e_v=n_v^{-1}\sum_{i\in I_v}e_{vi}$, decompose
\begin{align}
 \sum_{i\in I_v}(Z_i-p)e_{vi}
 &=\underbrace{\sum_{i\in I_v}(Z_i-p_v)e_{vi}}_{A_v}
   +\underbrace{(m_v-pn_v)\overline e_v}_{B_v}.
 \label{eq:dir:folddecomposition}
\end{align}
The first term has conditional mean zero and conditional variance
\[
 \mathbb E(A_v\mid\mathcal H_v)=0,\qquad
 \operatorname{Var}(A_v\mid\mathcal H_v)
       =n_vp_v(1-p_v)S_{e,v}^2,
\]
where
$S_{e,v}^2=(n_v-1)^{-1}\sum_{i\in I_v}(e_{vi}-\overline e_v)^2$. This formula also gives zero when $m_v\in\{0,n_v\}$. Put $Q_v=n_v^{-1}\sum_{i\in I_v}e_{vi}^2$. For all sufficiently large $N$, $S_{e,v}^2\leq2Q_v$, and conditional Chebyshev's inequality implies
\[
 \mathbb P\!\left(\left.|A_v|>\epsilon\sqrt N\,
                    \right|\mathcal H_v\right)
 \leq \min\{1,CQ_v/\epsilon^2\}.
\]
Here and below $C$ denotes a finite constant independent of $N$.
Since $Q_v=o_{\mathbb P}(1)$ and the right side is bounded, its expectation
tends to zero. Hence $A_v/\sqrt N=o_{\mathbb P}(1)$.

Conditional on the assignment-independent partition, $m_v$ is
hypergeometric, with
\[
 \mathbb E(m_v)=pn_v ,\text{and }
 \operatorname{Var}(m_v)=n_vpq\frac{N-n_v}{N-1}=O(N).
\]
Therefore $(m_v-pn_v)/\sqrt N=O_{\mathbb P}(1)$. Also,
$|\overline e_v|\leq\sqrt{Q_v}=o_{\mathbb P}(1)$, so
$B_v/\sqrt N=o_{\mathbb P}(1)$. This product argument does not require
independence between the fold count and the fitted score.
Summing over the fixed number of folds gives
\[
 \sqrt N\,\Delta_Z(e)
   =\frac1{pq}\sum_{v=1}^V\frac{A_v+B_v}{\sqrt N}
   =o_{\mathbb P}(1).
\]
\end{proof}

\paragraph{Proof of Theorem~\ref{thm:dir:main}}

\begin{proof}

Combining \eqref{eq:dir:exactremainder} and Lemma~\ref{lem:dir:crossfit} proves the first assertion in \eqref{eq:dir:linearization}. In view of $\operatorname{Var}_Z\{\Delta_Z(r)\}=S_r^2/(Npq)=O(N^{-1})$, it also proves consistency.

For Finite-population central limit theorem, set $\widetilde u_i(z)=u_i(z)-\overline u(z)$ and $\widetilde r_i=r_i-\overline r=q\widetilde u_i(1)+p\widetilde u_i(0)$. Convexity of $x\mapsto x^4$ and the moment assumption imply
\[
 \frac1N\sum_i\widetilde r_i^4
 \leq q\frac1N\sum_i\widetilde u_i(1)^4
       +p\frac1N\sum_i\widetilde u_i(0)^4=O(1).
\]
Consequently,
\[
 \max_i\widetilde r_i^2
 \leq\left(\sum_i\widetilde r_i^4\right)^{1/2}=O(N^{1/2}),
\text{and }
 \frac{\max_i\widetilde r_i^2}{\min(N_1,N_0)S_r^2}\longrightarrow0,
\]
since $S_r^2/(pq)\to\sigma^2>0$ and $p,q$ are bounded away from zero. The treated units form a simple random sample without replacement from the finite population. The finite-population central limit theorem \citep[Theorem~1]{li2017general} therefore applies to its mean of $r$. Because $\Delta_Z(r)=(\overline r_1-\overline r)/q$,
\[
 \frac{\Delta_Z(r)}{\sqrt{S_r^2/(Npq)}}\xrightarrow{d}\mathcal N(0,1).
\]
Slutsky's theorem proves the normal limit in \eqref{eq:dir:linearization}. Equation~\eqref{eq:dir:oraclevariance} follows from Lemma~\ref{lem:dir:identities}.

Now, we show variance estimation. For $z\in\{0,1\}$, let $s_{z,*}^2$ be the sample variance, among units with $Z_i=z$, of the fixed values $u_i(z)$. Define $a_i=u_i(z)-\overline u(z)$. The arm is a simple random sample of size $N_z$, so its sample mean of $a_i$ is $O_{\mathbb P}(N^{-1/2})$. The fourth-moment assumption gives
\[
 S_{a^2}^2
 =\frac1{N-1}\sum_i\left(a_i^2-\frac1N\sum_j a_j^2\right)^2
 \leq\frac1{N-1}\sum_i a_i^4=O(1).
\]
The simple-random-sample variance formula applied to $a_i^2$ then yields
\[
 \frac1{N_z}\sum_{i:Z_i=z}a_i^2
       -\frac1N\sum_i a_i^2=o_{\mathbb P}(1).
\]
Using
\[
 s_{z,*}^2=\frac{N_z}{N_z-1}
 \left\{\frac1{N_z}\sum_{i:Z_i=z}a_i^2
            -\left(\frac1{N_z}\sum_{i:Z_i=z}a_i\right)^2\right\},
\]
we conclude that $s_{z,*}^2-S_{u(z)}^2=o_{\mathbb P}(1)$.

Let $s_{e,z}^2$ be the sample variance of $e_i$ in arm $z$. Since $N^{-1}\sum_i e_i^2\leq Q_N=o_{\mathbb P}(1)$,
\[
 s_{e,z}^2\leq\frac1{N_z-1}\sum_{i:Z_i=z}e_i^2
       \leq\frac1{N_z-1}\sum_i e_i^2=o_{\mathbb P}(1).
\]
Within arm $z$, $\widehat u_i=u_i(z)-e_i$. The sample covariance inequality therefore gives
\[
 |\widehat s_z^2-s_{z,*}^2|
 \leq 2\sqrt{s_{z,*}^2s_{e,z}^2}+s_{e,z}^2=o_{\mathbb P}(1).
\]
No independence between residuals and predictions is required for this inequality. Thus $\widehat s_z^2=S_{u(z)}^2+o_{\mathbb P}(1)$, and
\begin{align*}
 N\widehat V_{\mathrm{dir}}
 &=\frac{S_{u(1)}^2}{p}+\frac{S_{u(0)}^2}{q}+o_{\mathbb P}(1)\\
 &=\frac{S_r^2}{pq}+S_\tau^2+o_{\mathbb P}(1).
\end{align*} Therefore, \(
 N\left(\widehat V_{\mathrm{dir}}-V_N^*\right)
   -S_\tau^2=o_{\mathbb P}(1),\) and \(
 N\widehat V_{\mathrm{dir}}\xrightarrow{\mathbb P}
       \sigma^2+s_\tau^2.\)

Another application of Slutsky's theorem gives
\[
 \frac{\widehat\tau_{\mathrm{dir}}-\tau_N}
        {\sqrt{\widehat V_{\mathrm{dir}}}}
 \xrightarrow{d}
 \mathcal N\!\left(0,\frac{\sigma^2}{\sigma^2+s_\tau^2}\right).
\]
The limiting coverage of \eqref{eq:dir:CI} is consequently
\[
 2\Phi\!\left(z_{1-\alpha/2}
       \sqrt{1+s_\tau^2/\sigma^2}\right)-1\ \geq\ 1-\alpha.
\]
It equals $1-\alpha$ when $s_\tau^2=0$.
\end{proof}

\begin{corollary}[Attainment of the class-oracle variance]
\label{cor:dir:oracle}
Suppose the reference score minimizes the complementary weighted finite-population loss over a class $\mathcal G_N$ that contains the constant scores and is closed under adding constants. If Assumption~\ref{ass:dir:main} holds, the joint estimator has the same
first-order distribution as the best fixed-score estimator in $\mathcal G_N$. In particular,
\begin{equation}
 V_N^*\leq\frac{S_M^2}{Np_Nq_N}
       =\operatorname{Var}_Z(\widehat\tau_{\mathrm{DM}}).
 \label{eq:dir:noharm}
\end{equation}
\end{corollary}

\begin{proof}
For every numerical score vector $g$, completing the square gives
\begin{equation}
 \mathcal L_p(g)
 =\frac1N\sum_i(M_i-g_i)^2
     +pq\frac1N\sum_i\tau_i^2.
 \label{eq:dir:lossidentity}
\end{equation}
The second term does not depend on $g$. Because the class is closed under adding constants, minimizing the first term also chooses an intercept such that $\overline{M-g}=0$. More explicitly, for every candidate $g$, $g+\overline{M-g}$ belongs to the class and has squared prediction error $(N-1)S_{M-g}^2/N$. Hence minimizing $\mathcal L_p$ is equivalent to minimizing $S_{M-g}^2$ over this class. A constant score has $S_{M-g}^2=S_M^2$, which proves \eqref{eq:dir:noharm}. The oracle first-order distribution follows from Theorem~\ref{thm:dir:main}.
\end{proof}

\section{Simulation}
\subsection{Simulation implementation details}
\label{si:implement}

This section describes the implementation of the gradient-boosted-tree estimators and their benchmarks in Section~\ref{sec:simulation}.

\paragraph{Full-sample estimators.}
The unadjusted estimator is $\widehat\tau_{\DM}=\overline Y_1-\overline Y_0$. All-$X$ OLS reports the coefficient on $Z$ from an additive regression of $Y$ on an intercept, $Z$, and all 60 full-sample-centered covariates. All-$X$ Lin additionally interacts $Z$ with every centered covariate. We use HC3 heteroskedasticity-robust standard errors for their treatment coefficients.

\paragraph{Cross-fitting and score construction.}
Joint, Arm-specific, and Pooled unweighted learning use gradient-boosted
trees with the same 60 raw covariates. The learners are not given the active
covariate set or the quadratic transformations used in the DGP.
We randomly partition the units into four equal-sized outer folds,
independently of treatment assignment and outcomes, and use the same
partition for all three approaches. For a held-out fold $I_v$, fitting
and tuning use only observations outside $I_v$.

Joint learning fits a single prediction function with observation weights
\[
w_i=\begin{cases}
q/p,&Z_i=1,\\
p/q,&Z_i=0,
\end{cases}
\qquad p=N_1/N,\quad q=1-p.
\]
These weights use the overall design proportions throughout. Arm-specific learning fits separate outcome regressions using the treated
and control observations in the outer training sample, then combines their
held-out predictions as
\[
\widehat g_{\mathrm{AS},i}
=q\widehat m_1^{(-v)}(X_i)+p\widehat m_0^{(-v)}(X_i),
\qquad i\in I_v.
\]
Pooled unweighted learning fits a single prediction function with equal
observation weights. Under balanced allocation, its criterion coincides
with that of Joint learning. We therefore use the same predictions and
report Pooled learning separately only under unequal allocation.

\paragraph{Tuning and refitting.}
Within each outer training sample, we form three inner folds by partitioning
the treated and control observations separately. Joint and Pooled learning
use these same inner folds, while each Arm-specific fit uses their
restriction to its own arm. Every fit selects from the same four candidates:
an intercept-only predictor and boosted trees with 50, 100, or 150
iterations. 

Joint learning selects the candidate with the smallest design-weighted
validation squared error. Pooled learning uses unweighted validation
error, and each Arm-specific model uses unweighted validation error within
its own arm. For each candidate, squared errors are summed across all
inner validation observations with the relevant weights and divided by
their total weight. We do not average the three fold-specific losses
equally. Exact ties favor the intercept and then the smaller number of
iterations. The selected model is refitted on the full outer training
sample for Joint and Pooled learning, or on the corresponding arm of that
sample for Arm-specific learning, before predicting the held-out fold.

We use \texttt{xgboost} version 1.7.8.1 with the squared-error objective, maximum tree depth two, and learning rate $0.05$. Histogram splits use 64 bins. The minimum child weight and the $L_2$ leaf-weight penalty are both one; the $L_1$ penalty and minimum split-loss penalty are zero. There is no row or column subsampling. For every fit, observation weights are normalized to have mean one, and the initial prediction is the weighted training-outcome mean. Within each inner fold, a single 150-iteration fit supplies the candidate predictions at 50, 100, and 150 iterations. Each fit uses a prespecified seed, with matching seeds for Joint and Pooled learning. 

\paragraph{ATE estimation and inference.}
For any held-out score vector $\widehat g$, the final estimate is
\begin{equation}
\widehat\tau(\widehat g)
=\widehat\tau_{\DM}-\Delta_Z(\widehat g)
=\frac{1}{N_1}\sum_{i:Z_i=1}(Y_i-\widehat g_i)
 -\frac{1}{N_0}\sum_{i:Z_i=0}(Y_i-\widehat g_i).
\label{eq:simulation-ml-estimator}
\end{equation}
With the Arm-specific score,
it is algebraically identical to the known-propensity AIPW estimator in
Equation~\eqref{eq:greg}. The DGP oracle uses the same adjustment formula
with $g_s^*(X)=\sqrt{R_M^2}\,g_s^\circ(X)$, removing the covariate-predictable
component of $M$ while retaining its irreducible component.

For the learned-score estimators and the oracle, the estimated variance is
\[
\widehat V
=\frac{s_{Y-\widehat g,1}^{2}}{N_1}
 +\frac{s_{Y-\widehat g,0}^{2}}{N_0},
\]
where each sample variance is calculated within the full treatment or control arm. Difference in means uses the same formula with $\widehat g=0$. For every estimator, the nominal 95\% confidence interval is its estimate plus or minus 1.96 times its standard error. Coverage is evaluated against
the fixed finite-population ATE. Under the conditions of
Theorem~\ref{thm:dir:main}, the residual-based intervals are asymptotically conservative when treatment effects are heterogeneous.

\subsection{Linear LASSO Simulation}
\label{app:lasso-simulation}

The linear LASSO comparison uses the four linear scenarios described in Section~\ref{subsec:simulation-design}, with $N=400$, $K=60$, and 1,000 independent treatment assignments per scenario. The linear science tables are the same fixed tables used by the main boosting study. We compare difference in means, All-$X$ OLS,
All-$X$ Lin, One-step LASSO, Joint $M$-LASSO, Arm-specific LASSO/AIPW, and the DGP oracle. One-step LASSO penalizes covariate coefficients while leaving the treatment coefficient unpenalized. To be specific, one-step LASSO is computed from the full analysis sample by solving:
\begin{equation}
(\widehat\alpha_\lambda,\widehat\tau_{\mathrm{1S}},
\widehat\beta_\lambda)
\in
\arg\min_{\alpha,t,\beta}
\left[
\frac{1}{2N}\sum_{i=1}^N
\{Y_i-\alpha-tZ_i-\widetilde X_i^{\mathsf T}\beta\}^2
+\lambda\lVert\beta\rVert_1
\right].
\label{eq:simulation-one-step}
\end{equation}
Here $\widetilde X_{ij}=(X_{ij}-\overline X_j)/s_j$, where $s_j^2=N^{-1}\sum_i(X_{ij}-\overline X_j)^2$. The intercept and treatment coefficient are unpenalized, whereas all 60 standardized covariate coefficients receive the LASSO penalty. We choose $\lambda$ by five-fold cross-validation with folds stratified by treatment status, using the minimum-CV-error rule over a 40-value penalty path. Joint and Arm-specific LASSO use four fixed outer folds and three inner folds, selecting the minimum validation error over a path of up to 40 penalties.

Table~\ref{tab:app-lasso-ate-results} and Figure~\ref{fig:app-lasso-performance} first report baseline ATE performance. For each estimator $\ell$, we report bias, root MSE, and MSE relative to the difference in means:
\begin{equation}
\operatorname{RelativeMSE}_{\ell}
=
\frac{\sum_{r=1}^R(\widehat\tau_{\ell r}-\tau_N)^2}
{\sum_{r=1}^R(\widehat\tau_{\DM,r}-\tau_N)^2}.
\label{eq:lasso-relative-mse}
\end{equation}
A value below one indicates an MSE improvement over no adjustment. We report Monte Carlo standard errors (MCSE) and 95\% Monte Carlo intervals for relative MSE.

The null and sparse-signal scenarios first show how regularization performs under balanced allocation. In the null scenario, One-step, Joint, and Arm-specific LASSO have relative MSEs of $1.002$, $1.014$, and $1.011$, respectively. All-$X$ OLS and Lin increase MSE by $19.6\%$ and $24.3\%$. Regularization therefore limits the cost of adjustment in this fixed population, although the learned-score procedures still incur a small cost.
When the useful signal is sparse, One-step LASSO attains relative MSE $0.519$, compared with $0.530$ for Joint LASSO, $0.540$ for Arm-specific LASSO/AIPW, and $0.516$ for the oracle.

\begin{table}[!htbp]
\centering
\small
\renewcommand{\arraystretch}{1.08}
\caption{ATE performance in the linear LASSO simulation}
\label{tab:app-lasso-ate-results}
\resizebox{\textwidth}{!}{%
\begin{tabular}{llrrr}
\hline
Scenario & Method & Bias (MCSE) & RMSE & Rel. MSE (MCSE) [95\% MCI] \\
\hline
Null & DM & 0.0040 (0.0030) & 0.0940 & 1.000 (0.000) [1.000, 1.000] \\
Null & All-X OLS & 0.0027 (0.0033) & 0.1028 & 1.196 (0.027) [1.142, 1.250] \\
Null & All-X Lin & 0.0039 (0.0033) & 0.1048 & 1.243 (0.032) [1.180, 1.305] \\
Null & One-step LASSO & 0.0040 (0.0030) & 0.0941 & 1.002 (0.001) [1.001, 1.004] \\
Null & Joint $M$-LASSO & 0.0044 (0.0030) & 0.0946 & 1.014 (0.003) [1.008, 1.019] \\
Null & Arm-specific LASSO (AIPW) & 0.0036 (0.0030) & 0.0945 & 1.011 (0.003) [1.005, 1.017] \\
Null & DGP oracle & 0.0040 (0.0030) & 0.0940 & 1.000 (0.000) [1.000, 1.000] \\
\hline
Sparse M signal & DM & 0.0012 (0.0030) & 0.0940 & 1.000 (0.000) [1.000, 1.000] \\
Sparse M signal & All-X OLS & -0.0010 (0.0024) & 0.0743 & 0.624 (0.030) [0.564, 0.684] \\
Sparse M signal & All-X Lin & -0.0014 (0.0024) & 0.0748 & 0.633 (0.031) [0.572, 0.693] \\
Sparse M signal & One-step LASSO & -0.0006 (0.0021) & 0.0677 & 0.519 (0.020) [0.480, 0.557] \\
Sparse M signal & Joint $M$-LASSO & 0.0001 (0.0022) & 0.0685 & 0.530 (0.020) [0.492, 0.569] \\
Sparse M signal & Arm-specific LASSO (AIPW) & -0.0004 (0.0022) & 0.0691 & 0.540 (0.018) [0.504, 0.576] \\
Sparse M signal & DGP oracle & -0.0010 (0.0021) & 0.0676 & 0.516 (0.023) [0.471, 0.562] \\
\hline
Arm-specific cancellation & DM & -0.0001 (0.0033) & 0.1042 & 1.000 (0.000) [1.000, 1.000] \\
Arm-specific cancellation & All-X OLS & -0.0003 (0.0036) & 0.1146 & 1.209 (0.042) [1.126, 1.291] \\
Arm-specific cancellation & All-X Lin & -0.0018 (0.0034) & 0.1082 & 1.078 (0.032) [1.016, 1.141] \\
Arm-specific cancellation & One-step LASSO & 0.0005 (0.0033) & 0.1048 & 1.010 (0.007) [0.996, 1.024] \\
Arm-specific cancellation & Joint $M$-LASSO & 0.0001 (0.0033) & 0.1044 & 1.004 (0.005) [0.993, 1.014] \\
Arm-specific cancellation & Arm-specific LASSO (AIPW) & 0.0000 (0.0033) & 0.1036 & 0.988 (0.011) [0.967, 1.010] \\
Arm-specific cancellation & DGP oracle & -0.0001 (0.0033) & 0.1042 & 1.000 (0.000) [1.000, 1.000] \\
\hline
Unequal allocation with HTE & DM & 0.0017 (0.0036) & 0.1132 & 1.000 (0.000) [1.000, 1.000] \\
Unequal allocation with HTE & All-X OLS & 0.0024 (0.0034) & 0.1066 & 0.887 (0.044) [0.800, 0.973] \\
Unequal allocation with HTE & All-X Lin & 0.0040 (0.0037) & 0.1169 & 1.065 (0.056) [0.955, 1.176] \\
Unequal allocation with HTE & One-step LASSO & 0.0035 (0.0030) & 0.0933 & 0.679 (0.029) [0.623, 0.736] \\
Unequal allocation with HTE & Joint $M$-LASSO & 0.0031 (0.0027) & 0.0861 & 0.579 (0.020) [0.540, 0.618] \\
Unequal allocation with HTE & Arm-specific LASSO (AIPW) & 0.0030 (0.0028) & 0.0875 & 0.598 (0.021) [0.557, 0.639] \\
Unequal allocation with HTE & DGP oracle & 0.0036 (0.0026) & 0.0808 & 0.509 (0.022) [0.466, 0.552] \\
\hline
\end{tabular}

}
\begin{minipage}{0.96\textwidth}
\footnotesize
\emph{Notes:} Bias is calculated relative to the realized finite-population ATE $\tau_N$. RMSE is on the scale of the outcome. Each relative-MSE entry uses the difference in means as its within-scenario benchmark; parentheses contain
paired Monte Carlo standard errors (MCSE), and brackets contain 95\% Monte Carlo intervals.
\end{minipage}
\end{table}

\begin{figure}[!htbp]
\centering
\includegraphics[width=0.96\textwidth]{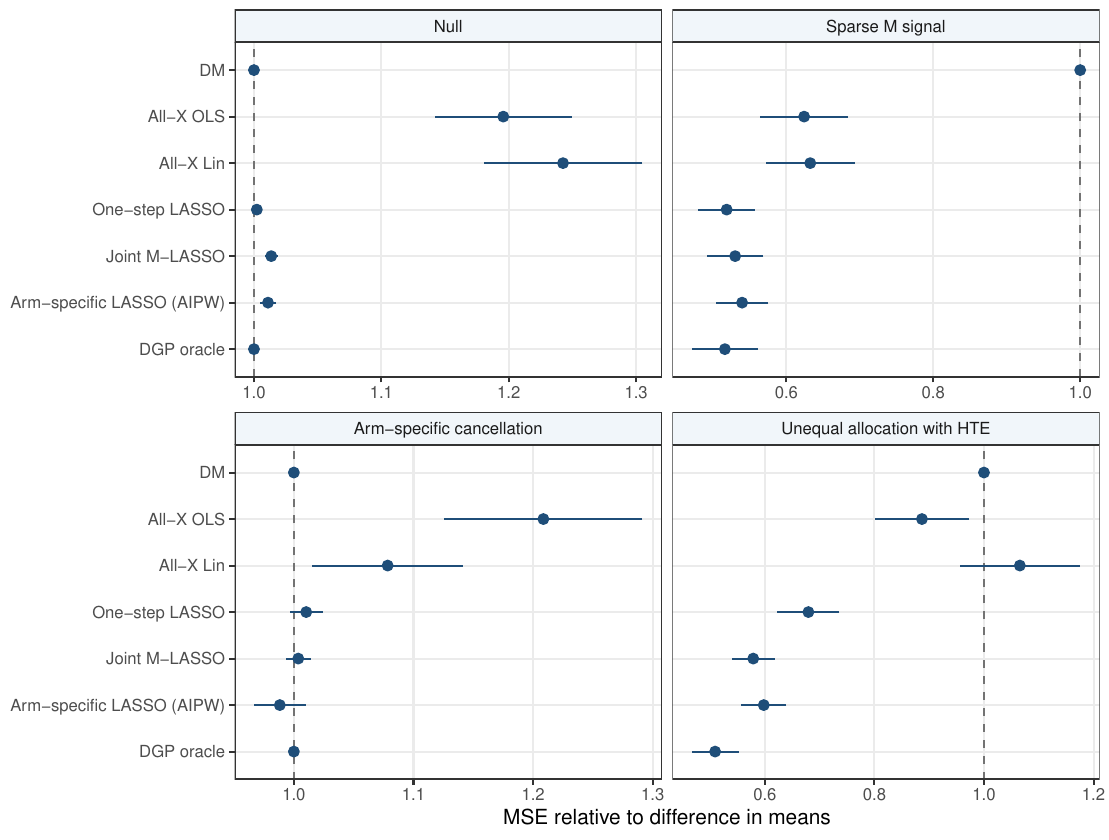}
\caption{MSE relative to the difference in means}
\label{fig:app-lasso-performance}
\begin{minipage}{0.94\textwidth}
\footnotesize
\emph{Notes:} Points show relative MSE, and horizontal intervals equal the estimate plus or minus $1.96$ paired Monte Carlo standard errors. The dashed vertical line marks the MSE of the unadjusted difference in means.
\end{minipage}
\end{figure}

Arm-specific cancellation produces mixed results across adjustment methods. All-$X$ OLS and Lin have relative MSEs of $1.209$ and $1.078$, respectively. One-step and Joint LASSO remain close to no adjustment, at $1.010$ and $1.004$. Arm-specific LASSO/AIPW attains $0.988$, with a 95\% Monte Carlo interval of $[0.967,1.010]$.

Unequal allocation makes the choice of prediction target consequential. Joint LASSO attains relative MSE $0.579$, and Arm-specific LASSO/AIPW attains $0.598$, compared with $0.679$ for conventional One-step LASSO. The oracle value is $0.509$, while All-$X$ OLS and Lin are farther away at $0.887$ and $1.065$. This pattern illustrates the value of complementary design weighting when allocation is unequal.

The baseline results at $N=400$ leave open the question of whether the relative performance of the estimators changes in smaller experiments. We therefore hold $K=60$ and $p=1/2$ fixed and repeat the sparse-signal and arm-specific-cancellation scenarios at $N\in\{160,240,400\}$. All seven estimators retain the fitting rules described above. The smallest sample is deliberately demanding: each arm contains 80 observations, so All-$X$ Lin estimates 60 slopes and an intercept with only 19 residual degrees of freedom per arm. The sparse-signal scenario compares how well the methods convert useful covariate information into precision, whereas cancellation compares their behavior when covariates predict outcomes within arms but the oracle coincides with no adjustment.

Figure~\ref{fig:simulation-sample-size-performance} displays the normalized oracle gap defined in Equation~\eqref{eq:simulation-oracle-gap}, and Table~\ref{tab:app-small-n-results} reports exact results for all seven estimators. To compare performance across sample sizes, we summarize each method's distance from the known-score benchmark using the normalized oracle gap
\begin{equation}
G_{\ell,N}
=
\frac{
\operatorname{MSE}(\widehat\tau_{\ell,N})
-\operatorname{MSE}(\widehat\tau_{\mathrm{or},N})
}{
\operatorname{MSE}(\widehat\tau_{\DM,N})
}
=
\operatorname{RelativeMSE}_{\ell,N}
-\operatorname{RelativeMSE}_{\mathrm{or},N}.
\label{eq:simulation-oracle-gap}
\end{equation}
Thus, $G_{\ell,N}$ is the MSE difference from the DGP oracle in units of difference-in-means MSE. It is zero for the oracle, and smaller values indicate performance closer to that benchmark. In the cancellation scenario, the oracle equals the difference in means, so $G_{\ell,N}$ is the proportional MSE difference between adjustment and no adjustment.

\begin{figure}[!htbp]
\centering
\includegraphics[width=0.96\textwidth]{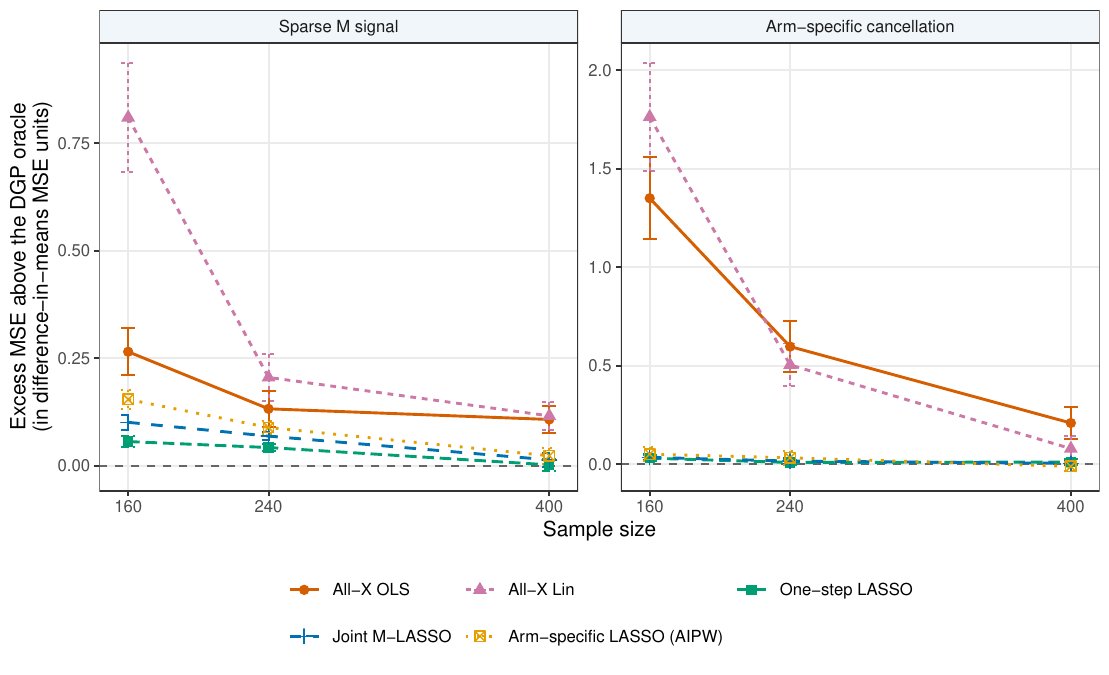}
\caption{Estimator performance relative to the DGP oracle across sample sizes}
\label{fig:simulation-sample-size-performance}
\begin{minipage}{0.94\textwidth}
\footnotesize
\emph{Notes:} Points report the normalized oracle gap $G_{\ell,N}$ from Equation~\eqref{eq:simulation-oracle-gap}; vertical intervals equal the estimate plus or minus $1.96$ paired Monte Carlo standard errors. Lower values indicate performance closer to the DGP oracle. Both panels hold $K=60$ and $p=1/2$ fixed. The DGP oracle, whose gap is zero, and the difference in means are omitted for legibility. The $N=400$ estimates reuse the baseline simulation.
\end{minipage}
\end{figure}

In the sparse-signal setting, the methods separate more sharply in the smaller fixed populations. At $N=160$, the oracle gap is $0.265$ for All-$X$ OLS and $0.810$ for All-$X$ Lin. The corresponding gaps are $0.056$ for One-step LASSO, $0.101$ for Joint $M$-LASSO, and $0.154$ for Arm-specific LASSO/AIPW. Thus, the regularized methods remain much closer to the DGP-score benchmark in the smallest sample, whereas fully interacted regression performs worse than no adjustment. At~$N=400$, the corresponding regularized-method gaps are $0.002$, $0.014$, and $0.023$.

Cancellation sharpens the contrast between unpenalized regression and regularized adjustment in the smaller populations. At $N=160$, All-$X$ OLS and Lin lie $135.0\%$ and $176.3\%$ above the oracle, respectively. One-step, Joint, and Arm-specific LASSO lie $3.0\%$, $3.4\%$, and $5.1\%$ above it. At $N=240$, the corresponding regularized-method costs are $0.8\%$, $1.5\%$, and $3.2\%$. At $N=400$, all three remain close to no adjustment, and Arm-specific LASSO is slightly below it in estimated MSE. These results show a large cost of fitting many unpenalized coefficients in the smaller fixed populations.

\subsection{Covariate Selection Using an Independent Sample}
\label{app:selection-simulation}

To avoid additional inferential complications from selecting covariates in the target experiment, we assume that an independent sample is available for covariate selection and external-score fitting. This supplementary study asks which criterion should determine whether a covariate enters adjustment. In particular, it examines whether selecting predictors of the control outcome is sufficient for selecting predictors of ATE estimation error. We focus on LASSO-based estimators, which perform covariate selection automatically.

\subsubsection{Design and covariate-set rules}

For each scenario, we generate a selection science table and an independent target science table once, with
\begin{equation}
N_{\mathrm{sel}}=N_{\mathrm{target}}=400, K=20.
\label{eq:app-selection-sample}
\end{equation}
 In each replication, we independently completely randomize treatment within the selection and target populations, using the same allocation fraction and structural scenario from Section~\ref{subsec:simulation-design}. The first six covariates retain the roles specified there; the remaining covariates are irrelevant. Because normalization is fixed by $\Sigma$, coefficients fitted in the selection sample and covariates in the target sample are expressed in the same units. Outcome-dependent selection and external-score fitting use only outcomes from the selection sample. The raw-balance rule uses target assignments and covariates but never target outcomes.

We compare five covariate sets. The control-outcome set $\widehat{\mathcal J}_0$ is the support selected by a LASSO of $Y$ on $X$ among controls in the selection sample. The raw-imbalance set is
\begin{equation}
\widehat{\mathcal J}_{\mathrm{imb}}(Z_T)
=\{j:p_{\mathrm{bal},j}(Z_T)<0.10\},
\label{eq:app-selection-imbalance}
\end{equation}
where $p_{\mathrm{bal},j}$ is the two-sided standardized-mean-difference test for target covariate $j$. We also use the intersection and union
\begin{equation}
\widehat{\mathcal J}_{\cap}
=\widehat{\mathcal J}_0\cap\widehat{\mathcal J}_{\mathrm{imb}},
\qquad
\widehat{\mathcal J}_{\cup}
=\widehat{\mathcal J}_0\cup\widehat{\mathcal J}_{\mathrm{imb}}.
\label{eq:app-selection-intersection-union}
\end{equation}
Finally, the design-relevant set $\widehat{\mathcal J}_M$ is selected from
both treatment arms in the selection sample by the complementary weighted LASSO loss
\begin{equation}
\frac{1}{N_{\mathrm{sel}}}
\sum_{i=1}^{N_{\mathrm{sel}}}
w_i^S\{Y_i^S-\alpha-(\widetilde X_i^S)^{\mathsf T}\beta\}^2
+\lambda_M\lVert\beta\rVert_1,
\qquad
w_i^S=
\begin{cases}
q/p,&Z_i^S=1,\\
p/q,&Z_i^S=0.
\end{cases}
\label{eq:app-selection-weighted-lasso}
\end{equation}

Here $\widetilde X_i^S$ denotes covariates centered and standardized within the fitting sample using its observation weights. To isolate the effect of selecting different sets, every nonempty set $\mathcal J$ receives the same unpenalized complementary-weighted refit in the selection sample. If $\widehat\gamma_{\mathcal J}$ denotes the resulting coefficient, the primary estimator applies the external score to the target experiment:
\begin{equation}
\widehat\tau_{\mathrm{ext}}(\mathcal J)
=\widehat\tau_{\DM}^{T}
-\widehat\gamma_{\mathcal J}^{\mathsf T}
\Delta_{Z_T}(X_{\mathcal J}^{T}).
\label{eq:app-selection-external}
\end{equation}
An empty set produces no adjustment. As a secondary practical comparison, we apply target-sample Lin regression to the same selected set:
\begin{equation}
\widehat\tau_{\lin}(\mathcal J)
=\widehat\tau_{\DM}^{T}
-\{q\widehat\beta_{1,\mathcal J}^{T}
+p\widehat\beta_{0,\mathcal J}^{T}\}^{\mathsf T}
\Delta_{Z_T}(X_{\mathcal J}^{T}).
\label{eq:app-selection-lin}
\end{equation}
For the external score, we also report centered prediction risk in the fixed target population:
\begin{equation}
\operatorname{RelPredRisk}_{M,r}(\widehat g)
=
\frac{
\sum_{i=1}^{N}
\left[(M_{i}-\overline M)
-\{\widehat g_r(X_{i})-\overline{\widehat g}_r\}\right]^2
}{
\sum_{i=1}^{N}(M_{i}-\overline M)^2
}.
\label{eq:app-selection-risk}
\end{equation}

\subsubsection{Results}

The study uses 8,000 independent pairs of selection--target treatment assignments per scenario. Table~\ref{tab:app-selection-results} reports the complete comparison, and Figure~\ref{fig:app-selection-performance} displays relative MSE. When control-outcome and design-relevant prognosticity coincide, the two selection rules perform nearly identically. In the sparse-signal scenario, their external-score relative MSEs both round to $0.521$, close to the DGP-oracle value $0.510$.

\begin{table}[!htbp]
\centering
\small
\renewcommand{\arraystretch}{1.08}
\caption{Separate-sample covariate selection and ATE performance}
\label{tab:app-selection-results}
\resizebox{\textwidth}{!}{%
\begin{tabular}{llrrrr}
\hline
Scenario & Rule & External rel. MSE (MCSE) [95\% MCI] & Lin rel. MSE (MCSE) [95\% MCI] & $M$ risk & Size \\
\hline
Null & No adjustment & 1.000 (0.000) [1.000, 1.000] & 1.000 (0.000) [1.000, 1.000] & 1.000 & 0.00 \\
Null & Control-outcome & 1.000 (0.000) [1.000, 1.000] & 1.000 (0.000) [1.000, 1.000] & 1.000 & 0.00 \\
Null & Raw imbalance & 1.026 (0.003) [1.020, 1.032] & 1.019 (0.004) [1.012, 1.026] & 1.005 & 2.00 \\
Null & Intersection & 1.000 (0.000) [1.000, 1.000] & 1.000 (0.000) [1.000, 1.000] & 1.000 & 0.00 \\
Null & Union & 1.026 (0.003) [1.020, 1.032] & 1.019 (0.004) [1.012, 1.026] & 1.005 & 2.00 \\
Null & Design-relevant & 1.000 (0.000) [1.000, 1.000] & 1.000 (0.000) [1.000, 1.000] & 1.000 & 0.00 \\
Sparse M signal & No adjustment & 1.000 (0.000) [1.000, 1.000] & 1.000 (0.000) [1.000, 1.000] & 1.000 & 0.00 \\
Sparse M signal & Control-outcome & 0.521 (0.008) [0.506, 0.536] & 0.507 (0.008) [0.491, 0.522] & 0.521 & 4.54 \\
Sparse M signal & Raw imbalance & 0.748 (0.010) [0.728, 0.768] & 0.742 (0.010) [0.721, 0.762] & 0.932 & 2.04 \\
Sparse M signal & Intersection & 0.739 (0.010) [0.720, 0.758] & 0.734 (0.010) [0.715, 0.753] & 0.931 & 0.46 \\
Sparse M signal & Union & 0.525 (0.008) [0.510, 0.541] & 0.510 (0.008) [0.494, 0.526] & 0.522 & 6.11 \\
Sparse M signal & Design-relevant & 0.521 (0.008) [0.506, 0.536] & 0.506 (0.008) [0.491, 0.522] & 0.521 & 4.30 \\
Arm-specific cancellation & No adjustment & 1.000 (0.000) [1.000, 1.000] & 1.000 (0.000) [1.000, 1.000] & 1.000 & 0.00 \\
Arm-specific cancellation & Control-outcome & 1.014 (0.003) [1.008, 1.020] & 0.994 (0.004) [0.987, 1.002] & 1.015 & 3.22 \\
Arm-specific cancellation & Raw imbalance & 1.036 (0.004) [1.027, 1.044] & 1.034 (0.005) [1.024, 1.044] & 1.009 & 2.00 \\
Arm-specific cancellation & Intersection & 1.004 (0.002) [1.000, 1.007] & 0.997 (0.003) [0.991, 1.003] & 1.002 & 0.33 \\
Arm-specific cancellation & Union & 1.048 (0.005) [1.038, 1.058] & 1.014 (0.005) [1.005, 1.024] & 1.023 & 4.90 \\
Arm-specific cancellation & Design-relevant & 1.000 (0.000) [1.000, 1.000] & 1.000 (0.000) [1.000, 1.000] & 1.000 & 0.00 \\
Unequal allocation with HTE & No adjustment & 1.000 (0.000) [1.000, 1.000] & 1.000 (0.000) [1.000, 1.000] & 1.000 & 0.00 \\
Unequal allocation with HTE & Control-outcome & 0.547 (0.010) [0.527, 0.566] & 0.512 (0.008) [0.495, 0.528] & 0.549 & 7.30 \\
Unequal allocation with HTE & Raw imbalance & 0.812 (0.012) [0.787, 0.836] & 0.778 (0.011) [0.756, 0.800] & 0.942 & 2.01 \\
Unequal allocation with HTE & Intersection & 0.765 (0.011) [0.743, 0.787] & 0.742 (0.010) [0.722, 0.763] & 0.933 & 0.74 \\
Unequal allocation with HTE & Union & 0.564 (0.010) [0.544, 0.584] & 0.524 (0.009) [0.507, 0.541] & 0.554 & 8.57 \\
Unequal allocation with HTE & Design-relevant & 0.534 (0.010) [0.515, 0.554] & 0.499 (0.008) [0.484, 0.515] & 0.535 & 3.37 \\
\hline
\end{tabular}

}
\begin{minipage}{0.96\textwidth}
\footnotesize
\emph{Notes:} Relative MSE is calculated against no adjustment within each scenario; parentheses contain paired Monte Carlo standard errors. The $M$-risk column reports Equation~\eqref{eq:app-selection-risk}, and Size is the mean number of selected covariates.
\end{minipage}
\end{table}

\begin{figure}[!htbp]
\centering
\includegraphics[width=0.96\textwidth]{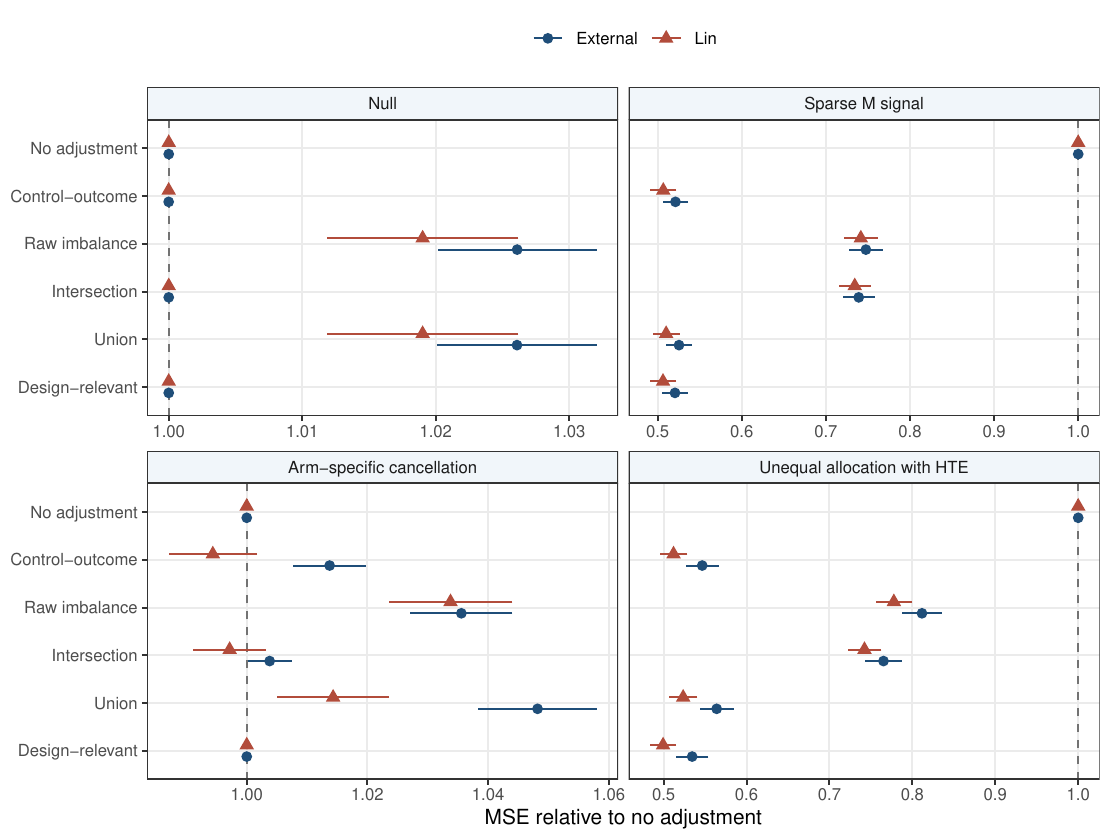}
\caption{Separate-sample covariate selection: MSE relative to no adjustment}
\label{fig:app-selection-performance}
\begin{minipage}{0.94\textwidth}
\footnotesize
\emph{Notes:} Points show relative MSE, and horizontal intervals equal the estimate plus or minus $1.96$ paired Monte Carlo standard errors. The dashed vertical line marks the MSE of the unadjusted difference in means.
\end{minipage}
\end{figure}

Arm-specific cancellation distinguishes the two prediction targets in the external-score comparison. The control-outcome rule selects 3.22 covariates on average, with average inclusion frequency $0.997$ for the heterogeneity-only covariates. Their strong outcome relationships arise from treatment-effect heterogeneity rather than a systematic component of $M$. The fitted score has centered prediction risk $1.015$ and relative MSE $1.014$. The design-relevant rule selects no covariates in these replications, has prediction risk $1.000$, and reproduces no adjustment. This illustrates why control-outcome prognosticity alone need not translate into improved ATE estimation.

Under unequal allocation, the control-outcome rule selects the three heterogeneity-only variables along with the signal variables, yielding a mean set size of $7.30$ and external-score relative MSE $0.547$. The design-relevant rule largely excludes the heterogeneity-only variables, selects 3.37 covariates on average, and lowers relative MSE to $0.534$. The oracle value is $0.488$, so both feasible external-score procedures retain an estimation cost.

The balance-based rules are less reliable in the external-score comparison. Raw-imbalance selection chooses about two variables even under the null and raises relative MSE to $1.026$. In the sparse-signal and unequal-allocation scenarios, it captures too little useful signal, yielding relative MSE $0.748$ and $0.812$. The union preserves the control-outcome predictors but adds variables, producing relative MSE $0.525$ under sparse signal, $1.048$ under cancellation, and $0.564$ under unequal allocation. The target-sample Lin panel gives broadly similar comparisons for useful signal, but not an identical ranking in every cell: under cancellation, its control-outcome rule has relative MSE $0.994$, compared with $1.000$ for the design-relevant rule. 

\subsection{Small-Sample Estimation-Cost Results}
\label{app:small-n-simulation}

\subsubsection{Boosted trees}

To examine the cost of learning an adjustment score in smaller samples, we repeat the main simulation with $N=200$, holding $K=60$ fixed. The extension includes all seven cells: the four linear specifications and the quadratic versions of the sparse-signal, arm-specific-cancellation, and unequal-allocation scenarios. The covariate distribution, signal strength, treatment-effect heterogeneity, and allocation fractions remain
those in Section~\ref{sec:simulation}. Each cell uses 1,000 replications. The $N=400$ results reuse the fixed-population main simulation. For $N=200$, we construct a separate science table for each cell and hold it fixed across treatment assignments, together with its outer folds and fitting seeds. 

The smaller sample creates one limitation for the fully interacted linear benchmark. Under unequal allocation, the treated arm contains only 50 units, whereas its regression includes an intercept and 60 covariate coefficients. The All-$X$ Lin estimator is therefore not uniquely identified in those two cells. We report it as unavailable and keep all 60 covariates for the remaining estimators. Under balanced allocation, both arms contain 100 units and the Lin benchmark remains identified.

Tables~\ref{tab:ml-small-n-linear} and \ref{tab:ml-small-n-quadratic} compare the two sample sizes, and Table~\ref{tab:ml-small-n-paired} reports paired comparisons between the learned-score estimators. The cost of fitting all 60 raw covariates is substantial in the smaller fixed populations. Under the null, All-$X$ OLS and Lin have relative MSEs of $1.545$ and $1.879$, whereas Joint and Arm-specific boosting incur smaller costs, at $1.026$ and $1.016$. In the linear sparse-signal scenario, Joint and Arm-specific boosting attain $0.726$ and $0.816$, compared with $0.903$ for OLS, $1.136$ for Lin, and $0.607$ for the oracle. Both boosting procedures improve on the regression benchmarks, although they remain above the known-score benchmark.

Cancellation continues to distinguish the two learning targets. In the linear specification, Joint and Arm-specific boosting have relative MSEs of $1.060$ and $1.271$. Their paired difference is $-0.212$, with a 95\% Monte Carlo interval of $[-0.266,-0.157]$. In the quadratic specification, the corresponding relative MSEs are $1.018$ and $1.075$, with a paired difference of $-0.057$ and interval $[-0.094,-0.021]$. Joint learning limits, but does not eliminate, the cost of fitting a score in these fixed populations constructed without a systematic covariate component in $M$.

In the quadratic sparse-signal scenario, Joint and Arm-specific boosting attain relative MSEs of $0.750$ and $0.953$ at $N=200$, compared with $0.695$ and $0.806$ at $N=400$. The oracle's relative MSE is $0.430$ at $N=200$ and $0.441$ at $N=400$. Joint's normalized oracle gap increases from $0.254$ to $0.320$; the increase is $0.066$, with MCSE $0.027$ and a 95\% Monte Carlo interval of $[0.012,0.119]$ calculated from the independent assignment runs. 

Under linear unequal allocation, Joint, Arm-specific, and Pooled boosting have relative MSEs of $0.645$, $0.709$, and $0.684$. Joint's difference from Arm-specific learning is $-0.064$, with a 95\% Monte Carlo interval of $[-0.097,-0.031]$. Its difference from Pooled learning is $-0.038$, with interval $[-0.071,-0.005]$, indicating an MSE gain from complementary weighting in this fixed population. Under quadratic unequal allocation, however, Joint, Arm-specific, and Pooled boosting attain $1.018$, $1.007$, and $0.970$. 

At $N=200$, coverage for Joint and Arm-specific boosting ranges from $94.4\%$ to $95.7\%$ in the null and sparse-signal cells, from $99.0\%$ to $99.7\%$ in the cancellation cells, and from $96.0\%$ to $96.7\%$ under unequal allocation. Null-scenario coverage is $94.4\%$ for Joint learning and $94.8\%$ for Arm-specific learning, slightly below the nominal $95\%$ level. The higher coverage in the heterogeneous-effect scenarios is consistent with the conservative residual-based variance formula.

\begin{table}[!htbp]
\centering
\small
\caption{Paired comparisons of learned-score estimators across sample sizes}
\label{tab:ml-small-n-paired}
\begin{tabular}{llrr}
\hline
\noalign{\vskip 3pt}
Scenario & Comparator & $N=200$ & $N=400$ \\[3pt]
\hline
\multicolumn{4}{l}{\textit{Linear specifications}} \\[2pt]
Null & AIPW & 0.0096 (0.0049) & 0.0004 (0.0022) \\
Sparse signal & AIPW & -0.0905 (0.0147) & 0.0064 (0.0105) \\
Cancellation & AIPW & -0.2115 (0.0278) & -0.0253 (0.0178) \\
Unequal allocation & AIPW & -0.0640 (0.0170) & -0.0205 (0.0136) \\
Unequal allocation & Pooled & -0.0381 (0.0169) & -0.0534 (0.0191) \\
\hline
\multicolumn{4}{l}{\textit{Quadratic specifications}} \\[2pt]
Sparse signal & AIPW & -0.2026 (0.0201) & -0.1113 (0.0148) \\
Cancellation & AIPW & -0.0570 (0.0186) & -0.0734 (0.0201) \\
Unequal allocation & AIPW & 0.0113 (0.0160) & -0.0455 (0.0197) \\
Unequal allocation & Pooled & 0.0482 (0.0185) & 0.0416 (0.0223) \\
\hline
\end{tabular}

\par\smallskip\begin{minipage}{0.96\textwidth}\footnotesize
\emph{Notes:} Entries subtract the comparator's MSE from Joint MSE and divide by difference-in-means MSE within each cell.
Negative values favor Joint learning. Parentheses contain paired Monte Carlo standard errors.
\end{minipage}
\end{table}

\begin{table}[!htbp]
\centering
\small
\renewcommand{\arraystretch}{1.02}
\caption{Boosting simulation across sample sizes: linear specifications}
\label{tab:ml-small-n-linear}
\resizebox{\textwidth}{!}{%
\begin{tabular}{lrrrr}
\hline
\noalign{\vskip 3pt}
Estimator & \multicolumn{2}{c}{$N=200$} & \multicolumn{2}{c}{$N=400$} \\
 & Rel. MSE (MCSE) & Coverage (MCSE) & Rel. MSE (MCSE) & Coverage (MCSE) \\[3pt]
\hline
\multicolumn{5}{l}{\textit{Null}} \\[2pt]
Difference in means & 1.000 (0.000) & 94.2 (0.7) & 1.000 (0.000) & 95.4 (0.7) \\
All-$X$ OLS & 1.545 (0.051) & 97.5 (0.5) & 1.196 (0.027) & 97.4 (0.5) \\
All-$X$ Lin & 1.879 (0.073) & 99.9 (0.1) & 1.243 (0.032) & 98.1 (0.4) \\
Joint boosted trees & 1.026 (0.005) & 94.4 (0.7) & 1.011 (0.002) & 95.4 (0.7) \\
AIPW boosted trees & 1.016 (0.006) & 94.8 (0.7) & 1.011 (0.003) & 95.6 (0.6) \\
DGP oracle & 1.000 (0.000) & 94.2 (0.7) & 1.000 (0.000) & 95.4 (0.7) \\
\hline
\multicolumn{5}{l}{\textit{Sparse signal}} \\[2pt]
Difference in means & 1.000 (0.000) & 95.3 (0.7) & 1.000 (0.000) & 95.6 (0.6) \\
All-$X$ OLS & 0.903 (0.041) & 97.8 (0.5) & 0.624 (0.030) & 97.8 (0.5) \\
All-$X$ Lin & 1.136 (0.056) & 99.8 (0.1) & 0.633 (0.031) & 98.8 (0.3) \\
Joint boosted trees & 0.726 (0.021) & 95.5 (0.7) & 0.617 (0.022) & 96.0 (0.6) \\
AIPW boosted trees & 0.816 (0.018) & 95.4 (0.7) & 0.611 (0.018) & 96.2 (0.6) \\
DGP oracle & 0.607 (0.026) & 94.7 (0.7) & 0.516 (0.023) & 95.7 (0.6) \\
\hline
\multicolumn{5}{l}{\textit{Arm-specific cancellation}} \\[2pt]
Difference in means & 1.000 (0.000) & 99.2 (0.3) & 1.000 (0.000) & 98.9 (0.3) \\
All-$X$ OLS & 1.786 (0.080) & 99.6 (0.2) & 1.209 (0.042) & 99.2 (0.3) \\
All-$X$ Lin & 1.688 (0.070) & 99.8 (0.1) & 1.078 (0.032) & 97.6 (0.5) \\
Joint boosted trees & 1.060 (0.009) & 99.1 (0.3) & 1.005 (0.006) & 98.9 (0.3) \\
AIPW boosted trees & 1.271 (0.028) & 99.0 (0.3) & 1.030 (0.017) & 98.9 (0.3) \\
DGP oracle & 1.000 (0.000) & 99.2 (0.3) & 1.000 (0.000) & 98.9 (0.3) \\
\hline
\multicolumn{5}{l}{\textit{Unequal allocation with HTE}} \\[2pt]
Difference in means & 1.000 (0.000) & 95.7 (0.6) & 1.000 (0.000) & 97.4 (0.5) \\
All-$X$ OLS & 0.881 (0.050) & 98.4 (0.4) & 0.887 (0.044) & 98.3 (0.4) \\
All-$X$ Lin & -- & -- & 1.065 (0.056) & 99.7 (0.2) \\
Joint boosted trees & 0.645 (0.025) & 96.6 (0.6) & 0.663 (0.023) & 97.8 (0.5) \\
AIPW boosted trees & 0.709 (0.023) & 96.0 (0.6) & 0.684 (0.023) & 96.9 (0.5) \\
Pooled boosted trees & 0.684 (0.027) & 96.6 (0.6) & 0.717 (0.028) & 97.4 (0.5) \\
DGP oracle & 0.456 (0.023) & 97.2 (0.5) & 0.509 (0.022) & 97.6 (0.5) \\
\hline
\end{tabular}

}
\par\smallskip\begin{minipage}{0.98\textwidth}\footnotesize
\emph{Notes:} Relative MSE uses difference in means in the same sample-size and scenario cell as its benchmark. Parentheses contain Monte Carlo standard errors; relative-MSE uncertainty uses paired replication-level errors. Coverage and its MCSE are expressed as percentages for nominal 95\% intervals; coverage can exceed 95\% under heterogeneous treatment effects.
All cells use $K=60$ and 1000 replications. The $N=400$ results are reused from the main simulation. Pooled learning coincides with Joint learning under balanced allocation and is omitted in those panels. A dash denotes an unidentified all-$X$ Lin estimator under unequal allocation at $N=200$: the treated arm contains 50 observations for 61 regression coefficients, including the intercept.
\end{minipage}
\end{table}

\begin{table}[!htbp]
\centering
\small
\renewcommand{\arraystretch}{1.02}
\caption{Boosting simulation across sample sizes: quadratic specifications}
\label{tab:ml-small-n-quadratic}
\resizebox{\textwidth}{!}{%
\begin{tabular}{lrrrr}
\hline
\noalign{\vskip 3pt}
Estimator & \multicolumn{2}{c}{$N=200$} & \multicolumn{2}{c}{$N=400$} \\
 & Rel. MSE (MCSE) & Coverage (MCSE) & Rel. MSE (MCSE) & Coverage (MCSE) \\[3pt]
\hline
\multicolumn{5}{l}{\textit{Sparse signal}} \\[2pt]
Difference in means & 1.000 (0.000) & 95.4 (0.7) & 1.000 (0.000) & 95.0 (0.7) \\
All-$X$ OLS & 1.569 (0.051) & 97.6 (0.5) & 1.061 (0.029) & 96.8 (0.6) \\
All-$X$ Lin & 2.022 (0.079) & 99.7 (0.2) & 1.114 (0.033) & 98.1 (0.4) \\
Joint boosted trees & 0.750 (0.021) & 95.5 (0.7) & 0.695 (0.019) & 95.7 (0.6) \\
AIPW boosted trees & 0.953 (0.010) & 95.7 (0.6) & 0.806 (0.014) & 95.0 (0.7) \\
DGP oracle & 0.430 (0.019) & 95.7 (0.6) & 0.441 (0.022) & 94.8 (0.7) \\
\hline
\multicolumn{5}{l}{\textit{Arm-specific cancellation}} \\[2pt]
Difference in means & 1.000 (0.000) & 99.7 (0.2) & 1.000 (0.000) & 99.5 (0.2) \\
All-$X$ OLS & 1.803 (0.080) & 99.9 (0.1) & 1.403 (0.042) & 99.3 (0.3) \\
All-$X$ Lin & 2.367 (0.117) & 100.0 (0.0) & 1.478 (0.048) & 99.5 (0.2) \\
Joint boosted trees & 1.018 (0.009) & 99.7 (0.2) & 1.040 (0.007) & 99.3 (0.3) \\
AIPW boosted trees & 1.075 (0.018) & 99.4 (0.2) & 1.113 (0.021) & 99.1 (0.3) \\
DGP oracle & 1.000 (0.000) & 99.7 (0.2) & 1.000 (0.000) & 99.5 (0.2) \\
\hline
\multicolumn{5}{l}{\textit{Unequal allocation with HTE}} \\[2pt]
Difference in means & 1.000 (0.000) & 96.6 (0.6) & 1.000 (0.000) & 96.2 (0.6) \\
All-$X$ OLS & 1.581 (0.063) & 98.9 (0.3) & 1.212 (0.039) & 97.9 (0.5) \\
All-$X$ Lin & -- & -- & 2.032 (0.091) & 99.8 (0.1) \\
Joint boosted trees & 1.018 (0.016) & 96.7 (0.6) & 0.850 (0.022) & 97.0 (0.5) \\
AIPW boosted trees & 1.007 (0.013) & 96.7 (0.6) & 0.895 (0.015) & 96.7 (0.6) \\
Pooled boosted trees & 0.970 (0.015) & 96.8 (0.6) & 0.808 (0.026) & 97.1 (0.5) \\
DGP oracle & 0.572 (0.027) & 98.1 (0.4) & 0.507 (0.021) & 97.8 (0.5) \\
\hline
\end{tabular}

}
\par\smallskip\begin{minipage}{0.98\textwidth}\footnotesize
\emph{Notes:} Relative MSE uses difference in means in the same sample-size and scenario cell as its benchmark. Parentheses contain Monte Carlo standard errors; relative-MSE uncertainty uses paired replication-level errors. Coverage and its MCSE are expressed as percentages for nominal 95\% intervals; coverage can exceed 95\% under heterogeneous treatment effects. All cells use $K=60$ and 1000 replications. The $N=400$ results are reused from the main simulation. Pooled learning coincides with Joint learning under balanced allocation and is omitted in those panels. A dash denotes an unidentified all-$X$ Lin estimator under unequal allocation at $N=200$: the treated arm contains 50 observations for 61 regression coefficients, including the intercept.
\end{minipage}
\end{table}

\subsubsection{LASSO}
Table~\ref{tab:app-small-n-results} supplements Figure~\ref{fig:simulation-sample-size-performance} by reporting the complete results for all seven estimators. The design holds $K=60$ and $p=1/2$ fixed and varies $N\in\{160,240,400\}$ in the sparse-signal and arm-specific-cancellation scenarios. The $N=400$ entries are the corresponding results from the baseline simulation.

\begin{table}[!htbp]
\centering
\scriptsize
\renewcommand{\arraystretch}{1.03}
\caption{Small-sample estimation cost relative to the DGP oracle}
\label{tab:app-small-n-results}
\resizebox{\textwidth}{!}{%
\begin{tabular}{lrllrrr}
\hline
Scenario & $N$ & Method & Bias (MCSE) & RMSE & Rel. MSE (MCSE) & Oracle gap (MCSE) [95\% MCI] \\
\hline
Sparse M signal & 160 & DM & -0.0070 (0.0056) & 0.1759 & 1.000 (0.000) & 0.535 (0.019) [0.497, 0.573] \\
Sparse M signal & 160 & All-X OLS & -0.0021 (0.0048) & 0.1504 & 0.730 (0.037) & 0.265 (0.028) [0.211, 0.319] \\
Sparse M signal & 160 & All-X Lin & -0.0027 (0.0063) & 0.1987 & 1.275 (0.072) & 0.810 (0.065) [0.683, 0.936] \\
Sparse M signal & 160 & One-step LASSO & -0.0013 (0.0040) & 0.1271 & 0.522 (0.017) & 0.056 (0.007) [0.043, 0.070] \\
Sparse M signal & 160 & Joint $M$-LASSO & -0.0005 (0.0042) & 0.1324 & 0.566 (0.018) & 0.101 (0.009) [0.084, 0.118] \\
Sparse M signal & 160 & Arm-specific LASSO (AIPW) & -0.0035 (0.0044) & 0.1385 & 0.620 (0.016) & 0.154 (0.012) [0.132, 0.177] \\
Sparse M signal & 160 & DGP oracle & -0.0007 (0.0038) & 0.1200 & 0.465 (0.019) & 0.000 (0.000) [0.000, 0.000] \\
\hline
Sparse M signal & 240 & DM & 0.0038 (0.0042) & 0.1327 & 1.000 (0.000) & 0.527 (0.020) [0.487, 0.567] \\
Sparse M signal & 240 & All-X OLS & 0.0018 (0.0033) & 0.1033 & 0.606 (0.032) & 0.133 (0.021) [0.092, 0.174] \\
Sparse M signal & 240 & All-X Lin & 0.0019 (0.0035) & 0.1093 & 0.679 (0.038) & 0.206 (0.028) [0.150, 0.261] \\
Sparse M signal & 240 & One-step LASSO & 0.0028 (0.0030) & 0.0953 & 0.516 (0.020) & 0.043 (0.004) [0.035, 0.050] \\
Sparse M signal & 240 & Joint $M$-LASSO & 0.0037 (0.0031) & 0.0977 & 0.542 (0.020) & 0.069 (0.005) [0.059, 0.079] \\
Sparse M signal & 240 & Arm-specific LASSO (AIPW) & 0.0035 (0.0031) & 0.0995 & 0.562 (0.018) & 0.089 (0.007) [0.075, 0.103] \\
Sparse M signal & 240 & DGP oracle & 0.0028 (0.0029) & 0.0913 & 0.473 (0.020) & 0.000 (0.000) [0.000, 0.000] \\
\hline
Sparse M signal & 400 & DM & 0.0012 (0.0030) & 0.0940 & 1.000 (0.000) & 0.484 (0.023) [0.438, 0.529] \\
Sparse M signal & 400 & All-X OLS & -0.0010 (0.0024) & 0.0743 & 0.624 (0.030) & 0.108 (0.016) [0.077, 0.139] \\
Sparse M signal & 400 & All-X Lin & -0.0014 (0.0024) & 0.0748 & 0.633 (0.031) & 0.116 (0.017) [0.083, 0.149] \\
Sparse M signal & 400 & One-step LASSO & -0.0006 (0.0021) & 0.0677 & 0.519 (0.020) & 0.002 (0.007) [-0.011, 0.016] \\
Sparse M signal & 400 & Joint $M$-LASSO & 0.0001 (0.0022) & 0.0685 & 0.530 (0.020) & 0.014 (0.008) [-0.002, 0.030] \\
Sparse M signal & 400 & Arm-specific LASSO (AIPW) & -0.0004 (0.0022) & 0.0691 & 0.540 (0.018) & 0.023 (0.010) [0.004, 0.042] \\
Sparse M signal & 400 & DGP oracle & -0.0010 (0.0021) & 0.0676 & 0.516 (0.023) & 0.000 (0.000) [0.000, 0.000] \\
\hline
Arm-specific cancellation & 160 & DM & 0.0051 (0.0056) & 0.1762 & 1.000 (0.000) & 0.000 (0.000) [0.000, 0.000] \\
Arm-specific cancellation & 160 & All-X OLS & 0.0120 (0.0085) & 0.2701 & 2.350 (0.106) & 1.350 (0.106) [1.142, 1.558] \\
Arm-specific cancellation & 160 & All-X Lin & 0.0067 (0.0093) & 0.2929 & 2.763 (0.140) & 1.763 (0.140) [1.489, 2.037] \\
Arm-specific cancellation & 160 & One-step LASSO & 0.0062 (0.0057) & 0.1788 & 1.030 (0.008) & 0.030 (0.008) [0.014, 0.046] \\
Arm-specific cancellation & 160 & Joint $M$-LASSO & 0.0065 (0.0057) & 0.1792 & 1.034 (0.008) & 0.034 (0.008) [0.019, 0.049] \\
Arm-specific cancellation & 160 & Arm-specific LASSO (AIPW) & 0.0041 (0.0057) & 0.1806 & 1.051 (0.018) & 0.051 (0.018) [0.015, 0.086] \\
Arm-specific cancellation & 160 & DGP oracle & 0.0051 (0.0056) & 0.1762 & 1.000 (0.000) & 0.000 (0.000) [0.000, 0.000] \\
\hline
Arm-specific cancellation & 240 & DM & -0.0047 (0.0040) & 0.1270 & 1.000 (0.000) & 0.000 (0.000) [0.000, 0.000] \\
Arm-specific cancellation & 240 & All-X OLS & -0.0058 (0.0051) & 0.1605 & 1.597 (0.066) & 0.597 (0.066) [0.468, 0.725] \\
Arm-specific cancellation & 240 & All-X Lin & -0.0066 (0.0049) & 0.1558 & 1.504 (0.055) & 0.504 (0.055) [0.395, 0.612] \\
Arm-specific cancellation & 240 & One-step LASSO & -0.0063 (0.0040) & 0.1276 & 1.008 (0.007) & 0.008 (0.007) [-0.005, 0.021] \\
Arm-specific cancellation & 240 & Joint $M$-LASSO & -0.0028 (0.0040) & 0.1280 & 1.015 (0.007) & 0.015 (0.007) [0.002, 0.028] \\
Arm-specific cancellation & 240 & Arm-specific LASSO (AIPW) & -0.0044 (0.0041) & 0.1290 & 1.032 (0.018) & 0.032 (0.018) [-0.004, 0.067] \\
Arm-specific cancellation & 240 & DGP oracle & -0.0047 (0.0040) & 0.1270 & 1.000 (0.000) & 0.000 (0.000) [0.000, 0.000] \\
\hline
Arm-specific cancellation & 400 & DM & -0.0001 (0.0033) & 0.1042 & 1.000 (0.000) & 0.000 (0.000) [0.000, 0.000] \\
Arm-specific cancellation & 400 & All-X OLS & -0.0003 (0.0036) & 0.1146 & 1.209 (0.042) & 0.209 (0.042) [0.126, 0.291] \\
Arm-specific cancellation & 400 & All-X Lin & -0.0018 (0.0034) & 0.1082 & 1.078 (0.032) & 0.078 (0.032) [0.016, 0.141] \\
Arm-specific cancellation & 400 & One-step LASSO & 0.0005 (0.0033) & 0.1048 & 1.010 (0.007) & 0.010 (0.007) [-0.004, 0.024] \\
Arm-specific cancellation & 400 & Joint $M$-LASSO & 0.0001 (0.0033) & 0.1044 & 1.004 (0.005) & 0.004 (0.005) [-0.007, 0.014] \\
Arm-specific cancellation & 400 & Arm-specific LASSO (AIPW) & 0.0000 (0.0033) & 0.1036 & 0.988 (0.011) & -0.012 (0.011) [-0.033, 0.010] \\
Arm-specific cancellation & 400 & DGP oracle & -0.0001 (0.0033) & 0.1042 & 1.000 (0.000) & 0.000 (0.000) [0.000, 0.000] \\
\hline
\end{tabular}

}
\begin{minipage}{0.98\textwidth}
\footnotesize
\emph{Notes:} Bias is calculated relative to the realized finite-population ATE. Relative MSE uses the difference in means as its within-cell benchmark. The oracle gap is in difference-in-means MSE units. Parentheses contain Monte Carlo standard errors; brackets in the final column contain untruncated 95\% Monte Carlo intervals based on paired
replication-level influence values. Each newly simulated cell contains 1,000 replications; the $N=400$ entries reuse the baseline replications.
\end{minipage}
\end{table}

\section{Application: Counterspeech on Twitter}\label{sec:twitter-application}


 We illustrate the framework using the Twitter field experiment of \citet{siegel2020no2sectarianism}, which tested whether counterspeech could reduce sectarian hate speech on Arabic-language Twitter. Conducted in January--February 2018, the researchers identified eligible accounts with at least five anti-Shia slurs during a six-month pretreatment period and a slur posted within the previous six hours. Users were individually randomized to one of five reply messages or a no-message control. The replies discouraged sectarian language and varied their appeals to shared Arab or Muslim identity and endorsements attributed to political or religious leaders. All messages came from a researcher-created account presented as an ordinary Sunni Muslim man from the Gulf. Our focal comparison is the reply combining common Muslim identity with religious-leader endorsement versus no reply. Following the original analysis, we examine the change in anti-Shia tweet count from the day before treatment to the day after treatment. We use eight pretreatment covariates: anti-Shia tweet counts and total tweet counts in four nonoverlapping intervals before treatment.\footnote{The intervals are the prior day, days 2--7, days 8--14, and days 15--30.}

 \subsection{Raw and prognostic imbalance}

 We first examine raw covariate imbalance and prognostic imbalance. Figure~\ref{fig:twitter-imbalance} compares the raw imbalance in each covariate, \(D(X_j)\), with the corresponding prognostic imbalance, \(D(\widehat g_j)\), where the prognostic score is estimated using ridge regression. For each covariate \(j\), the component difference \(\Delta_Z(\widehat g_j)\), measured in tweets, quantifies that covariate's contribution to the realized adjustment.

\begin{figure}[!htbp]
\centering
\includegraphics[width=\linewidth]{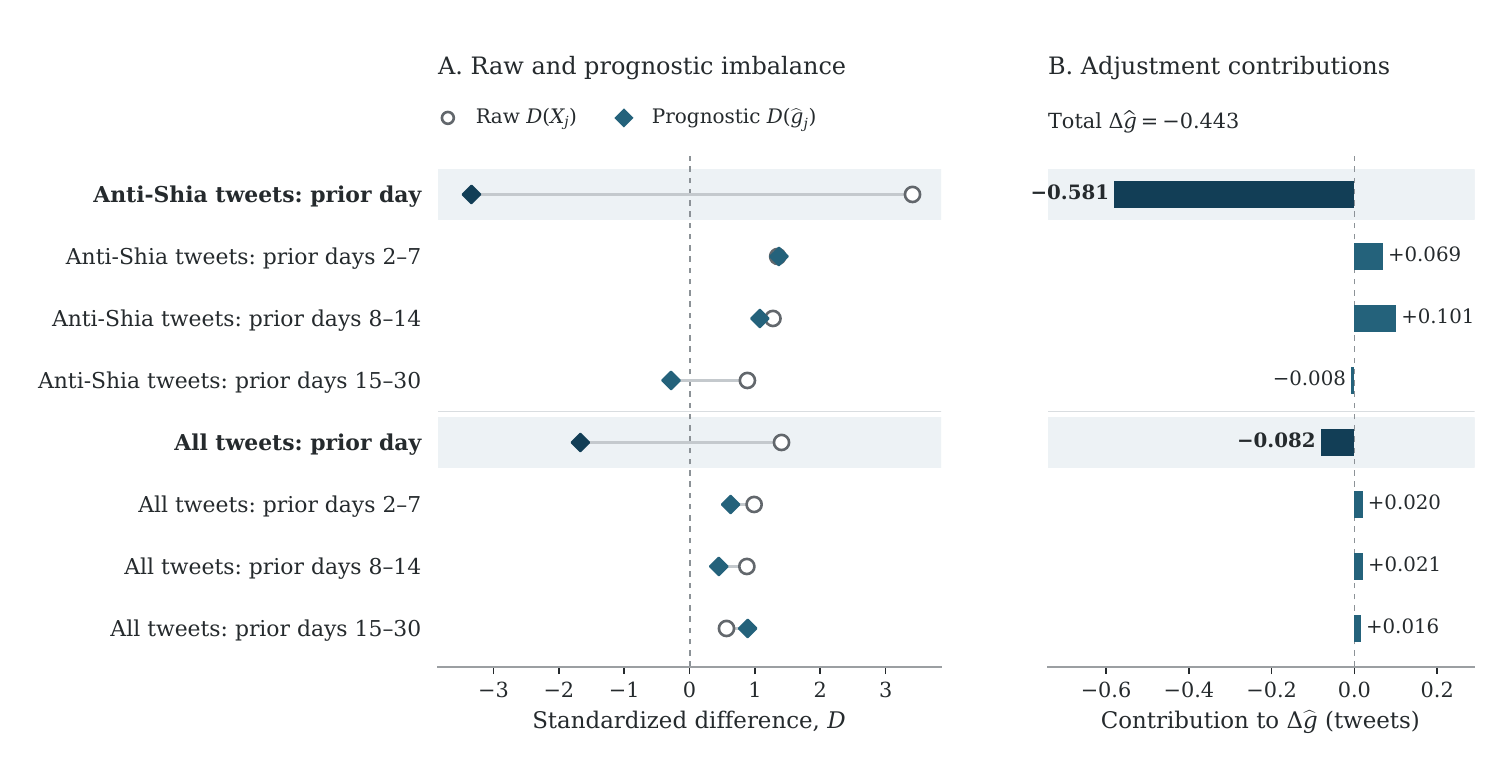}
\caption{Raw and prognostic imbalance for the first-day change outcome. Left: raw $D(X_j)$ and joint-score component $D(\widehat g_j)$. Right: component contributions to the signed correction, in tweets.}
\label{fig:twitter-imbalance}
\end{figure}

We observe that, for some covariates, raw and prognostic imbalance have opposite signs. For example, prior-day anti-Shia tweeting exhibits positive raw imbalance, $D(X_j)=3.41$, but negative prognostic imbalance, $D(\widehat g_j)=-3.34$. Treated users posted more anti-Shia tweets on the day before the intervention. Based on the raw imbalance alone, one might expect treated users to subsequently post more hate speech, making the intervention appear less effective; under this interpretation, adjusting for baseline activity should reveal a larger treatment-induced reduction. The prognostic imbalance, however, implies the opposite adjustment. Conditional on the other covariates, the fitted model associates greater prior-day anti-Shia tweeting with a larger subsequent decline in hate speech. Consequently, the positive raw imbalance translates into negative prognostic imbalance, and subtracting this negative contribution produces a smaller estimated reduction in hate speech. 

Prior-day overall tweeting exhibits a similar sign reversal, with its standardized imbalance changing from $1.41$ to $-1.67$. Prior-day anti-Shia tweeting and overall tweeting contribute -0.581 and -0.082 tweets, respectively, to the score difference. These negative contributions are partially offset by other components: anti-Shia tweeting during days 2--7 and days 8--14 before treatment contributes +0.069 and +0.101 tweets, respectively. Aggregating across all components, the jointly learned score yields an overall difference of \(\Delta_Z(\widehat g)=-0.443\) tweets and a standardized prognostic imbalance of \(D_{\widehat g}=-3.67\).

\subsection{Precision and ex-post bias correction}

We compare the difference-in-means (DM) estimator with the two adjustment estimators developed above. Both the Joint estimator and AIPW use ridge regression to estimate the nuisance functions. Both methods use all eight covariates, five-fold cross-fitting, and nested cross-validation for penalty selection.

\begin{table}[!htbp]
\centering\small
\caption{First-day counterspeech estimates and score diagnostics}
\label{tab:twitter-adjustment}
\begin{tabular}{lrrrrr}
\hline
Estimator & Estimate & SE & $\Delta_Z(\widehat g)$ & $D_{\widehat g}$ & Loss gain \\
\hline
DM & -0.948 & 0.224 & --- & --- & --- \\
AIPW & -0.593 & 0.189 & -0.354 & -4.26 & 31.1\% \\
Joint & -0.504 & 0.187 & -0.443 & -3.67 & 33.4\% \\
\hline
\end{tabular}
\begin{minipage}{\linewidth}\vspace{0.4em}\footnotesize\emph{Notes:} The outcome is the change in anti-Shia tweet count from the prior day to the following day. Loss gain is the percentage reduction in held-out weighted squared loss relative to an out-of-fold intercept-only prediction.\end{minipage}
\end{table}

The results are reported in Table~\ref{tab:twitter-adjustment} and Figure~\ref{fig:twitter-estimates}. Relative to the DM estimate of \(-0.948\), the AIPW and Joint estimators yield estimates of \(-0.593\) and \(-0.504\), respectively. Both adjustment methods therefore correct a substantial downward component in the unadjusted estimate. At the same time, the standard error decreases from \(0.224\) for DM to \(0.189\) for AIPW and \(0.187\) for the Joint estimator, corresponding to reductions of \(15.7\%\) and \(16.7\%\), respectively.

\begin{figure}[!htbp]
\centering
\includegraphics[width=0.90\linewidth]{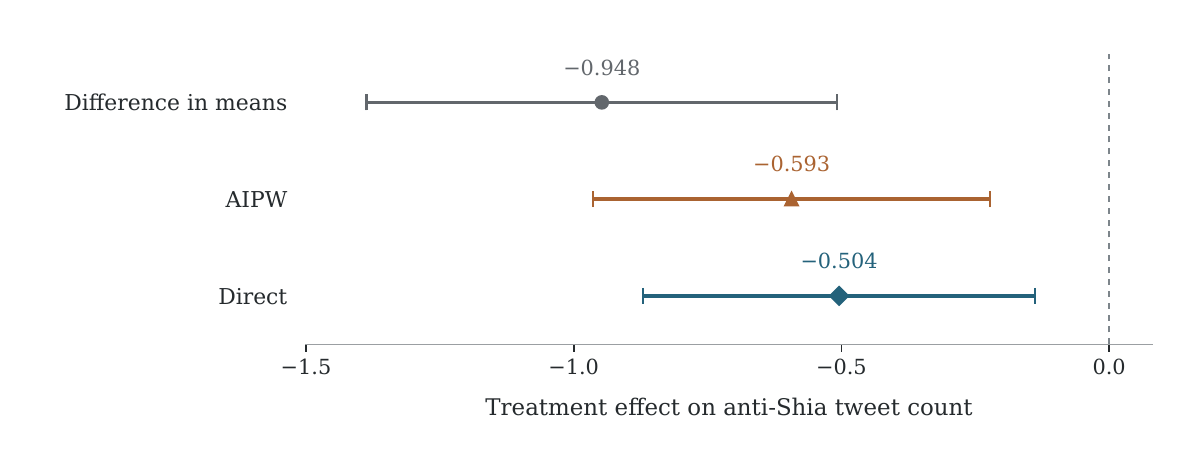}
\caption{Primary treatment-effect estimates and approximate 95\% intervals for the change in anti-Shia tweet count from the prior day to the following day. Adjustment yields smaller estimated reductions in hate speech and narrower intervals. SEs use within-arm residual variances.}
\label{fig:twitter-estimates}
\end{figure}

\end{document}